\documentclass[12pt]{article}

\usepackage[english]{babel}
\usepackage[shortlabels]{enumitem}
\usepackage[mathscr]{euscript}
\usepackage[margin=1in]{geometry}
\usepackage[small]{titlesec}
\usepackage[dvipsnames]{xcolor}

\usepackage{adjustbox, afterpage, algpseudocode, algorithm, amsbsy, amsmath, amssymb, amsthm, authblk, bm, commath, chngcntr, dsfont, econometrics, fancyhdr, gensymb, graphicx, IEEEtrantools, lscape, longtable, marginnote, mathrsfs, mathtools, mdframed, multirow, natbib, parskip, pgf, rotating, setspace, soul, subfigure, tabularx, tcolorbox, textcomp, tikz}

\usepackage[pdfauthor={Manu Navjeevan},
            bookmarks=true,%
            colorlinks=true,%
            linkcolor=pastelred,%
            citecolor=ucladarkblue,%
            pdftitle={Model Assisted Inference for Conditional Average Treatment Effects with High Dimensonal Controls},%
            pdftoolbar=false,%
            pdfmenubar=true]{hyperref} %

\usepackage{cleveref, setup_long}

\title{Doubly-Robust Inference for Conditional Average Treatment Effects with High-Dimensional Controls\thanks{We are grateful to Denis Chetverikov, Andres Santos, Zhipeng Liao, Jinyong Hahn, Rosa Matzkin, Shuyang Sheng and participants in UCLA's Econometrics Proseminar for helpful comments.}}%
\author{Adam Baybutt}
\author{Manu Navjeevan\thanks{Corresponding author: \texttt{mnavjeevan@g.ucla.edu}.} }
\affil{University of California, Los Angeles}
\date{Revised \today}

\begin{document}

\begin{doublespace}
    \maketitle
\end{doublespace}
\begin{abstract}

Plausible identification of conditional average treatment effects (CATEs) may rely on controlling for a large number of variables to account for confounding factors. In these high-dimensional settings, estimation of the CATE requires estimating first-stage models whose consistency relies on correctly specifying their parametric forms. While doubly-robust estimators of the CATE exist, inference procedures based on the second stage CATE estimator are not doubly-robust. Using the popular augmented inverse propensity weighting signal, we propose an estimator for the CATE whose resulting Wald-type confidence intervals are doubly-robust. We assume a logistic model for the propensity score and a linear model for the outcome regression, and estimate the parameters of these models using an \(\ell_1\) (Lasso) penalty to address the high dimensional covariates. Our proposed estimator remains consistent at the nonparametric rate and our proposed pointwise and uniform confidence intervals remain asymptotically valid even if one of the logistic propensity score or linear outcome regression models are misspecified. 

\end{abstract}

\newpage
\section{Introduction}%
\label{sec:introduction}

Consider a potential outcomes framework \citep{RubinPotentialOutcomes, Rubin_1978_AnnalsStat}  where an observed outcome \(Y \in \SR\) and treatment \(D \in \{0,1\}\) are related to two latent potential outcomes \(Y_1, Y_0\in\SR\) via \(Y = DY_1 + (1-D)Y_0\). To account for unobserved confounding factors a common strategy is to assume the researcher has access to a vector of covariates, \(Z = (Z_1, X) \in \calZ_1 \times \calX \subseteq \SR^{d_z-d_x,d_x}\), such that the potential outcomes are independent of the treatment decision after conditioning on the observed covariates, \((Y_1,Y_0)\perp D | Z\). In this setting, we  are interested in estimation of and inference on the conditional average treatment effect (CATE):
\begin{equation}
    \label{eq:CATE-def}
    \E[Y_1 - Y_0 \mid X = x]
.\end{equation} 
Estimation of the CATE generally requires first fitting propensity score and/or outcome regression models. When the number of control variables \(Z\) is large (\(d_z \gg n\)), these first stage models must be estimated using regularized methods which converge slower than the nonparametric rate and typically rely on the correctness of parametric specifications for consistency.\footnote{Recent works by \citet{BK-2019-neural-nets,Schmidt-Hieber-2020-neural-networks} provide some limited nonparametric results in high-dimensional settings using deep neural networks.} 

Fortunately, so long as both models are correctly specified, one can obtain a nonparametric-rate consistent estimator and valid inference procedure for the CATE by using the popular augmented inverse propensity weighted (aIPW) signal \citep{SC-2020,fan2022estimation}. This is because the aIPW signal obeys an orthogonality condition at the true nuisance model values that limits the first stage estimation error passed on to the second stage estimator. Moreover, estimators based on the aIPW signal are doubly-robust; consistency of the resulting second-stage estimators requires correct specification of only one of the first stage propensity score or outcome regression models. However inference based on these estimators is not doubly-robust. Under misspecification the aIPW signal orthogonality fails and resulting testing procedures and confidence intervals are rendered invalid. 

This paper proposes a doubly-robust estimator and inference procedure for the conditional average treatment effect when the number of control variables \(d_z\) is potentially much larger than the sample size \(n\). The dimensionality of the conditioning variable, \(d_x\), remains fixed in our analysis. Our approach is based on \citet{Tan-2018} wherein doubly-robust inference is developed for the average treatment effect. Following \citet{SC-2020} we take a series approach to estimating the CATE, using a quasi-projection of the aIPW signal onto a growing set of basis functions. By assuming a logistic form for the propensity score model and a linear form for the outcome regression model, we construct novel \(\ell_1\)-regularized first-stage estimating equations to recover a partial orthogonality of the aIPW signal at the limiting values of the first stage estimators. This restricted orthogonality is enough to achieve doubly robust pointwise and uniform inference; pointwise and uniform confidence intervals centered at the second-stage estimator are valid even if one of the logistic or linear functional forms is misspecified.

To achieve doubly-robust inference at all points in the support of the conditioning variable, we must obtain this restricted orthogonality for each basis term in the series approximation. This is accomplished by employing distinct first-stage estimating equations for each basis term used in the second-stage series approximation. This results in the number of first-stage estimators growing with the number of basis terms. These estimators converge uniformly to limiting values under standard conditions in high-dimensional analysis. Improving on prior work in doubly-robust inference, our \(\ell_1\) regularized first-stage estimation incorporates a data-dependent penalty parameter based on the work of \citet{CS-2021}. This allows practical implementation of our proposed estimation procedure with minimal knowledge of the underlying data generating process. 

The use of multiple pairs of nuisance parameter estimates limits our ability to straightforwardly apply existing nonparametric results for series estimators \citep{Newey-1997,BCCK-2015}. Under modified conditions, we analyze the asymptotic properties of our second-stage series estimator to re-derive pointwise and uniform inference results. These modified conditions are in general slightly stronger than those of \citet{BCCK-2015}, though in certain special cases collapse exactly to the conditions of \citet{BCCK-2015}.

\paragraph{Prior Literature.}
\citet{CCDDHNR-2018} analyze the general problem of estimating finite dimensional target parameters in the presence of potentially high dimensional nuisance functions. Using score functions that are Neyman-orthogonal with respect to nuisance parameters they show that it is possible to obtain target parameter estimates that are \(\sqrt{n}\)-consistent and asymptotically normal so long as the nuisance parameters are consistent at rate \(n^{-1/4}\), a condition satisfied by many machine learning-based estimators. \citet{SC-2020} take advantage of new results for series estimation in \citet{BCCK-2015} and consider series estimation of functional target parameters after high-dimensional nuisance estimation.\footnote{\citet{fan2022estimation} provides a similar analysis using a second stage kernel estimator.} 

In the same setting as this paper, \citet{Tan-2018} considers estimation of the average treatment effect. After assuming a logistic form for the propensity score and a linear form for the outcome regression, \citet{Tan-2018} proposes \(\ell_1\)-regularized first-stage estimators that allow for partial control of the derivative of the aIPW signal away from true nuisance values and thus allow for doubly-robust inference. \citet{SRR-2019} extends the analysis of \citet{Tan-2018} to consider doubly-robust inference for a larger class of finite dimensional target parameters with bilinear influence functions. 
\citet{tanCATE} provide doubly-robust inference procedures for covariate-specific treatment effects with discrete conditioning variables; their results depend on exact representation assumptions that are unlikely to hold with continuous covariates. Moreover, no uniform inference procedures are described.

\citet{CS-2021} propose a data-driven ``bootstrap after cross-validation'' approach to penalty parameter selection that is modified for and implemented in our setting. This work is related to other work on the lasso \citep{tibshirani1996regression,BRT-2009,BC-2013,CLC-2021-CVLasso} and \(\ell_1\)-regularized M-estimation in high dimensional settings \citep{vanDerGreer2016,Tan-2017}.

\paragraph{Paper Structure.} This paper proceeds as follows. Section~\ref{sec:setup} defines the problem and introduces our methods for estimation and inference. Section~\ref{sec:theory-overview} provides intuition for how the first stage estimation procedure allows for doubly-robust estimation and inference on the CATE as well as formally establishes the necessary first stage convergence. Section~\ref{sec:first-stage} presents the main results: valid pointwise and uniform inference for the second-stage series estimator if either the first-stage logistic propensity score model or linear outcome regression model is correctly specified. Section~\ref{sec:cate-wrapup} ties up a technical detail. Section~\ref{sec:simulations} provides evidence from a simulation study while Section~\ref{sec:empirical} applies our proposed estimator to examine the effect of maternal smoking on infant birth weight. Section~\ref{sec:conclusion} concludes. Proofs of main results are deferred to
the Appendix.

\paragraph{Notation.} 
For any measure \(F\) and any function \(f\), define the \(L^2\) norm,  \(\|f\|_{F, 2} = (\E_{F}[f^2])^{1/2}\) and the \(L^\infty\) norm  \(\|f\|_{F, \infty} = \esssup_F |f| \). For any vector in \(\SR^p\) let \(\|\cdot\|_p\) for \(p \in [1,\infty]\) denote the \(\ell_p\) norm, \(\|a\|_p = (\sum_{l=1}^p a_l^p)^{1/p}\) and \(\|a\|_\infty = \max_{1\leq l\leq \infty}|a_l|\). If the subscript is unspecified, we are using the \(\ell_2\) norm. For two vectors \(a,b\in\SR^p\), let \(a\circ b = (a_ib_i)_{i=1}^p\) denote the Hadamard (element-wise) product. We adopt the convention that for \(a\in\SR^p\) and \(c\in\SR\), \(a + c = (a_i + c)_{i=1}^p\). For a matrix \(A\in\SR^{m\times n}\) let \(\|A\| = \max_{\|v\|_{\ell_2} \leq 1}\|Av\|_{\ell_2}\) denote the operator norm and \(\|A\|_\infty = \sup_{{1\leq r\leq m,1\leq s \leq n}} |A_{rs}|\). For any real valued function \(f\) let  \(\E_n[f(X)] = \frac{1}{n} \sum_{i=1}^n f(X_i)\) denote the empirical expectation and \(
\mathbb{G}_n[f(X)] = \frac{1}{\sqrt{n}}\sum_{i=1}^n (f(X_i) - \E[X_i])\) denote the empirical process. For two sequences of random variables \(\{a_n\}_{\SN}\) and \(\{b_n\}_{\SN}\), we say \(a_n \lesssim_P b_n\) or \(a_n = O_p(b_n)\) if \(a_n/b_n\) is bounded in probability and say \(a_n = o_p(b_n)\) if \(a_n/b_n \to_p 0\).

\section{Setup}%
\label{sec:setup}
Below, we formally define the setting and identification strategy that we consider. We then introduce our doubly-robust estimator and inference procedure.  The parameter of interest is the conditional average treatment effect: \(\E[Y_1 - Y_0\mid X=x]\). However, for this paper we largely focus on estimation and inference for the conditional average counterfactual outcome: 
\begin{equation}
    \label{eq:target-paramter}
    g_0(x) := \E[Y_1\mid X= x]
.\end{equation} 
Doubly-robust estimation and inference on the other conditional counterfactual outcome, \(\E[Y_0\,|\,X=x]\), follows a similar procedure and is described in \Cref{sec:cate-wrapup}. The procedures can be combined for doubly-robust estimation and inference for the CATE. 
\subsection{Setting}
\label{subsec:setting}
We assume that the researcher observes i.i.d data and that conditioning on \(Z\) is sufficient to control for all confounding factors affecting both the treatment decision \(D\) and the potential outcomes, \(Y_1\) and \(Y_0\). Our analysis allows the dimensionality of the controls, \(Z = (Z_1,X)\), to grow much faster than sample size \((d_z \gg n)\), while assuming the dimensionality of the conditioning variables, \(X\), remains fixed \((d_x \ll n)\). 
\begin{assumption}[Identification]
    \label{assm:identification}
    \leavevmode
    \begin{enumerate}[(i)]
        \item \(\{Y_i,D_i,Z_i\}_{i=1}^n\) are independent and identically distributed.
        \item \((Y_1,Y_0)\perp D \mid Z\).
        \item There exists a value \(\eta \in (0,1) \) such that  \(\eta < \E[D\mid Z= z] < 1-\eta\) almost surely in  \(Z\). 
    \end{enumerate}
\end{assumption}
To obtain doubly-robust estimation and inference we use the augmented inverse propensity weighted (aIPW) signal,
\begin{equation}
    \label{eq:aIPW-signal}
    Y(\pi, m) = \frac{DY}{\pi(Z)} - \left(\frac{D}{\pi(Z)} - 1\right)m(Z),
\end{equation}
which is a function of a fitted propensity score model, \(\pi(Z),\) and a fitted outcome regression model, \(m(Z)\), whose true values are given \(\pi^\star(Z) := \E[D\mid Z]\) and \(m^\star(Z) := \E[Y\mid D= 1, Z]\).  
Under \Cref{assm:identification}, the aIPW signal \(Y(\cdot,\cdot)\) provides doubly-robust identification of \(g_0(x)\). That is, for integrable \(\pi\neq \pi^\star\) and  \(m\neq m^\star\),
\begin{equation}
    \label{eq:double-robustness}
    \begin{split}
        \E[Y_1 \mid X = x] &= \E[Y(\pi^\star, m^\star) \mid X = x] \\
                           &= \E[\,Y(\pi, m^\star)\;\mid X = x] \\ 
                           &= \E[\,Y(\pi^\star, m)\; \mid X = x].
    \end{split}
\end{equation}

We use a series approach to estimate \(g_0(x)\), taking a quasi-projection of the aIPW signal onto a growing set of \(k\) weakly positive basis terms:
\begin{equation}
    \label{eq:basis-terms}
    p^k(x) := \left(p_1(x),\dots,p_k(x)\right)' \in \SR_+^k
.\end{equation} 
The basis terms are required to be weakly positive as they are used as weights within the convex first-stage estimators estimating equations.\footnotemark 
Examples of weakly positive basis functions are B-splines or shifted polynomial series terms. To ensure that the basis terms are well behaved, we make assumptions on \(\xi_{k,\infty} := \sup_{x\in\calX}\|p^k(x)\|_\infty\), \(\xi_{k,2} := \sup_{x\in\calX}\|p^k(x)\|_2\), and the eigenvalues of the design matrix \(Q := \E[p^k(x)p^k(x)']\).

For each basis term \(p_j(x), j = 1,\dots, k\), we estimate a separate propensity score model, \(\widehat\pi_j(Z)\), and outcome regression model, \(\widehat m_j(Z)\). Under standard moment and sparsity conditions, these converge uniformly over \(j=1,\dots,k\) to limiting values \(\bar\pi_j(Z)\) and  \(\bar m_j(Z)\). If the propensity score model and outcome regression models are correctly specified these limiting values coincide with the true values \(\pi^\star(Z)\) and \(m^\star(Z)\). However, in general the limiting and true values may differ. The double robustness of the aIPW signal allows for identification of the CATE even if only one of the nuisance models is correctly specified. If either \(\bar\pi_j = \pi^\star\) or \(\bar m_j = m^\star\), we can write for all \(j=1,\dots,k\):
\begin{equation}
    \label{eq:second-stage-setup}
    \begin{split}
        Y(\bar\pi_j, \bar m_j) 
        &= g_0(x) + \eps_j,\,\;\;\;\;\;\;\;\;\;\;\;\;\;\E[\eps_j\mid X] = 0\\
        &= g_k(x) + r_k(x) + \eps_j
    \end{split}
\end{equation}
where \(g_0(x)\) is the conditional counterfactual outcome~\eqref{eq:target-paramter}, \(g_k(x) := p^k(x)'\beta^k\) is the projection of \(g_0(x)\) onto the first  \(k\) basis terms, and \(r_k (x) := g_0(x) - g_k (x) \) denotes the approximation error from this projection. Note the separate error terms for each \(j=1,\dots,k\) in \eqref{eq:second-stage-setup}, which are collected together in the vector \(\eps^k := (\eps_1,\dots,\eps_k)\).  As long as one of the first-stage models is correctly specified, the least squares parameter \( \beta^k \) governing the projection in \(g_k(x)\) can be identified by the projection of the aIPW signal onto the basis terms \(p^k(x)\):
\begin{equation}
    \label{eq:beta-k-population}
    \begin{split}
        \beta^k &:= Q^{-1}\E[p^k(X)Y_1] \\
                &\;= Q^{-1}\E[p^k(X)Y(\pi^\star, m^\star)] \\
                &\;=  Q^{-1}\E[p^k(X)Y(\bar\pi_j,\bar m_j)],\;\; \forall j = 1,\dots,k.
    \end{split}
\end{equation}
\footnotetext{\Cref{sec:first-stage-second-stage} provides a slightly modified method of constructing our doubly-robust estimator and inference procedure that does not require the first stage weights to directly be the second stage basis terms. This may be useful in case the researcher wants to use a second stage basis that cannot be transformed to be weakly positive.}
\subsection{Estimator and Inference Procedure}
\label{subsec:estimator-inference}
We assume a logistic regression form for the propensity score model and a linear form for the outcome regression model:   
\begin{equation}
    \label{eq:nuisance-parameter-functional-forms}
    \begin{split}
        \pi(Z; \gamma) &= \left(1 + \exp(-\gamma'Z)\right)^{-1}, \\
        m(Z; \alpha) &= \alpha'Z.
    \end{split}
\end{equation}
For each \(j=1,\dots,k,\) the parameters of \eqref{eq:nuisance-parameter-functional-forms}, \(\gamma,\alpha \in \SR^{d_z} ,\) are estimated by
\begin{align}
    \label{eq:gamma-j-estimating-equation}
    \widehat\gamma_j &:= \arg\min_\gamma\, \E_n[p_j(X)\{De^{-\gamma'Z} + (1-D)\gamma'Z\}] + \lambda_{\gamma,j}\|\gamma\|_1, \\
    \label{eq:alpha-j-estimating-equation}
    \widehat\alpha_j &:= \arg\min_\alpha\, \E_n[p_j(X)De^{-\widehat\gamma_j'Z}(Y - \alpha'Z)^2]/2 + \lambda_{\alpha,j}\|\alpha\|_1.
\end{align}
The penalty parameters \(\lambda_{\gamma, j}\) and \(\alpha_{\gamma, j}\) are chosen via a data dependent technique described below. These first stage estimating equations are designed so that their first order conditions directly limit the bias passed on to the second-stage series estimator, as is described in \Cref{sec:theory-overview}. Under standard assumptions the parameter estimators \(\widehat\gamma_j, \widehat\alpha_j\) will converge uniformly over \(j=1,\dots,k\) to population minimizers
\begin{align}
    \label{eq:gamma-bar-j}
    \bar\gamma_j &:= \arg\min_\gamma\, \E[p_j(X)\{De^{-\gamma'Z} + (1-D)\gamma'Z\}], \\
    \label{eq:alpha-bar-j} 
    \bar\alpha_j &:= \arg\min_\alpha \E[p_j(Z)De^{-\bar\gamma_j'Z}(Y - \alpha'Z)^2].
\end{align}
which we assume are sufficiently sparse. Our first stage estimators are then \(\widehat\pi_j(Z) := \pi(Z;\widehat\gamma_j)\) and \(\widehat m_j(Z) := m(Z;\widehat\alpha_j)\) with limiting values \(\bar\pi_j(Z) := \pi(Z;\bar\gamma_j)\) and \(\bar m_j(Z) := m(Z;\bar\alpha_j)\), respectively. 

Our second stage estimator is then \(\widehat g(x) := p^k(x)'\widehat\beta^k\) where \(\widehat\beta^k\) is an estimate of the population projection parameter, \(\beta^k\), obtained by combining all \(k\) pairs of first stage estimators according to
\begin{equation}
     \label{eq:beta-k}
     \widehat\beta^k = \widehat{Q}^{-1}\E_n
     \begin{bmatrix} p_1(X)Y(\widehat\pi_1, \widehat m_1) \\ \vdots \\ p_k(X)Y(\widehat\pi_k, \widehat m_k) \end{bmatrix},
\end{equation} 
and \(\widehat Q := \E_n[p^k(X)p^k(X)']\).
We estimate the variance of \(\widehat g(x)\) using \(\widehat\sigma(x) := \|\widehat\Omega^{1/2}p^k(x)\|/\sqrt{n}\) for
\begin{equation}
    \label{eq:omega-hat-definitions}
    \widehat\Omega := \widehat Q^{-1}\E_n[\{p^k(X)\circ\widehat\eps^k\}\{p^k(X)\circ\widehat\eps^k\}' ]\widehat Q^{-1},
\end{equation}
where \(\circ\) represents the Hadamard product and \(\widehat\eps^k := (\widehat\eps_1,\dots,\widehat\eps_k)\); \(\widehat\eps_j := Y(\widehat\pi_j, \widehat m_j) - \widehat g(x)\), \(j=1,...,k\).

Inference is based on the \(100(1-\eta)\%\) confidence bands
\begin{equation}
    \label{eq:confidence-bands}
    \left[\underline i(x), \bar i(x)\right] := \left[\widehat g(x) - c^\star\left(1-\eta/2\right)\widehat\sigma(x),\; \widehat g(x) + c^\star\left(1-\eta/2\right)\widehat\sigma(x)\right]
.\end{equation} 
For pointwise inference, the critical value \(c^\star(1-\eta/2)\) is taken as the \((1-\eta/2)\) quantile of a standard normal distribution. For uniform inference \(c^\star(1-\eta/2)\) is taken  
\[
    c^\star(1-\eta/2) := (1-\eta/2)\text{-quantile of }\sup_{x\in\calX} \left|\frac{p^k(x)\widehat\Omega^{1/2}}{\widehat\sigma(x)}N_k^b \right|
\]
where \(N_k^b\) is a bootstrap draw from  \(N(0,I_k)\). \Cref{sec:theory-overview,sec:first-stage} show that, under standard sparsity and moment conditions, these pointwise and uniform inference procedures remain valid even under misspecification of either first-stage model.

\subsection{Penalty Parameter Selection}
\label{subsec:additional}

To select the penalty parameters \(\lambda_{\gamma,j}\) and \(\lambda_{\alpha,j}\) in  \eqref{eq:gamma-j-estimating-equation}-\eqref{eq:alpha-j-estimating-equation} we propose a data driven two-step procedure based on the work of \citet{CS-2021}. 
For each \(j= 0,1\dots,k,\) we start with pilot penalty parameters given by
\begin{equation}
    \label{eq:pilot-penalty}
    \lambda^{\text{\tiny pilot}}_{\gamma, j} = c_{\gamma,j}\times \sqrt{\frac{\ln^3(d_z)}{n}} \andbox \lambda_{\alpha,j}^{\text{\tiny pilot}} = c_{\alpha,j}\times \sqrt{\frac{\ln^3(d_z)}{n} }
\end{equation}
for some constants \(c_{\gamma, j}, c_{\alpha, j}\) selected from the interval \([\underline c_n, \bar c_n]\) with \(\underline c_n > 0\). In practice, the researcher has a fair bit of flexibility in choosing these constants. The optimal choice of these constants may depend on the underlying data generating process. We recommend using cross validation to pick these constants from a fixed-cardinality set of possible values. In line with \Cref{assm:logistic-model-convergence}(vi), the values in the set should be chosen to be on the order of the maximum value of \(\|p^k(X_i)\|_\infty\) observed in the data. 

Using \(\lambda^{\text{\tiny pilot}}_{\gamma, j}\) and \(\lambda^{\text{\tiny pilot}}_{\alpha,j}\) in lieu of \(\lambda_{\gamma,j}\) and  \(\lambda_{\alpha,j}\) in  \eqref{eq:gamma-j-estimating-equation}-\eqref{eq:alpha-j-estimating-equation} we generate pilot estimators \(\widehat\gamma^{\text{\tiny pilot}}_j\) and \(\widehat\alpha^{\text{\tiny pilot}}_j\). These pilot estimators are used to generate plug in estimators \(\widehat U_{\gamma, j}\) and \(\widehat U_{\alpha, j}\) of the residuals
\begin{equation}
    \label{eq:residual-estimates}
    \begin{split}
        \widehat U_{\gamma,j} &:= -p_j(X)\{De^{-\widehat\gamma_j^{\text{\tiny pilot}'}Z} + (1-D)\}  \\
        \widehat U_{\alpha,j} &:= p_j(X)De^{-\widehat\gamma_j^{\text{\tiny pilot}'}Z}(Y - \widehat\alpha_j^{\text{\tiny pilot}'}Z).
    \end{split}
\end{equation}
We then use a multiplier bootstrap procedure to select our final penalty parameters \(\lambda_{\gamma, j}\) and \(\lambda_{\alpha, j}\). 
\begin{equation}
    \label{eq:final-penalty-parameters}
    \begin{split}
    \lambda_{\gamma, j} &= c_0\times (1-\eps)\text{-quantile of} \max_{1\leq l \leq d_z} |\E_n[e_i\widehat U_{\gamma,j}Z_l]|\text{ given }\{Y_i, D_i, Z_i\}_{i=1}^n, \\ 
    \lambda_{\alpha, j} &= c_0\times(1-\eps)\text{-quantile of} \max_{1\leq l \leq d_z} |\E_n[e_i\widehat U_{\alpha,j}Z_l]|\text{ given }\{Y_i, D_i, Z_i\}_{i=1}^n  
    \end{split}
\end{equation}
where \(e_1,\dots,e_n\) are independent standard normal random variables generated independently of the data \(\{Y_i,D_i,X_i\}_{i=1}^n\) and \(c_0 > 1\) is a fixed constant.\footnote{The constant \(c_0\) can be different for the propensity score and outcome regression models and can also vary for each  \(j = 1,\dots,k\). All that matters is that each constant satisfies the requirements of \Cref{thm:first-stage-convergence}. This complicates notation, however.} In line with other work we find \(c_0 = 1.1\) works well in simulations. So long as our residual estimates converge in empirical mean square to limiting values, the choice of penalty parameter in \eqref{eq:final-penalty-parameters} will ensure that the penalty parameter dominates the noise with high probability. This allows for consistent variable selection and coefficient estimation.

For computational reasons, the researcher may not want to implement the bootstrap penalty parameter procedure. If this is the case, we note that the pilot penalty parameters of \eqref{eq:pilot-penalty} can be used directly after the constants \(c_{\gamma, j}\) and \(c_{\alpha,j}\) can be selected via cross validation from a growing set \(\Lambda_n \subseteq [\underline c_n, \bar c_n]\) under modified conditions. \Cref{sec:alt-penalty-CV} provides details for this implementation as well as formally shows the modified conditions needed.

\section{Theory Overview}%
\label{sec:theory-overview}

We begin with a main technical lemma which provides a bound on rate at which first stage estimation error is passed on to the second stage CATE and variance estimators. This bound is comparable to others seen in the inference after model-selection literature \citep{BCH-2013,Tan-2018} and is achieved under standard conditions in the \(\ell_1\)-regularized estimation literature \citep{BRT-2009,Bulmann-VanDeGeer-2011,BC-2013,CS-2021}. However, this bound is achieved at the limiting values of the propensity score and outcome regression models which may differ from the true values \(\pi^\star\) and \(m^\star\) under misspecification.

The potential misspecification of the first stage models which means we cannot directly apply orthogonality of the aIPW signal, discussed below, to show that the effect of first stage estimation error on the second stage is negligible. Instead, we use the first order conditions for \(\widehat\gamma_j\) and \(\widehat\alpha_j\) to directly control this quantity. After presenting the lemma \Cref{subsec:first-stage-intuition} provides some intuition for how this is done. Controlling the rate at which first stage estimation error is passed on to the second stage estimator even at points away from the true values \(\pi^\star\) and  \(m^\star\) is key for obtaining doubly-robust inference for the CATE.

\subsection{Uniform First-Stage Convergence}%
\label{subsec:uniform-first-stage-convergence}

To show uniform convergence of the first stage estimators and thus uniform control of the bias passed on from the first stage estimation to the second stage estimator we rely on the following assumption:

\begin{assumption}[First Stage Convergence]
    \label{assm:logistic-model-convergence}
    \leavevmode
    \begin{enumerate}[(i)]
        \item The regressors \(Z\) are bounded, \(\max_{1 \leq l \leq d_z} |Z_l| \leq C_0\) almost surely.
        \item The errors \(Y_1 - \bar m_j(Z)\) are uniformly subgaussian conditional on  \(Z\) in the following sense. There exist fixed positive constants  \(G_0\) and  \(G_1\) such that for any  \(j\):
        \[
            G_0 \E\left[\exp\big(\{Y_1 - \bar m_j(Z)\}^2/G_0^2\big) - 1\mid Z\right] \leq G_1^2
        \]
         almost surely.
        \item There is a constant \(B_0\) such that \(\bar\gamma_j'Z \geq B_0\) almost surely for all \(j\).
        \item There exist fixed constants \(\xi_0 > 1\) and  \(1> \nu_0 > 0\) such that for each  \(j = 1,\dots,k\) the following empirical compatability condition holds for the empirical hessian matrix \(\tilde\Sigma_{\gamma, j} := \E_n[De^{-\bar\gamma_j'Z}ZZ']\). For any \(b\in\SR^{d_z}\) and \(\calS_j = \{l: \bar\gamma_l\vee\bar\alpha_l\neq 0\}\):
        \[
           \sum_{l\not\in\calS_j} |b_l|\leq \xi_0\sum_{l\in\calS_j} |b_l|
           \implies \nu_0^2\Big(\sum_{l\in\calS_j}|b_l|\Big)^2 \leq |\calS| \left(b'\tilde\Sigma_{\gamma, j}b\right)
        .\] 
        \item There exist fixed constants  \(c_u\) and  \(C_U > 0\) such that for all \(j \leq k\),  \(\E[U_{\gamma,j}^4] \leq (\xi_{k,\infty}C_U)^4\) and \(\min_{1\leq l\leq d_z}\E[U_{\gamma, j}^2 Z_l^2] \geq c_u\).
        \item The constant \(\underline c_n\) is chosen such that  \(\xi_{k.\infty} \lesssim \underline c_n\) and the following sparsity bounds hold for \(s_k = \max_{1 \leq j \leq k}|\calS_j|\)
        \[
           \frac{\xi_{k,\infty}s_k^2\bar c_n^2\ln^5(d_zn)}{n} \to 0,\andbox \frac{\xi_{k,\infty}^4\ln^7(d_zkn)}{n}\to 0 
        .\] 
    \end{enumerate}
\end{assumption}

The first part of \Cref{assm:logistic-model-convergence} assumes that the regressors are bounded while the second assumes that tail behavior of the outcome regression errors are uniformly thin. Both of these can be relaxed somewhat with sufficient moment conditions on the tail behavior of the controls and errors. We should note that compactness of \(\calX\) is generally required by nonparametric estimators. The third part of the assumption bounds all limiting propensity scores \(\bar\pi_j(Z)\) away from zero uniformly. The fourth assumption is an empirical compatibility condition on the weighted first-stage design matrix. It is slightly weaker than the restricted eigenvalue conditions often assumed in the literature \citep{BRT-2009,BCCK-2012}. The penultimate condition is an identifiability constraint that limits the moments of the noise and bounds it away from zero uniformly over all estimation procedures.  Many of the constants in \Cref{assm:logistic-model-convergence} are assumed to be fixed across all \(j\). This is mainly to simplify the exposition of the results below and in practice all constants can be allowed to grow slowly with \(k\). However, the growth rate of these terms affects the required first-stage sparsity.

The last condition is required for the validity of the bootstrap penalty parameter selection procedure and is comparable to the requirements needed for the bootstrap after cross validation technique described by \citet{CS-2021}. The main difference is the additional assumption on the growth rate of the basis functions, \(\xi_{k,\infty}\) which is to ensure uniform stability of the estimation procedures \eqref{eq:gamma-j-estimating-equation}-\eqref{eq:alpha-j-estimating-equation} as well as some assumptions on the order of the constants \(c_{\gamma,j}\) and \(c_{\alpha,j}\) in \eqref{eq:pilot-penalty}.

\begin{lemma}[First-Stage Convergence]
    \label{thm:first-stage-convergence}
    Suppose that \Cref{assm:logistic-model-convergence} holds. In addition assume that \(c_0 > (\xi_0 + 1)/(\xi_0-1)\), \(k/n \to 0\), \(k\eps \to 0\), and there is a fixed constant  \(c > 0\) such that for all  \(j\),  \(\lambda_{\alpha, j}/\lambda_{\gamma,j} \geq c\).\footnotemark Then the following weighted means converge uniformly in absolute value at least at rate:
    \begin{equation}
        \label{eq:uniform-mean-convergence}
        \max_{1\leq j\leq k}\left|\E_n[p_j(X)Y(\widehat\pi_j, \widehat m_j)] - \E_n[p_j(X)Y(\bar\pi_j, \bar m_j)]\right| 
        \lesssim_P 
        \frac{s_k\,\xi_{k,\infty}^2\ln(d_z)}{n}
    \end{equation}
    and in empirical mean square at least at rate:
    \begin{equation}
        \label{eq:second-moment-convergence}
        \max_{1\leq j\leq k}\E_n[p_j^2(X)(Y(\widehat\pi_j,\widehat m_j) - Y(\bar\pi_j,\bar m_j))^2] \lesssim_P \frac{s_k^2\,\xi_{k,\infty}^4 \ln(d_z)}{n} 
    \end{equation}
\end{lemma}
\Cref{thm:first-stage-convergence} provides a tight bound on the first-stage estimation error passed on to the second stage estimator even when the first-stage estimators converge  to values that are not the true propensity score or outcome regression. In particular notice that under the (familiar) sparsity bound \(s_k\xi_{k,\infty}^2k^{1/2}\ln^2(d_z)/\sqrt{n}\to 0\), any linear combination of the means in both \eqref{eq:uniform-mean-convergence} and \eqref{eq:second-moment-convergence} is \(o_p(\sqrt{n})\).  This allows us to obtain doubly-robust inference for the CATE. 
\footnotetext{The requirement \(\lambda_{\alpha,j}/\lambda_{\gamma,j} \geq c\) may seem a bit unnatural, but it can be enforced in practice  without upsetting any assumptions by setting the linear penalty 
    \(
        \lambda_{\alpha,j}^{\text{\tiny ratio}} := \max\{\lambda_{\gamma,j}/5, \lambda_{\alpha, j}\}
    .\)
In simulations, we find this constraint is rarely binding.}

\subsection{Managing First-Stage Bias}%
\label{subsec:first-stage-intuition}
Below, we provide some intuition for how this result is obtained and the role our particular estimating equations play in establishing this fact.
We focus on control of the vector \(\vB^k\), defined in \eqref{eq:beta-means}, which measures the bias passed on from first-stage estimation to the second-stage estimate \(\widehat\beta^k\). Limiting the size of \(\vB^k\) is crucial in showing convergence of \(\widehat\beta^k\) to the true parameter \(\beta^k\) and thus consistency of the nonparametric estimator \(\widehat g(x)\). 
\begin{equation}
    \label{eq:beta-means}
    \vB^k :=\E_n\begin{bmatrix} p_1(X)\left\{Y(\widehat\pi_1, \widehat m_1) - Y(\bar\pi_1, \bar m_1)\right\} \\ \vdots \\ p_k(X)\left\{Y(\widehat \pi_k, \widehat m_k) - Y(\bar \pi_k, \bar m_k)\right\}\end{bmatrix}.
\end{equation}
For exposition, we consider a single term of \eqref{eq:beta-means}, \(\vB^k_j\), which roughly measures the first stage estimation bias taken on from adding the \(j^\text{th}\) basis term to our series approximation of  \(g_0(x)\). The discussion that follows is a bit informal, instead of considering the derivatives with respect to the true parameters below our proof strategy will directly use the Kuhn-Tucker conditions of the optimization routines in \eqref{eq:gamma-j-estimating-equation}-\eqref{eq:alpha-j-estimating-equation}. However, the general intuition is the same as is used in the proofs.

In addition to the doubly-robust identification property \eqref{eq:double-robustness}, the aIPW signal is typically useful in the high-dimensional setting because it obeys an orthogonality condition at the true values \((\pi^\star,m^\star)\):\footnote{Robustness and orthogonality are indeed closely related, see Theorem 6.2 in \citet{nmf-1994} for a discussion.} 
\begin{equation}
    \label{eq:neyman-orthogonality}
    \E[\nabla_{\pi,m} Y(\pi^\star, m^\star) \mid Z] = 0
.\end{equation} 
When both the propensity score model and outcome regression model are correctly specified we can (loosely speaking) examine the bias \(\vB_j^k\) by replacing \(\bar\pi_j = \pi^\star\) and  \(\bar m_j = m^*\) and considering the following first order expansion:
\begin{equation}
    \label{eq:taylor-expansion-true-parameters}
    \begin{split}
       \vB_j^k= 
       \E_n[p_j(X)Y(\widehat \pi_j, \widehat m_j)] 
        &- \E_n[p_j(X)Y(\pi^\star, m^\star)] \\
        &= \underbrace{\E_n[p_j(X)\nabla_{\pi, m}\,Y(\pi^\star, m^\star)]}_{O_p(n^{-1/2})\text{ by \eqref{eq:neyman-orthogonality}}}\begin{bmatrix} \widehat \pi_j - \pi^\star \\ \widehat m_j - m^\star \end{bmatrix}   + o_p(n^{-1/2}).
    \end{split} 
\end{equation}
By orthogonality of the aIPW signal the gradient term is close to zero, which guarantees that the bias is asymptotically negligible even if the nuisance parameters converge slowly to the true values, \(\pi^\star\) and  \(m^\star\).\footnote{Typically all that is required is that \(\|\hat\pi_j - \pi^\star\| = o_p(n^{-1/4})\) and \(\|\hat m_j - m^\star\| = o_p(n^{-1/4})\) in order to make the second order remainder term \(\sqrt{n}\)-negligible}  This allows the researcher to ignore first stage nuisance parameter estimation error and treat \(\pi^\star\) and \(m^\star\) as known when analyzing the asymptotic properties of the second stage series estimator. Indeed, since the aIPW signal orthogonality holds conditional on \(Z = (Z_1,X)\), if both models are correctly specified only a single pair of first stage estimators would be needed to provide control over all the elements in \(\vB^k\). This is the approach followed by \citet{SC-2020}.

So long as either one of \(\bar\pi_j = \pi^\star\) or  \(\bar m_j = m^\star\), double robustness of the aIPW signal~\eqref{eq:double-robustness} still delivers identification: \(\E[p_j(X)Y_1] \approx \E_n[p_j(X)Y(\bar\pi_j,\bar m_j)\).
However, the aIPW orthogonality tells us nothing about the expectation of the gradient away from the true parameters, \(\pi^\star, m^\star\); if either \(\bar\pi_j \neq \pi^\star\) or  \(\bar m_j \neq m^\star\) there is no reason to believe that the gradient on the right hand side of \eqref{eq:taylor-expansion-true-parameters} is mean zero when evaluated instead at \(Y(\bar\pi_j,\bar m_j)\). In general, the bias \(\vB_j^k\) will then diminish at the rate of convergence of our nuisance parameters. Because we have high dimensional controls, this convergence rate will generally be much slower than the standard nonparametric rate \citep{Newey-1997,BCCK-2015}. 

To get around this, we design the first-stage objective functions \eqref{eq:gamma-j-estimating-equation}-\eqref{eq:alpha-j-estimating-equation} such that the resulting first-order conditions control the bias passed on to the second stage. Consider the following expansion instead around the limiting parameters \(\bar\gamma_j\) and \(\bar\alpha_k\).
\begin{equation}
    \label{eq:taylor-expansion-gamma-alpha}
    \begin{split}
        \vB_j^k = \E_n[p_j(X)Y(\widehat \pi_j, \widehat m_j)] 
        &- \E_n[p_j(X)Y(\bar\pi_j, \bar m_j)] \\
        &= \E_n[p_j(X)\nabla_{\gamma_j, \alpha_j}\,Y(\bar\pi_j, \bar m_j)]
        \begin{bmatrix} \widehat \gamma_j - \bar\gamma_j \\ \widehat \alpha_j - \bar\alpha_j \end{bmatrix} + o_p(n^{-1/2})
    \end{split}
\end{equation}
After substituting the forms of \(\bar\pi_j(z) = \pi(z;\bar\gamma_j)\) and \(\bar m_j(z) = m(z;\bar\alpha_j)\) described in  \eqref{eq:nuisance-parameter-functional-forms} and differentiating with respect to \(\gamma_j\) and \(\alpha_j\) we obtain 
\begin{equation}
    \label{eq:gradient-plim-parameters}
    \E[p_j(X)\nabla_{\gamma_j,\alpha_j}\,Y(\bar\pi_j,\bar m_j)] 
    = \E\begin{bmatrix} 
    -p_j(X)De^{-\bar\gamma'Z}(Y- \bar\alpha'Z)Z \\
    p_j(x)\{D(1+e^{-\bar\gamma'Z})Z + Z\} 
    \end{bmatrix} 
\end{equation}
However, by definition \(\bar\gamma_j\) and  \(\bar\alpha_j\) solve the minimization problems defined in  \eqref{eq:gamma-bar-j}-\eqref{eq:alpha-bar-j}, the population analogs of our finite sample estimating equations. The first order conditions of these minimization problems yield 
\begin{equation}
    \label{eq:foc-plim-parameters}
    \begin{split}
     \E
        \underbrace{\overbrace{\begin{bmatrix} 
                p_j(X)\{D(1 + e^{\bar\gamma'Z})Z +  Z\} \vspace{0.2cm}\\
                p_j(X)De^{-\bar\gamma'Z}(DY - \bar\alpha'Z)Z
        \end{bmatrix}}^{\text{First order condition of \(\bar\gamma_j\)}}}_{\text{First order condition of \(\bar\alpha_j\)}} 
        = 0 \;\implies\; \E[p_j(X)\nabla_{\gamma_j,\alpha_j}Y(\bar\pi_j,\bar m_j)] = 0
    \end{split}
\end{equation}
Examining the first order conditions in \eqref{eq:foc-plim-parameters}, we see that they exactly give us control over the gradient \eqref{eq:gradient-plim-parameters}. Under suitable convergence of the first stage parameter estimates, this guarantees the bias examined in expansion \eqref{eq:taylor-expansion-gamma-alpha} is negligible even under misspecification of the propensity score or outcome regression models.

Control of this gradient under misspecification is not provided using other estimating equations, such as maximum likelihood for the logistic propensity score model or ordinary least squares for the linear outcome regression model. Moreover, control over the gradient of \(\vB_j^k\) from \eqref{eq:beta-means} is not provided by the first-order conditions for \(\bar\gamma_l\) and  \(\bar\alpha_l\) for \(l\neq j\):
\begin{equation}
    \label{eq:foc-l-does-not-control-gradient-j}
    \begin{split}
        \E[p_j(X)\nabla_{\gamma_j, \alpha_j}Y(\bar\pi_j, \bar m_j)] 
        &= \E\begin{bmatrix} 
        -p_j(X)De^{-\bar\gamma'Z}(Y - \bar\alpha'Z)Z \vspace{0.2cm} \\
        p_j(X)\{D(1 + e^{\bar\gamma'Z})Z +  Z\}
        \end{bmatrix} \\
        &\neq \E
        \underbrace{\overbrace{\begin{bmatrix} 
        p_l(X)\{D(1 + e^{\bar\gamma'Z})Z +  Z\} \vspace{0.2cm}\\
         p_l(X)De^{-\bar\gamma'Z}(Y - \bar\alpha'Z)Z
        \end{bmatrix}}^{\text{First order condition of \(\bar\gamma_l\)}}}_{\text{First order condition of \(\bar\alpha_l\)}} 
    .\end{split}
\end{equation}
Showing that the inference procedure of \Cref{sec:setup} remains valid at all points \(x\in\calX\) under misspecification requires showing negligible first stage estimation bias for any linear combination of the vector~\eqref{eq:beta-means}. As outlined above, this requires using \(k\) separate pairs of nuisance parameter estimator to obtain \(k\) separate pairs of first order conditions, one for each term of the vector.

\section{Main Results}%
\label{sec:first-stage}

In this section, we present the main consistency and distributional results for our second-stage estimator \(\widehat g(x)\) described in \Cref{sec:setup}. A full set of second stage results, including pointwise and uniform linearization lemmas and uniform convergence rates, can be found in \Cref{sec:additional-second-stage}. The first set of results is established under the following condition, which limits the bias passed from first-stage estimation onto the second-stage estimator. In particular, \Cref{cond:no-effect} implies that the bias vector \(\vB^k\) from \eqref{eq:beta-means} satisfies \(\|\vB^k\| = o_p(n^{-1/2})\). 

\begin{condition}[No Effect of First-Stage Bias]
   \label{cond:no-effect}
    \begin{equation}
        \label{eq:no-effect-first-stage-bias}
        \max_{1 \leq j \leq k}\big|\E_n[p_j(X)Y(\widehat\pi_j, \widehat m_j)] - \E_n[p_j(X)Y(\bar\pi_j, \bar m_j)]\big| = o_p(n^{-1/2}k^{-1/2})
    .\end{equation} 
\end{condition}

Via \Cref{thm:first-stage-convergence} we can see that is a logistic propensity score model and a linear outcome regression model and estimating the first stage models using the estimating equations \eqref{eq:gamma-j-estimating-equation}-\eqref{eq:alpha-j-estimating-equation}, \Cref{cond:no-effect} can be achieved under  \Cref{assm:logistic-model-convergence} and the sparsity bound 
\begin{equation}
    \label{eq:first-stage-sparsity-bound}
    \frac{s_k\,\xi_{k,\infty}^2k^{1/2}\ln(d_z)}{\sqrt{n}} \to 0 
.\end{equation} 
If the researcher were to assume different parametric forms for the first stage model, different first estimating equations would have to be used to obtain doubly-robust estimation and inference. However, so long as the \Cref{cond:no-effect} can be established at the limiting values of the first stage models, the results of this section hold.

Having dealt with the first stage estimation error, the main complication remaining is that under misspecification the aIPW signals  \(Y(\hat\pi_j, \hat m_j)\) for \(j = 1,\dots,k\) do not all converge to the same limiting values. However, so long as at least one of the first stage models is correctly specified, all of the limiting aIPW signals have the same conditional mean, \(g_0(x)\). 
In the standard setting, consistency of nonparametric estimator relies on certain conditions on the error terms. In our setting, we require that these assumptions hold uniformly over \(k\) the error terms. We note though that there is a non-trivial dependence structure between that limiting aIPW signals. This strong dependence gives plausibility to our uniform conditions. For example, if the logistic propensity score model is correctly specified and the limiting outcome regression models are uniformly bounded conditional on  \(Z\), our conditions reduce exactly to the conditions of \citet{BCCK-2015}. In general, however, the uniform conditions suggest that a degree of undersmoothing is optimal when implementing our estimation procedure.

\subsection{Pointwise Inference}%
\label{subsec:pointwise-inference}

Pointwise inference relies on the following assumption in tandem with \Cref{cond:no-effect}.
\begin{assumption}[Second-Stage Pointwise Assumption]
    \label{assm:second-stage-assumptions}
    Let \(\bar \eps_k := \max_{1\leq j\leq k}|\eps_j|\). Assume that
    \begin{enumerate}[(i)]
        \item Uniformly over all \(n\), the eigenvalues of  \(Q = \E[p^k(x)p^k(x)']\) are bounded from above and away from zero.
        \item The conditional variance of the error terms is uniformly bounded in the following sense. There exist constants \(\underline\sigma^2\) and  \(\bar\sigma^2\) such that for any \(j=1,2\dots\) we have that  \(\underline{\sigma}^2 \leq \Var(\eps_j \mid X)\leq \bar\sigma^2 < \infty;\)
        \item For each \(n\) and  \(k\) there are finite constants  \(c_k\) and  \(\ell_k\) such that for each \(f\in \calG\)
        \[
            \|r_k\|_{L,2} = (\E[r_k(x)^2])^{1/2}\leq c_k\andbox \|r_k\|_{L,\infty} = \sup_{x\in\calX}|r_k(x)| \leq \ell_kc_k
        .\] 
        \item \(\sup_{x\in\calX}\E[\bar\eps_k^2\,\bm{1}\{\bar\eps_k + \ell_kc_k > \delta\sqrt{n}/\xi_k\} \mid X= x]\to 0\) as \(n\to\infty\)  and \(\sup_{x\in\calX}\E[\ell_k^2c_k^2\bm{1}\{\bar\eps_k + \ell_kc_k > \delta\sqrt{n}/\xi_k\} \mid X= x]\to 0\) as \(n\to\infty\) for any  \(\delta > 0\). 
    \end{enumerate}
\end{assumption}

As mentioned, these are exactly the conditions required by \citet{BCCK-2015}, with the modification that the bounds on conditional variance and other moment conditions on the error term hold uniformly over \(j = 1,\dots,k\). The assumptions on the series terms being used in the approximation can be shown to be satisfied by a number of commonly used functional bases, such as polynomial bases or splines, under adequate normalizations and smoothness of the underlying regression function.  Readers should refer to \citet{Newey-1997}, \citet{Chen-2007}, or \citet{BCCK-2015} for a more in depth discussion of these assumptions.\footnote{In practice, we recommend the use of B-splines in order to to satisfy the first requirement that the basis functions are weakly positive and to reduce instability of the convex optimization programs described in \eqref{eq:gamma-j-estimating-equation}-\eqref{eq:alpha-j-estimating-equation}.} 

Under these assumptions, the variance of our second stage estimator is governed by one of the following variance matrices:
\begin{equation}
    \label{eq:omega-definitions}
    \begin{split}
        \tilde\Omega &:= Q^{-1}\E[\{p^k(x)\circ(\eps^k + r_k)\}\{p^k(x)\circ(\eps^k + r_k)\}' ]Q^{-1} \\
        \Omega_0 &:= Q^{-1}\E[\{p^k(x)\circ\eps^k\}\{p^k(x)\circ\eps^k\}']Q^{-1}\\ 
    \end{split}
\end{equation}
where \(\circ\) represents the Hadamard (element-wise) product and, abusing notation, for a vector \(a \in \SR^k\) and scalar \(c\in \SR\) we let \(a + c = (a_i + c)_{i=1}^k\). Later on, we establish the validity of the plug-in analog \(\hat\Omega\) \eqref{eq:omega-hat-definitions}, as an estimator of these matrices.

\begin{theorem}[Pointwise Normality]
    \label{thm:pointwise-normality}
    Suppose that \Cref{cond:no-effect} and \Cref{assm:second-stage-assumptions} hold. In addition suppose that \(\xi_k^2\log k/n\to 0\). Then so long as either the logistic propensity score model or linear outcome regression model is correctly specified, for any \(\alpha\in S^{k-1}\): 
    \begin{equation}
        \label{eq:thm-pointwise-normality-1}
        \sqrt{n}\frac{\alpha'(\widehat\beta^k-\beta^k)}{\|\alpha'\Omega^{1/2}\|} \to_d N(0,1) 
    \end{equation}
    where generally \(\Omega = \tilde\Omega\) but if  \(\ell_kc_k \to 0\) then we can set  \(\Omega = \Omega_0\). Moreover, for any  \(x\in\calX\) and  \(s(x) := \Omega^{1/2}p^k(x)\),
    \begin{equation}
        \label{eq:thm-pointwise-normality-2}
        \sqrt{n}\frac{p^k(x)'(\widehat\beta^k - \beta^k)}{\|s(x)\|} \to_d N(0,1) 
    \end{equation}
    and if the approximation error is negligible relative to the estimation error, namely \(\sqrt{n}r_k(x) = o(\|s(x)\|)\), then
    \begin{equation}
        \label{eq:thm-pointwise-normality-3}
        \sqrt{n}\frac{\widehat g(x) - g(x)}{\|s(x)\|}\to_d N(0,1) 
    \end{equation}
\end{theorem}
\Cref{thm:pointwise-normality} shows that the estimator proposed in \Cref{sec:setup} has a limiting gaussian distribution even under misspecification of either first stage model. This allows for doubly-robust pointwise inference after establishing a consistent variance estimator.

\subsection{Uniform Convergence}%
\label{subsec:uniform-inference}

Next, we turn to strengthening the pointwise results to hold uniformly over all points \(x\in\calX\). This requires stronger conditions. we make the following assumptions on the tail behavior of the error terms which strengthens \Cref{assm:second-stage-assumptions}. 

\begin{assumption}[Uniform Limit Theory]
    \label{assm:uniform-limit-theory}
    Let \(\bar\eps_k = \sup_{1\leq j\leq k} |\eps_j|\), \(\alpha(x) := p^k(x)/\|p^k(x)\|\), and let 
    \[
        \xi_k^L := \sup_{\substack{x,x'\in\calX \\ x\neq x'}} \frac{\|\alpha(x) - \alpha(x')\|}{\|x-x'\|} 
    .\] 
    Further for any integer \(s\) let \(\bar\sigma_k^s = \sup_{x\in\calX}\E[|\bar\eps_k|^s|X=x]\). For some \(m > 2\)  assume
    \begin{enumerate}[(i)]
        \item The regression errors satisfy \(\sup_{x\in\calX} \E[\max_{1 \leq i \leq n} |\bar\eps_{k,i}|^m \mid X = x] \lesssim_P n^{1/m}\)
        \item The basis functions are such that (a) \(\xi_k^{2m/(m-2)}\log k/n \lesssim 1\), (b) \((\bar\sigma_k^2\vee\bar\sigma_k^m)\log\xi_k^L \lesssim \log k\), and  (c) \(\log\bar\sigma_k^m\xi_k \lesssim \log k\).
    \end{enumerate}
\end{assumption}

As before, \Cref{assm:uniform-limit-theory} is very similar to its analogue in \citet{BCCK-2015}, with the modification that the conditions are required to hold for \(\bar\eps_k\) as opposed to \(\eps_k\). Under this assumption, we derive doubly-robust uniform rates of convergence uniform inference procedures for the conditional counterfactual outcome  \(g_0(x)\).
 
\begin{theorem}[Strong Approximation by a Gaussian Process]
    \label{thm:strong-approximation}
    Assume that \Cref{cond:no-effect} holds and that Assumptions~\ref{assm:second-stage-assumptions}-\ref{assm:uniform-limit-theory} hold with \(m\geq 3\). In addition assume that (i) \(\bar R_{1n} = o_p(a_n^{-1})\) and (ii) \(a_n^6k^4\xi_k^2(\bar\sigma_k^3 + \ell_k^3c_k^2)^2\log^2 n/n\to 0\) where
    \begin{align*}
            \bar R_{1n} := \sqrt{\frac{\xi_k^2\log k}{n}}(n^{1/m}\sqrt{\log k} + \sqrt{k}\ell_k c_k) \andbox
            \bar R_{2n} := \sqrt{\log k}\cdot\ell_kc_k
    \end{align*}
    Then so long as either the propensity score model or outcome regression model is correctly specified, for some  \(\calN_k\sim N(0,I_k)\):
    \begin{equation}
        \label{eq:uniform-normality-beta}
        \sqrt{n}\frac{\alpha(x)'(\widehat\beta-\beta)}{\|\alpha(x)'\Omega^{1/2}\|} =_d \frac{\alpha(x)'\Omega^{1/2}}{\|\alpha(x)'\Omega^{1/2}\|}N_k + o_p(a_n^{-1})\;\;\text{in }\ell^\infty(\calX)  
    \end{equation}
    so that for \(s(x) := \Omega^{1/2}p^k(x)\)
    \begin{equation}
        \label{eq:uniform-normality-intermediate}
        \sqrt{n}\frac{p^k(x)'(\widehat\beta-\beta)}{\|s(x)\|} =_d \frac{s(x)}{\|s(x)\|}N_k + o_p(a_n^{-1})\;\;\text{in }\ell^\infty(\calX)  
    \end{equation}
    and if \(\sup_{x\in\calX} \sqrt{n}|r_k(x)|/\|s(x)\| = o(a_n^{-1})\), then
    \begin{equation}
        \label{eq:uniform-normality-ghat}
        \sqrt{n}\frac{\widehat g(x) - g(x)}{\|s(x)\|} =_d \frac{s(x)'}{\|s(x)\|}\calN_k + o_p(a_n^{-1})\;\;\text{in }\ell^\infty(\calX)  
    \end{equation}
    where in general we take \(\Omega = \tilde\Omega\) but if  \(\bar R_{2n} = o_p(a_n^{-1})\) then we can set  \(\Omega = \Omega_0\) where  \(\tilde\Omega\) and  \(\Omega_0\) are as in  \eqref{eq:omega-definitions}.
\end{theorem}

\Cref{thm:strong-approximation} establishes conditions under which we obtain a doubly-robust strong approximation of the empirical process \(x \mapsto \sqrt{n}(\widehat g(x) - g_0(x))\) by a Gaussian process. After establishing consistent estimation of the matrix \(\Omega\), this  strong approximation result allows us to show validity of the uniform confidence bands described in \Cref{sec:setup}.  
As noted by \citet{BCCK-2015}, this is distinctly different from a Donsker type weak convergence result for the estimator \(\widehat g(x)\) as viewed as a random element of \(\ell^\infty(X)\). In particular, the covariance kernel is left completely unspecified and in general need not be well behaved.

\subsection{Matrix Estimation and Uniform Inference}%
\label{subsec:uniform-bootstrap}

We establish that the estimator \(\widehat\Omega\) proposed in \eqref{eq:omega-hat-definitions} is a consistent estimator of the true limiting variance \(\Omega\), where  \(\Omega = \tilde\Omega\) in general but if  \(\bar R_{2n} = o_p(a_n^{-1})\) then \(\Omega = \Omega_0\). To do so, we rely on the second stage assumptions \Cref{assm:second-stage-assumptions,assm:uniform-limit-theory} as well as the following condition limiting the first stage estimation error passed on to the variance estimator \(\widehat\Omega\).

\begin{condition}[Variance Estimation]
    \label{cond:variance-estimation}
    Let \(m > 2\) be as in \Cref{assm:uniform-limit-theory}. Then, 
    \begin{equation}
        \label{eq:variance-condition}
        \xi_{k,\infty}\max_{1\leq j\leq k}\E_n[p_j(X)^2(Y(\widehat\pi_j, \widehat m_j) - Y(\bar\pi_j, \bar m_j))^2] = o_p(k^{-2}n^{-1/m})
    \end{equation}
\end{condition}

Via \Cref{thm:first-stage-convergence} we can establish \Cref{cond:variance-estimation} under \Cref{assm:logistic-model-convergence} as well as the additional sparsity bound\footnote{The sparsity bound \eqref{eq:variance-sparsity-bound} required for consistent variance estimation can be significantly sharpened if the researcher is willing to use a cross fitting procedure, using one sample to estimate the nuisance parameters and another to evaluate the aIPW signal. This is because one could more directly follow \citet{SC-2020} and control alternate quantities with bounds that converge more quickly to zero.}
\begin{equation}
    \label{eq:variance-sparsity-bound}
    \frac{\xi_{k,\infty}^5s_k^2k^2\ln(d_z)}{n^{(m-1)/m}} 
.\end{equation}

\begin{theorem}[Matrix Estimation]
    \label{thm:matrix-estimation}
    Suppose that Conditions~\ref{cond:no-effect} and \ref{cond:variance-estimation} and Assumptions~\ref{assm:second-stage-assumptions}-\ref{assm:uniform-limit-theory} hold. In addition, assume that \(\bar R_{1n} + \bar R_{2n} \lesssim (\log k)^{1/2}\). Then, so long as either the propensity score model or outcome regression model is correctly specified then for \(\widehat\Omega = \widehat{Q}^{-1}\widehat\Sigma \widehat{Q}^{-1}\):
    \begin{align*}
        \|\widehat\Omega - \Omega\| &\lesssim_P (v_n \vee \ell_kc_k)
        \sqrt{\frac{\xi_k^2\log k}{n}} = o(1)
    \end{align*}
\end{theorem}
\Cref{thm:matrix-estimation} establishes that pointwise inference based on the test statistic described in \Cref{sec:setup}, obtained by replacing \(\Omega\) in \Cref{thm:pointwise-normality} with the consistent estimator \(\widehat\Omega\), is doubly-robust. Hypothesis tests based on the test statistic as well as pointwise confidence intervals for \(g_0(x)\) remain valid even if one of the first stage parameters is misspecified.

We now establish the validity of uniform inference based on the gaussian bootstrap critical values \(c_u^\star(1-\alpha)\) defined in \Cref{sec:setup}.

\begin{theorem}[Validity of Uniform Confidence Bands]
    \label{thm:uniform-confidence-bands}
    Suppose \Cref{cond:no-effect,cond:variance-estimation} are satisfied and \Crefrange{assm:second-stage-assumptions}{assm:uniform-limit-theory} hold with \(m\geq 4\). In addition suppose (i) \(R_{1n} + R_{2n} \lesssim \log^{1/2} n\), (ii)  \(\xi_k\log^2 n /n^{1/2-1/m} = o(1)\), (iii) \(\sup_{x\in\calX} |r_k(x)|/\|p^k(x)\| = o(\log^{-1/2} n)\), and (iv) \(k^4\xi_k^2(1 + l_k^3r_k^3)^2\log^5 n/n=o(1)\). Then, so long as either the propensity score model or outcome regression model is satisfied 
    \[
       \Pr\left(\sup_{x\in\calX} |\frac{\widehat g(x)-g(x)}{\widehat\sigma(x)}| \leq c^\star(1-\alpha)\right) = 1-\alpha + o(1)
    .\] 
    As a result, uniform confidence intervals formed in \eqref{eq:confidence-bands} satisfy 
    \[
        \Pr(g(x) \in [\underline i(x), \bar i(x)],\;\forall x\in\calX) = 1 - \alpha + o(1)
    .\] 
\end{theorem}

In conjunction with \Cref{thm:first-stage-convergence}, \Cref{thm:pointwise-normality} and \Cref{thm:matrix-estimation}, \Cref{thm:uniform-confidence-bands} shows the validity of the uniform inference procedure described in \Cref{sec:setup}.

\section{Estimation of the Conditional Average Treatment Effect}
\label{sec:cate-wrapup}

Up to now, we have mainly focused on doubly-robust estimation and model-assisted inference for the function
\[
    g_0(x) = \E[Y_1 \mid X= x]
.\]
We conclude by noting that we can use a symmetric procedure to obtain model-assisted inference for the additional conditional counterfactual outcome
\[
    \tilde g_0(x) = \E[Y_0 \mid X = x]
.\]
To do so, we use the alternate aIPW signal 
\[
    Y_0(\pi_0,  m_0) = \frac{(1-D)Y}{1-\pi_0(Z)} + \left(\frac{1-D}{1-\pi_0(Z)} - 1\right)m_0(Z)
\] 
where as before the true value for \(\pi^\star_0(z) = \Pr(D = 1\mid Z = z)\) but now \( m_0^\star(z) = \E[Y \mid D= 0, Z= z]\). To estimate these nuisance models we again assume a logistic form for the propensity score model \(\pi_0(z) = \pi(z; \gamma^0)\) and a linear form for the outcome regression model \(m_0(z) = m(z, \alpha^0)\) as in \eqref{eq:nuisance-parameter-functional-forms} and use a separate estimation procedure for each basis term in our series approximation of \(\tilde g_0(x)\). The estimating equations we use to estimate each \(\gamma_j^0\) and  \(\alpha_j^0\) differ from those in  \eqref{eq:gamma-j-estimating-equation}-\eqref{eq:alpha-j-estimating-equation} however, and are instead given
\begin{align*}
    \widehat\gamma_j^0 &:= \arg\min_\gamma\, \E_n[p_j(X)\{(1-D)e^{\gamma'Z} - D\gamma'Z\}] + \lambda_{\gamma,j}\|\gamma\|_1 \\
    \widehat\alpha_j^0 &:= \arg\min_\alpha\, \E_n[p_j(Z)(1-D)e^{{\widehat{\gamma}_j^{0'}}Z}(Y - \alpha'Z)^2]/2 + \lambda_{\alpha,j}\|\alpha\|_1
\end{align*}
which under the natural analog of \Cref{assm:logistic-model-convergence} converge uniformly to population minimizers:
\begin{align*}
    \bar\gamma_j^0 &:= \arg\min_\gamma\, \E[p_j(X)\{(1-D)e^{\gamma'Z} - D\gamma'Z\}] \\
    \bar\alpha_j^0 &:= \arg\min_\alpha\, \E[p_j(Z)(1-D)e^{{\bar{\gamma}_j^{0'}}Z}(Y - \alpha'Z)^2] 
\end{align*}
Letting \(\bar\pi_{0,j}(z) = \pi(z, \bar\gamma^0_j)\), and \(\bar m_{0,j}(z) = m(z, \bar\alpha_j^0)\) we can repeat the decomposition of \Cref{sec:theory-overview}, expressing \(\tilde Y(\bar\pi_{0,j}, \bar m_{0,j})\) as functions of the parameters \(\bar\gamma_j^0\) and  \(\bar\alpha_j^0\) and show that the first order conditions for \(\bar\gamma_j^0\) and  \(\bar\alpha_j^0\) directly control the bias passed on to the second stage nonparametric estimator for \(\tilde g_0(x)\). Convergence rates and validity of inference then follow from symmetric analysis of the results in \Cref{sec:theory-overview,sec:first-stage}. Combining estimation and inference of the two conditional counterfactual outcomes then gives a doubly-robust estimator and inference procedure for the CATE. To perform inference on the CATE we can use the variance matrix
\[
    \bar\Omega = \Omega_0 + \Omega_1 - 2\Omega_2
\] 
where \(\Omega_0\) is as in \eqref{eq:omega-definitions} but \(\Omega_1\) and \(\Omega_2\) are given 
\begin{equation}
    \label{eq:omega12}
    \begin{split}
        \Omega_1 &= Q^{-1}\E[\{p^k(x)\circ\eps_0^k\}\{p^k(x)\circ\eps_0^k\}']Q^{-1} \\
        \Omega_2 &= Q^{-1}\E[\{p^k(x)\circ\eps^k\}\{p^k(x)\circ\eps_0^k\}']Q^{-1}
    \end{split}
\end{equation}
where \(\eps_{0,j}^k = Y_0(\bar\pi_{0,j},\bar m_{0,j}) - \tilde g_0(x)\) and  \(\eps_0^k = (\eps_{0,1}^k,\dots,\eps_{0,k}^k)'\). These matrices can be consistently estimated using their natural empirical analogs as in \eqref{eq:omega-hat-definitions}.

\section{Simulation Study}%
\label{sec:simulations}

We investigate the finite-sample performance of the doubly-robust estimator and inference procedure via simulation study. We find that our proposed estimation procedure retains good coverage properties even under misspecification.

\subsection{Simulation Design}

Observations are generated i.i.d. according to the following distributions
The error term is generated following \(\epsilon \sim N(0, 1) \). The controls are set \(Z_i = (Z_{1i},X_i) \in \SR^{d_z}\) where \(d_z= 100\), \(X\sim U(1,2)\), and the independent regressors \(Z_1 \) are jointly centered Gaussian with a covariance matrix of the Toeplitz form
\begin{align*}
    \mathrm{Cov}(Z_{1,j},Z_{1,k}) = \E [Z_{1,j} Z_{1,k}] = 2^{-|j-k|}, \ \ \ 3\leq j,k \leq d_z.
\end{align*}
To capture misspecification, we let \(Z^\dagger\) be a transformation of the regressors in \(Z_1\) where \(Z_j^\dagger = Z_j + \max (0,1+Z_j)^2, \ \forall \ j=3,\dots,d_z\). Let \texttt{sparsity} control the number of regressors in \(Z = (Z_1,X)\) entering the DGP.
\begin{enumerate}[label=(S\arabic*)]
    \item \textit{Correct specification}: Generate \(D\) given \(Z\) from a Bernoulli distribution with \(\Pr (D=1 | Z ) = \{ 1+\exp(p_1 - X - 0.5X^2 - \gamma'Z_1)  \}^{-1}\) and \(Y = D(1 + X + 0.5X^2 + \gamma'Z_1) + \epsilon.\)
    \item \textit{Propensity score model correctly specified, but outcome regression model misspecified}: Generate \(D\) given \(Z\) as in (S1), but \( Y = D(1 + X + 0.5X^2 + \gamma'Z_1^\dagger) + \epsilon. \)
    \item \textit{Propensity score model misspecified, but outcome regression model correctly specified}: Generate \(Y\) according to (S1), but generate \(D\) given \(Z\) from a Bernoulli distribution with \(\Pr (D=1 | Z ) = \{ 1+\exp(p_2 - X - 0.5X^2 + \gamma'Z_1^\dagger)  \}^{-1}\).
\end{enumerate}
where the constants \(p_1\) and \(p_2\) differ in various simulation setups but are always set so that the average probability of treatment is about one half. To consider various degrees of high-dimensionality, we implement \( N \in \{500, 1000\} \) with \(d_z = 100\). For (S1), \texttt{sparsity}\(=6\); for (S2), \texttt{sparsity}\(=4\); and, for (S3), \texttt{sparsity}\(=5\). Results are reported for \(S=1,000\) repeated simulations.

\subsection{Estimators and Implementation}

To select the first stage penalty parameters, we implement the multiplier bootstrap procedure described in \Cref{subsec:additional}. The constants \(c_{\gamma,j}\) and \(c_{\alpha,j}\) in the pilot penalty parameters \eqref{eq:pilot-penalty} are selected via cross validation from a set of size 5. To select the final bootstrap penalty parameter we set \(c_0 = 1.1\) and select the \(95^\text{\tiny th}\) quantile of \(B=10000\) bootstrap replications.
In our second-stage estimation, we use a b-spline basis of size \(k=3\).
B-splines are implemented from the R package \texttt{splines2} \citep{splines2-paper}, which uses the specification detailed in \cite{perperoglou2019review}. In the tables below, we refer to our method as \textit{MA-DML} (model assisted double machine learning).

We compare our proposed estimator and inference procedure to that of \citet{SC-2020}, which projects a single aIPW signal onto a growing series of basis terms. In implementing this \textit{DML} method, we use the standard \(\ell_1\)-penalized maximum likelihood (MLE) and ordinary least squares (OLS) loss functions to estimate the first stage propensity score and outcome regression models, respectively.\footnote{Vira Semenova provides several example \texttt{R} scripts implementing \textit{DML}: \url{https://sites.google.com/view/semenovavira/research}.}

Estimation error is studied for the target parameter \(g_0 (x)= \E [ Y| D=1, X=x]\) over a grid of 100 points spaced across \(x\in[1,2]\), i.e. the support of \( X \). We study average coverage across simulations of each method's pointwise (at \(x = 1.5 \)) and uniform confidence intervals.
To compare the estimation error for the target parameter \( g(x) \) across the two different estimators \( \widehat g_s (x) \) for each simulation \( s = 1,\dots, S \), we utilize integrated bias, variance, and mean-squared error where \( \Bar{g} (x) = S^{-1} \sum^S_{s=1} \widehat g_s (x), \)
\begin{align*}
    \mathrm{IBias}^2 &= \int_0^1 (\Bar{g} (x) - g_0 (x))^2 dx, \\
    \mathrm{IVar} &= S^{-1} \sum_{s=1}^S \int_0^1 (\widehat{g}_s (x) - \Bar{g} (x))^2 dx, \\
    \mathrm{IMSE} &= S^{-1} \sum_{s=1}^S \int_0^1 (\widehat{g}_s (x) - g_0 (x))^2 dx. \\
\end{align*}

\subsection{Simulation Results}

Table \ref{tab:simulation} presents the simulation results for all three specifications (S1)-(S3) for \(n=500\)  and \(n = 1000\). Integrated squared bias, variance, and mean squared error are presented in columns (1)-(3), respectively. Pointwise and uniform coverage results are presented in columns (4)-(7).

\begin{table}
\caption{Simulation study.} \label{tab:simulation}
\begin{center}
\begin{tabular}{ p{1cm} l c c c c c c c}
\hline \hline \\[-10pt]
 DGP & Estimator & IBias$^2$ & IVar & IMSE & Cov90 & Cov95 & UCov90 & UCov95 \\[1pt]
  &  & (1) & (2) & (3) & (4) & (5) & (6) & (7) \\[1ex]
 \hline
 & & \multicolumn{7}{c}{K=3, n=500, $d_z$ = 100} \\ \cline{3-9}
 \multirow{2}{1em}{(S1)} & DML & 0.04 & 0.31 & 0.35 & 0.92 & 0.96 & 1.00 & 1.00 \\
  & MA-DML & $\sim$0.0 & 0.34 & 0.34 & 0.93 & 0.97 & 1.00 & 1.00  \\[7pt]
 \multirow{2}{1em}{(S2)} & DML & 0.16 & 2.17 & 2.33 & 0.92 & 0.97 & 0.83 & 0.86 \\  
  & MA-DML & 0.03 & 2.12 & 2.15 & 0.90 & 0.94 & 0.88 & 0.91  \\ [7pt]
 \multirow{2}{1em}{(S3)} & DML & 0.03 & 0.55 & 0.59 & 0.87 & 0.93 & 0.95 & 0.97 \\  
  & MA-DML & 0.01 & 0.79 & 0.80 & 0.91 & 0.95 & 0.99 & 0.99 \\[1ex] \cline{3-9}
  & & \multicolumn{7}{c}{K=3, n=1000, $d_z$ = 100} \\ \cline{3-9}
 \multirow{2}{1em}{(S1)} & DML & 0.12 & 0.20 & 0.32 & 0.83 & 0.90 & 0.96 & 0.96 \\  
  & MA-DML & 0.01 & 0.22 & 0.23 & 0.83 & 0.90 & 0.99 & 0.99  \\[7pt]
 \multirow{2}{1em}{(S2)} & DML & 0.40 & 2.1 & 2.5 & 0.84 & 0.91 & 0.33 & 0.39 \\  
  & MA-DML & 0.19 & 2.07 & 2.26 & 0.83 & 0.89 & 0.50 & 0.55  \\ [7pt]
 \multirow{2}{1em}{(S3)} & DML & 0.11 & 0.34 & 0.46 & 0.74 & 0.82 & 0.80 & 0.84 \\  
  & MA-DML & 0.01 & 0.53 & 0.54 & 0.84 & 0.89 & 0.89 & 0.91  \\
 \hline\hline \\[-5pt]
 \multicolumn{9}{c}{\parbox{12.2cm}{\footnotesize Note: DGP refers to the three various data generating processes introduced above. IBias$^2$, IVar, and IMSE refer to integrated squared bias, variance, and mean squared error, respectively. Cov90, Cov95, UCov90, and UCov95 refer to the coverage proportion of the 90\% and 95\% pointwise and uniform confidence intervals across simulations. $K$ refers to the number of series terms, $N$ to the sample size, and $d_z$ to the dimensionality of the random variable $Z_1.$}} \\
\end{tabular}
\end{center}
\end{table}

For pointwise and uniform coverage under correct specification regime (S1), \textit{MA-DML} has some slight improvements. Under misspecification DGPs (S2) and (S3), the pointwise coverage of \textit{MA-DML} is closer to the targets except in the $N=1000$ and (S2) case where it slightly underperforms. However, \textit{MA-DML} has a notable improvement over \textit{DML} in the (S3) case when $N=1000.$ Similarly, \textit{MA-DML} outperforms \textit{DML} in three of the four misspecified regimes, i.e. all but (S3) when $N=500$ where \textit{MA-DML} has over-coverage. Under (S2) when $N=1000,$ both methods are markedly deterioated uniform coverage, although \textit{MA-DML} is noticably closer to target. 

In regards to estimation error, in four of the six settings, \textit{MA-DML} has a lower MSE than \textit{DML} where regardless of sample size \textit{MA-DML} underperforms in (S3). Notably, it does appear \textit{MA-DML} has substantially smaller IBias$^2$ across the DGPs.

Finally, we were surprised to find for both estimators that coverage properties, in general, improve under the higher-dimensional regime of $N=500$ with $d_z=100$ compared to $N=1,000$ and $d_z=100.$ In particular, with a higher ratio of covariates to observations, the uniform coverage properties under regime (S2) were substantially better. The estimation error results were in line with our priors as the higher-dimensional regime sees in general higher estimation errors for both methods.

For coverage under correct specification, we did anticipate the underperformance of \textit{MA-DML} given it is designed to handle misspecification with the cost of other estimators outperforming under correct specification. Additionally, we attribute the poor uniform coverage in DGP (S2) for both estimators under $N=1,000$ to a lack of a rich enough cross-validation given the performance was improved under a more difficult regime when the number of observations drops to $N=500.$ The integrated bias of \textit{MA-DML} is lower across the various DGPs compared to \textit{DML}. Following the discussion in \Cref{sec:theory-overview} this is expected since the first stage estimating equations for the model assisted procedure are specifically designed to minimize the bias passed on to the second stage estimator. However, the model assisted procedure has higher values of integrated variance compared to the standard procedure, which could be attributable to the use of \(k\) distinct first-stage estimations.

Our findings should not be interpreted as a critique of the \citet{SC-2020} benchmark method, whose work we rely on and were inspired by. 

\section{Empirical Application}
\label{sec:empirical}

We apply the model assisted estimator to estimate the effect of maternal smoking on infant birthweight conditional on the age of the mother. We use the \citet{Cattaneo_2010} dataset which can be found online on the Stata website.\footnote{The dataset can be downloaded \href{http://www.stata-press.com/data/r13/cattaneo2.dta}{here}.} The dataset describes each infant's birthweight in grams, \(Y\), whether or not the mother smoked during pregnancy, \(D=1\) indicating smoking, and a number of covariates containing information on the mother's health and socioeconomic background, \(Z = (X,Z_1)\), where \(X\) represents the conditioning variable, maternal age. A full summary of the data used as well as additional details/analysis from our empirical analysis can be found in \Cref{sec:empirical-details}.

We compare the model assisted estimator of the CATE against one where standard MLE and OLS loss functions are used to estimate the first stage propensity score and outcome regression models. We also qualitatively compare our results to \citet{Zimmert_Lechner_CATE}, who use a kernel based approach to estimate the CATE in this setting. While this sort of comparison is not perfect since we do not know the true DGP, this setting is advantageous for analysis since we strongly expect that (i) the effect of smoking on birthweight will be negative and (ii) this effect should grow stronger in magnitude as the age of the mother increases. These hypotheses have been corroborated by other work that examines the conditional average treatment effect in this setting \citep{Zimmert_Lechner_CATE,Abreya_2006,Lee_Ryo_Wang_2016}.

\subsection{Empirical Results}

\Cref{fig:emp1} displays our main results from implementing both the model assisted and standard MLE/OLS estimation procedures. After removing the top 3\% and bottom 3\% of smoker and non-smoker birthweights by maternal age, we select the penalty parameters for the first stage models via the bootstrap procedure described in \Cref{sec:first-stage}. The pilot penalty parameters are uniformly taken to be equal to zero, so that the residuals used in the bootstrap procedure are generated from non-regularized estimations. We take \(c_0 = 2\) in \eqref{eq:final-penalty-parameters} and and select the first stage penalty parameters using the 99\textsuperscript{th}, 95\textsuperscript{th}, and  90\textsuperscript{th} quantiles of the bootstrap distribution. For the second stage basis functions we implement second degree b-splines with 3 knots via the splines2 package in \texttt{R} \citep{splines2-paper}.

\begin{figure}[htp!]
    \centering
    \includegraphics[width=\linewidth]{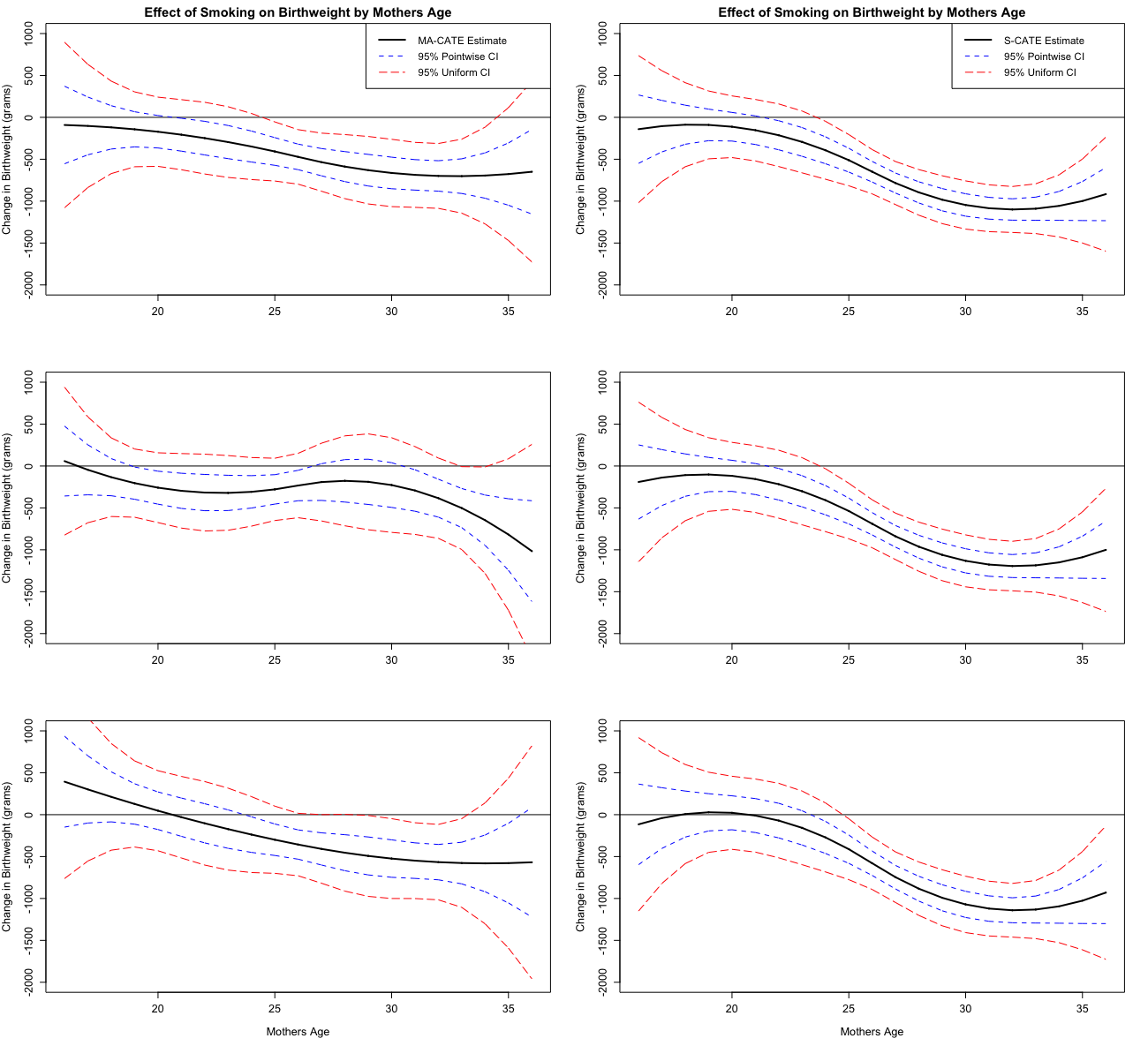}
    \caption{CATE of maternal smoking estimated using model assisted estimating equations (left) and standard MLE/OLS estimating equations (right). Top row uses the 99\textsuperscript{th} quantile of the bootstrap distribution to select the penalty parameters, second row uses 95\textsuperscript{th} quantile, and final row uses the 90\textsuperscript{th} quantile. Second stage is computed using b-splines of the second degree with 3 knots. 95\% pointwise confidence intervals are displayed in blue short dashes and 95\% uniform confidence bands are displayed in long red dashes.}%
    \label{fig:emp1}
\end{figure}

Consistent with prior work, both estimators of the CATE suggest that the effect of smoking on birthweight becomes more negative with age. Both estimation procedures also generally produces negative estimates for the CATE, but it should be noted that for the lowest levels of penalization the model assisted CATE estimate suggests a slightly positive effect of smoking for particularly young mothers, though this difference is not significantly different from zero. The shapes of the estimated functions remain relatively stable under various sizes of the penalty parameter, though the model assisted procedure displays a bit more sensitivity to the level of regularization introduced.\footnote{Numerically solving the minimization problems in \eqref{eq:gamma-j-estimating-equation}-\eqref{eq:alpha-j-estimating-equation} also typically requires more iterations to converge than solving the standard MLE/OLS minimization problems.} 

For the most part, the effects found here are similar to those found in \citet{Zimmert_Lechner_CATE}, though the effects estimated using standard first stage loss functions have somewhat larger magnitudes and in general both series estimation procedures seem to give less reasonable results on the boundaries. An advantage of using a series second stage however, compared to the kernel first stage of \citet{Zimmert_Lechner_CATE}, is the existence of the uniform confidence bands displayed. Reassuringly, the estimates of \citet{Zimmert_Lechner_CATE} seem to be within the 95\% uniform confidence bands generated by the model assisted estimator.

As a robustness check, we also try estimating the treatment effect using first degree b-splines instead of second degree splines. These results are displayed in \Cref{fig:empD1}. Again, we find that the effect of smoking on child birthweight is almost uniformly negative regardless of estimation procedure used or choice of penalty parameter. The shape of the estimated CATE function  using a standard MLE/OLS first stage is very stable to penalty choice here while the shape of the model assisted CATE function displays a bit more  instability here at the two lower levels of regularization.

\begin{figure}[htpb]
    \centering
    \includegraphics[width=\linewidth]{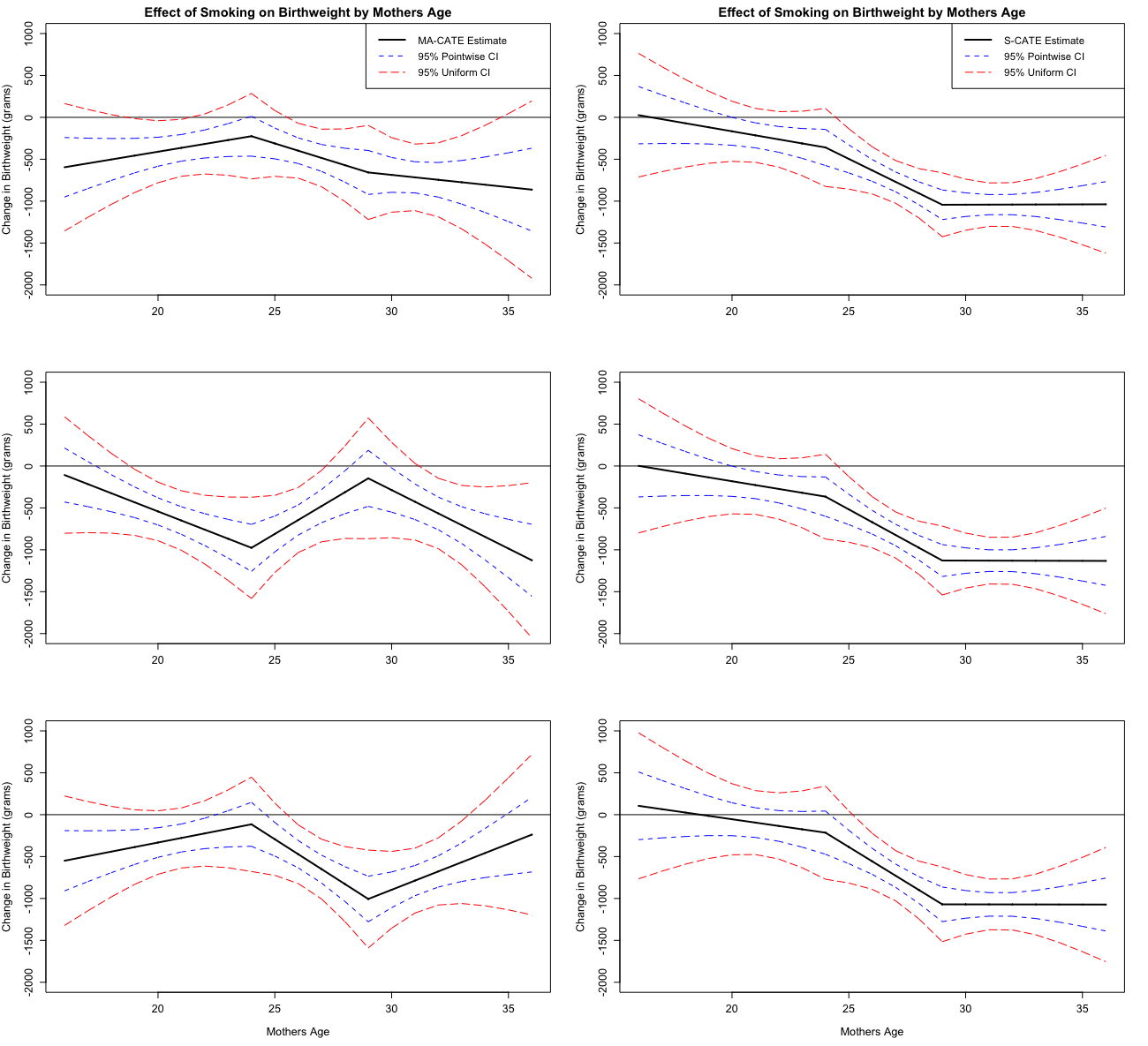}
    \caption{CATE of maternal smoking estimated using model assisted estimating equations (left) and standard MLE/OLS estimating equations (right). Top row uses the 99\textsuperscript{th} quantile of the bootstrap distribution to select the penalty parameters, second row uses 95\textsuperscript{th} quantile, and final row uses the 90\textsuperscript{th} quantile. Second stage is computed using b-splines of the first degree with 3 knots. 95\% pointwise confidence intervals are displayed in blue short dashes and 95\% uniform confidence bands are displayed in long red dashes.}%
    \label{fig:empD1}
\end{figure}
Finally, \Cref{tab:implied-ate} reports the smoothed average treatment effect estimates taken from averaging the model assisted CATE estimates from \Cref{fig:emp1} across observations. Again, these estimates are generally in line with prior work
\begin{table}[htpb]
    \centering
    \caption{Smoothed Model Assisted ATE Estimates}
    \vspace{0.1em}
    \label{tab:implied-ate}
    \begin{tabular}{c|ccc}
    \hline
    \hline
    Bootstrap Penalty Qt.& 99\textsuperscript{th} & 95\textsuperscript{th} & 90\textsuperscript{th}  \\
    \hline
     Implied ATE & -295.221 & -292.9086 & -453.2242 
    \end{tabular}
\end{table}

\section{Conclusion}%
\label{sec:conclusion}

Estimation of conditional average treatment effects with high dimensional controls typically relies on first estimating two nuisance parameters: a propensity score model and an outcome regression model. In a high-dimensional setting, consistency of the nuisance parameter estimators typically relies on correctly specifying their functional forms. While the resulting second-stage estimator for the conditional average treatment effect typically remains consistent even if one of the nuisance parameters is inconsistent, the confidence intervals may no longer be valid.

In this paper, we consider estimation and valid inference on the conditional average treatment effect in the presence of high dimensional controls and nuisance parameter misspecification. We present a nonparametric estimator for the CATE that remains consistent at the nonparametric rate, under slightly modified conditions, even under misspecification of either the logistic propensity score model or linear outcome regression model. The resulting Wald-type confidence intervals based on this estimator also provide valid asymptotic coverage under nuisance parameter misspecification.

\newpage
\singlespacing
\bibliography{bibtex/mlci.bib}

\appendix \onehalfspacing
\section{Proofs for Results in Main Text}%
\label{sec:proofs-main}

Here we provide proofs of the main results in Sections~\ref{sec:theory-overview}-\ref{sec:first-stage}. The proofs for \Cref{sec:first-stage} rely on an assortment of supporting lemmas proved in \Cref{sec:proofs-lemmas}. %

\subsection{Proofs for Main First Stage Results}
\label{subsec:first-stage-main-proofs}

\subsubsection*{Proof of \Cref{thm:first-stage-convergence}}

The proof of \Cref{thm:first-stage-convergence} relies on a series of non-asymptotic bounds that are established in Online Appendix Lemmas~\ref{lemma:nonasymptotic-logit} and \ref{lemma:nonasymptotic-outcome} that hold on \(\bigcap_{m=1}^6\Omega_{k,m}\) and depend on the quantity
\[
    \bar\lambda_k = M\xi_{k,\infty}\sqrt\frac{\log(d_z/\eps)}{n}
\]
where \(M\) is a fixed constant. In addition let \(\tilde\Sigma_{\alpha,j}^1 := \E_n[p_j(X)De^{-\bar\gamma_j'Z}|Y-\bar\alpha_j'Z|ZZ']\) and \(\Sigma_{\alpha,j}^1 := \E\tilde\Sigma_{\alpha,j}^1\) and define the event 
\begin{equation}
    \label{eq:mean-event}
    \Omega_{k,7} := \{\|\tilde\Sigma_{\alpha,j}^1 - \Sigma_{\alpha,j}^1\|_\infty \leq \bar\lambda_k, \forall j \leq k\} 
\end{equation}
In Online~\Cref{subsec:first-stage-probability-bounds} we show that \(\Pr(\bigcap_{m=1}^7) \geq 1  - o(1)\). Under these events, \Cref{lemma:nonasymptotic-means}, below provides the bound needed for first statement of \Cref{thm:first-stage-convergence} while \Cref{lemma:nonasymptotic-variance} provides the  bound needed for the second statement.

\begin{lemma}[Nonasymptotic Bounds for Weighted Means]
    \label{lemma:nonasymptotic-means}
    Suppose that \Cref{assm:logistic-model-convergence} holds, \(\xi_0  > (c_0 + 1)/(c_0 - 1)\), and \(2C_0\nu_0^{-2}s_k\bar\lambda_k \leq \eta < 1\). In addition, assume there is a constant \(c > 0\) such that  \(\lambda_{\alpha, j}/\lambda_{\gamma, j} \geq c\) for all \(j\leq k\). Then, under the event \(\bigcap_{m=1}^7 \Omega_{k,m}\), there is a constant  \(M_2\) that does not depend on  \(k\) such that
    \begin{equation}
        \label{eq:nonasymptotic-means}
        \max_{1\leq j\leq k}|\E_n[p_j(X)Y(\widehat\pi_j, \widehat m_j)] - \E_n[p_j(X)Y(\bar\pi_j,\bar m_j)]| \leq M_2s_k\bar\lambda_k^2 
    \end{equation}
\end{lemma}

\begin{proof}
    We show that the bound of \eqref{eq:nonasymptotic-means} holds for any \(j=1,\dots,k\) in a couple steps. To save notation, define 
    \begin{align*}
        \mu_j(\pi,m) &:= \E_n\left[p_j(X)Y(\pi, m)\right] \\
                     &= \E_n\bigg[p_j(X)\bigg\{\frac{DY}{\pi(Z)} + \bigg(\frac{D}{\pi(Z)} -1 \bigg)m(Z) \bigg\} \bigg]
    \end{align*}

    \emph{Step 1: Decompose Difference and Use Logistic FOCs.}
    Consider the following decomposition 
    \begin{align*}
        \mu_j(\widehat\pi_j,\widehat m_j) - \mu_j(\bar\pi_j,\bar m_j)
        &= \E\bigg[p_j(X)\{\widehat m_j(Z) - \bar m_j(Z)\}\bigg(1 - \frac{D}{\bar\pi_j(X)}\bigg) \bigg] \\
        &\;\;+ \E_n\bigg[p_j(X)D\{Y -  \bar m_j(Z)\}\bigg(\frac{1}{\widehat\pi_j(Z)} - \frac{1}{\bar\pi_j(Z)} \bigg) \bigg] \\
        &\;\;+ \E_n\bigg[p_j(X)\{\widehat m_j(Z) - \bar m_j(Z)\}\bigg(\frac{D}{\bar\pi_j(Z)} - \frac{D}{\widehat\pi_j(Z)}\bigg) \bigg] \\
        &:= \delta_{1,j} + \delta_{2,j} + \delta_{3,j}
    \end{align*}
    Notice that \(\delta_{1,j} + \delta_{3,j} = (\widehat\alpha_j - \bar\alpha_j)'\E_n[p_j(X)(1 - D/\widehat\pi_j(Z))Z]\).
    By the first order conditions for \(\widehat\gamma_j\) we have that 
    \[
        |\E_n[p_j(X)\{Z_l - DZ_l/\widehat\pi_j(Z)\}]| \leq \lambda_{\gamma,j}\;\forall l = 1,\dots,d_z \implies \|\E_n[p_j(X)\{Z_l - DZ_l/\widehat\pi_j(Z)\}]\|_\infty \leq \lambda_{\gamma, j}
    .\] 
    Applying Hölder's inequality to \(\delta_{1,j} + \delta_{3,j}\) then gives us that on the event \(\Omega_{k,2}\)
     \[
        |\delta_{1,j} + \delta_{3,j}| \leq \|\widehat\alpha_j-\bar\alpha_j\|_1\lambda_{\gamma,j} \leq \|\widehat\alpha_j -\bar\alpha_j\|\bar\lambda_k
    .\]
    By \Cref{lemma:nonasymptotic-outcome} on the event \(\bigcap_{m=1}^6\Omega_{k,m}\) and under the conditions of \Cref{lemma:nonasymptotic-means}, \(\|\widehat\alpha_j - \bar\alpha_j\| \leq M_1s_k\bar\lambda_k\) where \(M_1\) is a constant that does not depend on  \(k\). So
    \begin{equation}
        \label{eq:nonasymptotic-mean-step1}\tag{M.1}
        |\delta_{1,j} + \delta_{3,j}| \leq M_1s_k\bar\lambda_k^2 
    \end{equation}
    
    \emph{Step 2: Use Outcome Regression Score Domination to Bound \(\delta_{2,j}\).}
    Now deal with the term \(\delta_{2,j}\). By first order Taylor expansion, for some  \(u\in (0,1)\)
    \begin{align*}
        \delta_{2,j} &= -(\widehat\gamma_j - \bar\gamma_j)'\E_n[p_j(X)D\{Y - \bar m_j(Z)\}e^{-\bar\gamma_j'Z}Z] \\
                 &\;\;\;+ (\widehat\gamma_j-\bar\gamma_j)'\E_n[p_j(X)D\{Y-\bar m_j(Z)\}e^{-u\widehat\gamma_j'Z -(1-u)\bar\gamma_j'Z}ZZ'](\widehat\gamma_j - \bar\gamma_j)/2 \\
                 &:= \delta_{21, j} +\delta_{22, j}
    \end{align*}
    In the event \(\Omega_{k,1}\cap\Omega_{k,2}\cap\Omega_{k,3}\cap\Omega_{k,4}\) we have by score domination of the linear outcome regression model and \Cref{lemma:nonasymptotic-logit} that \(\delta_{21} \leq M_0s_k\bar\lambda_k^2\).

    The term \(\delta_{22,j}\) is second order. On the event \(\Omega_{k,0} \cap\Omega_{k,1}\) where \(\|\widehat\gamma_j - \bar\gamma_j\|_1 \leq M_0s_k\bar\lambda_k \leq M_0\eta/C_0\) it can be bounded with 
     \begin{align*}
        \delta_{22,j} 
        &\leq e^{C_0\|\widehat\gamma_j - \bar\gamma_j\|_1}\E_n[p_j(X)De^{-\bar\gamma_j'Z}|Y - \bar m_j(Z)|\{\widehat\gamma_j'Z - \bar\gamma_j'Z\}^2] \\
        &\leq e^{M_0\eta}\E_n[p_j(X)De^{-\bar\gamma_j'Z}|Y-\bar m_j(Z)|\{\widehat\gamma_j'Z - \bar\gamma_j'Z\}^2]
    .\end{align*}
    This in turn is bounded in a few steps. First note on the event \(\Omega_{k,7}\)     
    \[
        (\E_n - \E)[p_j(X)De^{-\bar\gamma_j'Z}|Y -\bar m_j(Z)|\{\widehat\gamma_j'Z - \bar\gamma_j'Z\}^2] \leq \bar\lambda_k\|\widehat\gamma_j - \bar\lambda_j\|_1^2
    .\] 
    By \Cref{assm:logistic-model-convergence} we have that \(G_0^2E[D|Y-\bar m_j(Z)|\mid Z] \leq G_1^2/G_0 + G_0\) so that,
    \[
        \E[p_j(X)De^{-\bar\gamma_j'Z}|Y - \bar m_j(Z)|\{\widehat\gamma_j'Z - \bar\gamma_j'Z\}^2] \leq (G_1^2/G_0 + G_0)\E[p_j(X)De^{-\bar\gamma_j'Z}\{\widehat\gamma_j'Z - \bar\gamma_j'Z\}^2]
    .\] 
    On the event \(\Omega_{k,6}\) we have that 
    \[
        (\E_n - \E)[p_j(X)De^{-\bar\gamma_j'Z}\{\widehat\gamma_j'Z-\bar\gamma_j'Z\}^2] \leq \bar\lambda_k\|\widehat\gamma_j - \bar\gamma_j\|_1
    .\] 
    Putting these all together gives 
    \begin{equation}
        \label{eq:nonasymptotic-mean-step2.1}\tag{M.2}
        \begin{split}
            \E_n[p_j(X)De^{-\bar\gamma_j'Z}|Y&-\bar m_j(Z)|\{\widehat\gamma_j'Z - \bar\gamma_j'Z\}^2] \\
            &\leq \bar\lambda_k\|\widehat\gamma_j - \bar\gamma_j\|_1^2 + (G_1^2/G_0 + G_0)\bar\lambda_k\|\widehat\gamma_j - \bar\gamma_j\|_1^2 \\
            &\;\;\; + (G_1^2/G_0 + G_0)\E_n[p_j(X)De^{-\bar\gamma'Z}\{\widehat\gamma_j'Z - \bar\gamma_j'Z\}]
        \end{split}
    \end{equation}
    To bound \eqref{eq:nonasymptotic-mean-step2.1} note again that in the event \(\Omega_{k,1}\cap\Omega_{k,2}\), \(\|\widehat\gamma_j- \bar\gamma_j\|_1 \leq M_0s_k\bar\lambda_k\) and that using by \eqref{eq:nonasymptotic-outcome-step2.2} in Online Appendix~\Cref{lemma:nonasymptotic-outcome}:
    \[
        \E_n[p_j(X)De^{-\bar\gamma_j'Z}\{\widehat\gamma_j'Z - \bar\gamma_j'Z\}^2] \leq e^{-M_0\eta}M_0s_k\bar\lambda_k^2
    .\] 
    Plugging these into \eqref{eq:nonasymptotic-mean-step2.1} gives 
    \begin{equation}
        \label{eq:nonasymptotic-mean-step2.2}\tag{M.3}
        \delta_{22,j} \leq e^{M_0\eta}M_0^2s_k^2\bar\lambda_k^3 + e^{M_0\eta}(G_1^2/G_0 + G_0)M_0^2s_k^2\bar\lambda_k^3 + (G_1^2/G_0 + G_0)M_0s_k\bar\lambda_k^2
    \end{equation}
    so that in total \(\delta_{2,j} = \delta_{21,j} + \delta_{22,j}\) is bouned
    \begin{equation}
        \label{eq:nonasymptotic-mean-step2}\tag{M.4}
        \delta_{2,j} \leq M_0s_k(G_1^2/G_0 + G_0 + 1)\bar\lambda_k^2  + e^{M_0\eta}M_0^2s_k^2(G_1^2/G_0 + G_0 + 1)\bar\lambda_k^3
    \end{equation}

    \emph{Step 3: Combine Terms.}
    Putting this together yields
    \begin{equation}
        \label{eq:nonasymptotic-mean-step3.1}\tag{M.5}
        \begin{split}
            |\delta_{1,j} +\delta_{2,j} + \delta_{3,j}| 
            &\leq \{M_1 + M_0(G_1^2/G_0 + G_0 + 1)\}s_k\bar\lambda_k^2 \\
            &\;\;\;+ e^{M_0\eta}(G_1^2/G_0 + G_0)M_0^2s_k^2\bar\lambda_k^3 
        \end{split}
    \end{equation}
    Use the fact that \(s_k\bar\lambda_k \leq \eta <1\) to simplify the last term of this expression 
      \begin{equation}
        \label{eq:nonasymptotic-mean-step3}\tag{M.6}
        \begin{split}
            |\delta_{1,j} +\delta_{2,j} + \delta_{3,j}| 
            &\leq \{M_1 + M_0(G_1^2/G_0 + G_0 + 1)\}s_k\bar\lambda_k^2 \\
            &\;\;\;+ e^{M_0\eta}(G_1^2/G_0 + G_0)M_0^2 s_k\bar\lambda_k
        \end{split}
    \end{equation}
    This gives the result \eqref{eq:nonasymptotic-means} after taking \(M_2 = M_1 + M_0(G_1^/G_0 + G_0 + 1) + e^{M_0\eta}(G_1^2/G_0 + G_0)M_0^2\).
 
\end{proof}

\begin{lemma}[Nonasymptotic Bounds for Variance Estimation]
    \label{lemma:nonasymptotic-variance}
    Suppose that \Cref{assm:logistic-model-convergence} hold, \(\xi_0  > (c_0 + 1)/(c_0 - 1)\), and \(2C_0\nu_0^{-2}s_k\bar\lambda_k \leq \eta < 1\). In addition, assume there is a constant \(c > 0\) such that  \(\lambda_{\alpha, j}/\lambda_{\gamma, j} \geq c\) for all \(j\leq k\). Then, under the event \(\bigcap_{m=1}^7 \Omega_{k,m}\), there is a constant  \(M_3\) that does not depend on  \(k\) such that
    \begin{equation}
        \label{eq:nonasymptotic-variance}
        \max_{1\leq j\leq k} \E_n[p_j^2(X)(Y(\widehat\pi_j,\widehat m_j) - Y(\bar\pi_j,\bar m_j))^2] \leq M_3\xi_{k,\infty}^2s_k^2\bar\lambda_k^2
    \end{equation}
\end{lemma}
\begin{proof}   
    We show the bound holds for each \(j=1,\dots,k\). We start by decomposing
    \begin{equation*}
        \begin{split}
        p_j(X)(Y(\widehat\pi_j, \widehat m_j) - Y(\bar\pi_j,\bar m_j)) 
        &= p_j(X)\{\widehat m_j(Z) - \bar m_j(Z)\}\bigg(1 - \frac{D}{\bar\pi_j(X)}\bigg)  \\
        &\;\;+ p_j(X)D\{Y -  \bar m_j(Z)\}\bigg(\frac{1}{\widehat\pi_j(Z)} - \frac{1}{\bar\pi_j(Z)} \bigg)  \\
        &\;\;+ p_j(X)\{\widehat m_j(Z) - \bar m_j(Z)\}\bigg(\frac{D}{\bar\pi_j(Z)} - \frac{D}{\widehat\pi_j(Z)}\bigg)  \\
        &:= \tilde\delta_{1,j} + \tilde\delta_{2,j} + \tilde\delta_{3,j} 
        \end{split}
    \end{equation*}
    We will use the fact that \((a+b+c)^2 \leq 4a^2 + 4b^2 + 4c^2\) to bound
    \begin{equation}
        \label{eq:nonasymptotic-variance-step0}\tag{V.1}
        \E_n[p_j^2(X)(Y(\widehat\pi_j,\widehat m_j) - Y(\bar\pi_j,\bar m_j))^2]  \leq 4\E_n[\tilde\delta_{1,j}^2] + 4\E_n[\tilde\delta_{2,j}^2] + 4\E_n[\tilde\delta_{3,j}^2]
    .\end{equation} 

    To bound \(\E_n[\tilde\delta_{2,j}]\) use the mean value equation \eqref{eq:nonasymptotic-outcome-mvt} in Online Appendix \Cref{lemma:nonasymptotic-outcome} and the lower bound on \(\bar g_j(z)\) from \Cref{assm:logistic-model-convergence}
    \begin{align}
        \E_n[\tilde\delta_{2,j}^2]  
        &= \E_n[p_j^2(X)D\{Y - \bar m_j(Z)\}^2\{\widehat\pi_j^{-1}(Z) - \bar\pi_j^{-1}(Z)\}^2]\nonumber \\ 
        &\leq \xi_{k,\infty}e^{-B_0}\left(1 + e^{C_0\|\widehat\gamma_j -\bar\gamma_j\|_1}\right)^2\E_n[p_j(X)De^{-\bar\gamma_j'Z}\{Y - \bar m_j(Z)\}^2\{\widehat g_j(Z) - \bar g_j(Z)\}^2 ]\nonumber
        \intertext{Applying \eqref{eq:nonasymptotic-outcome-step2.4} in Online Appendix~\Cref{lemma:nonasymptotic-outcome}, Online Appendix~\Cref{lemma:nonasymptotic-logit}, and \(s_k\bar\lambda_k \leq \eta < 1\) there is a constant \(\tilde M_1\) that does not depend on  \(k\) such that in the event  \(\bigcap_{m=1}^7 \Omega_{k,m}\) this is bounded}
        \label{eq:nonasymptotic-variance-delta2}\tag{V.2}
        &\leq \tilde M_1 \xi_{k,\infty}s_k\bar\lambda_k^2
    \end{align}
    To bound \(\E_n[\tilde\delta_{3,j}]\) write  \(\widehat \pi_j^{-1}(Z) - \bar\pi_j^{-1}(Z) = e^{-\bar\gamma_j'Z}\{e^{-\widehat\gamma_j'Z + \bar\gamma_j'Z}-1\}\) and use the lower bound on \(\bar g_j(z)\) from \Cref{assm:logistic-model-convergence}:
    \begin{align}
        \E_n[\tilde\delta_{3,j}^2] 
        &= \E_n[p_j^2(X)D\{\widehat m_j(Z) - \bar m_j(Z)\}^2\{\widehat\pi_j^{-1}(Z)-\bar\pi_j^{-1}(Z)\}^2] \nonumber \\
        &\leq \xi_{k,\infty}e^{-B_0}\left(1 + e^{C_0\|\widehat\gamma_j - \bar\gamma_j\|_1}\right)^2\E_n[p_j(X)e^{-\bar\gamma_j'Z}\{\widehat m_j(Z) - \bar m_j(Z)\}^2]\nonumber
        \intertext{Applying Online Appendix \Cref{lemma:nonasymptotic-outcome}, there is a constant \(\tilde M_2\) that does not depend on  \(k\) such that on the event \(\bigcap_{m=1}^6 \Omega_{k,m}\) this is bounded}
        \label{eq:nonasymptotic-variance-delta3}\tag{V.3}
        &\leq \tilde M_2\xi_{k,\infty} s_k \bar\lambda_k^2
    \end{align}
    Finally, to bound \(\E_n[\tilde\delta_{1,j}^2]\) again use the lower bound on \(\bar g_j(z)\) and decompose 
    \begin{align}
        \E_n[\tilde\delta_{1,j}^2] 
        &= \E_n[p_j^2(X)\{\widehat m(z) - \bar m(z)\}^2\{1 - D/\bar\pi_j(Z)\}^2] \nonumber \\
        &\leq \xi_{k,\infty}^2(1+e^{-B_0})^2\E_n[\{\widehat m_j(Z) - \bar m_j(Z)\}^2] \nonumber\\
        &\leq \xi_{k,\infty}^2(1+e^{-B_0})^2C_0^2\|\widehat\alpha_j - \bar\alpha_j\|_1^2\nonumber
        \intertext{Again on the event \(\bigcap_{m=1}^6 \Omega_{k,m}\) apply Online Appendix \Cref{lemma:nonasymptotic-outcome} this is bounded, for some constant \(\tilde M_3\) that does not depend on  \(k\) by}
        \label{eq:nonasymptotic-variance-delta1}\tag{V.4}
        &\leq \tilde M_3 \xi_{k,\infty}^2s_k^2\bar\lambda_k^2
    \end{align}
    The result \eqref{eq:nonasymptotic-variance} follows by collecting \eqref{eq:nonasymptotic-variance-step0}-\eqref{eq:nonasymptotic-variance-delta1}.
\end{proof}

\subsection{Proofs of Main Second Stage Results}%
\label{subsec:second-stage-main-proofs}

The proofs for \Cref{sec:first-stage} closely follow those of \citet{BCCK-2015} with some modifications to deal with the various error terms. They also rely on some additional second stage results proved in Online \Cref{sec:additional-second-stage} .

\subsubsection*{Proof of \Cref{thm:pointwise-normality}}
Equation \eqref{eq:thm-pointwise-normality-2} follows from applying \eqref{eq:thm-pointwise-normality-1} with \(\alpha = p(x)/\|p(x)\|\) and \eqref{eq:thm-pointwise-normality-3} follows from \eqref{eq:thm-pointwise-normality-2}. So it suffices to prove \eqref{eq:thm-pointwise-normality-1}.

For any \(\alpha\in S^{k-1}\),  \(1 \lesssim \|\alpha'\Omega^{1/2}\|\) because of the conditional variance of \(\bar\eps_j^2\) is bounded from below and from above and under the positive semidefinite ranking
\[
    \Omega \geq \Omega_0 \geq \underline\sigma^2Q^{-1}
.\] 
Moreover, by condition (ii) of the theorem and \Cref{lemma:pointwise-linearization}, \(R_{1n}(\alpha) = o_p(1)\). So we can write 
\begin{align*}
    \sqrt{n}\alpha'(\widehat\beta -\beta) 
    &=\frac{\sqrt{n}\alpha'}{\|\alpha'\Omega^{1/2}\|}\mathbb{G}_n[p^k(x)\circ(\eps^k + r_k)] + o_p(1)  \\
    &= \sum_{i=1}^n \frac{\alpha'}{\sqrt{n}\|\alpha'\Omega^{1/2}\|}\{p^k(x)\circ(\eps^k + r_k)\}
.\end{align*}
Goal will be to verify Lindberg's condition for the CLT. Throughout the rest of the proof, it will be helpful to make the following notations. First, for any vector \(a = (a_1,\dots,a_k)'\in S^{k-1}\), let \(|a| = (|a_1|,\dots,|a_k|)'\) and note that \(|a|\in S^{k-1}\) as well:
\begin{align*}
    \tilde\alpha_n' = \frac{\alpha'}{\sqrt{n}\|\alpha_n'\Omega^{1/2}\|} ,\;\;\;\omega_{n} := |\tilde\alpha|'p^k(x),\andbox \bar\eps_k := \sup_{1\leq j\leq k}|\eps_j|
\end{align*}
Now, by the definition of \(\Omega\) we have that
\[
    \Var\left(\sum_{i=1}^n \frac{\alpha'}{\sqrt{n}\|\alpha'\Omega^{1/2}\|}\{p^k(x)\circ(\eps^k + r_k)\}\right)  = 1
.\]
Second for each \(\delta > 0\)
\begin{align*}
    \sum_{i=1}^n\,&\E\left[(\tilde\alpha_n'\{p^k(x)\circ(\eps^k + r_k)\})^2 \bm{1}\left\{|\tilde\alpha_n'\{p^k(x)\circ(\eps^k + r_k)\}| > \delta \right\} \right] \\ 
    &\leq \sum_{i=1}^n \E\left[\omega_{n}^2\E\left[\bar\eps_k^2\,\bm{1}\{|\omega_{n}||\bar\eps_k + \ell_kc_k| > \delta\} \mid X = x\right]\right]\numberthis\label{eq:intermediate-pointwise-normality-1}
\end{align*}
\begin{tcolorbox}
    What we are using here is the following. Suppose \(\alpha\) is a nonrandom vector in  \(\SR^k\), \(a\) is a  (positive) random vector in \(\SR^k\) and \(b\) is a random vector in \( \SR^k\). Then, 
    \begin{equation}
        \label{eq:hadamard-rough-bound}
        \{\alpha'(a\circ b)\}  = \sum_{j=1}^k \alpha_j a_j b_j \leq \|b\|_\infty \sum_{j=1}^k |\alpha_j|a_j  = (|\alpha|'a) \|b\|_\infty
    .\end{equation} 
\end{tcolorbox}
To bound the right hand side of \eqref{eq:intermediate-pointwise-normality-1} use the fact that \(1 \lesssim \|\alpha'\Omega^{1/2}\|\) because \(1\lesssim \underline\sigma^2\) and 
\[
    \Omega \geq \Omega_0 \geq \underline\sigma^2 Q^{-1}
\]
in the positive semidefinite sense. Using these two we have 
\[
    n\E|\omega_n|^2 \leq  \E[(|\alpha|'p^k(x))^2]/(\alpha'\Omega\alpha) \lesssim 1
.\] 
\begin{tcolorbox}
    By the bounded eigenvalue condition and using the trace operator:
    \[\E[(|\alpha|p^k(x))^2] = \text{trace}(\E[|\alpha|'p^k(x)'p^k(x)|\alpha|])=  |\alpha|'Q||\alpha| \lesssim \|\alpha\|=1 \] 
\end{tcolorbox}
Further note, \(|\omega_{ni}| \lesssim \frac{\xi_k}{\sqrt{n}}\). Using \((a+b)^2 \leq 2a^2 + 2b^2\), the right hand side of \eqref{eq:intermediate-pointwise-normality-1} is bounded by 
\begin{align*}
    2n\E[|\omega_n|^2\bar\eps_k^2\,\bm{1}\{|\bar\eps_k| + \ell_kc_k > \delta/|\omega_{n}|\}] + 2n\E[|\omega_n|^2\ell_k^2c_k^2\bm{1}\{|\bar\eps_k| + \ell_kc_k > \delta/|\omega_n|\}] 
\end{align*}
and both terms converge to zero. Indeed, to bound the first term note that, for some \(c > 0\):
\begin{align*}
    2n\E[|\omega_n|^2\bar\eps_k^2\,\bm{1}\{|\bar\eps_k| + \ell_kc_k > \delta/|\omega_{n}|\}]  
    &\lesssim n\E[|\omega_n|^2]\sup_{x\in\calX}\E[\bar\eps_k^2\bm{1}\{\bar\eps_k^2+ \ell_kc_k > c\delta\sqrt{n}/\xi_k\}\mid X= x ] \\
        &= o(1)
\end{align*}
where here we use the first part of \Cref{assm:second-stage-assumptions}(iv). To show the second term converges to zero, follow the same steps as for the first term, but apply the second part of \Cref{assm:second-stage-assumptions}(iv).

\subsubsection*{Proof of \Cref{thm:strong-approximation}}
We apply Yurinskii's coupling lemma \citep{pollard_2001}
\begin{tcolorbox}[title = Yurinskii's Coupling Lemma]
    Let \(\xi_1,\dots, \xi_n\) be independent random \(k\)-vectors with  \(\E[\xi_i] = 0\) and \(\beta := \sum_{i=1}^n \E[\|\xi_i\|^3]\) finite. Let \(S := \xi_1 + \dots + \xi_n\). For each \(\delta > 0\) there exists a random vector  \(T\) with a  \(N(0,\var(S))\) distribution such that
     \begin{equation}
        \label{eq:YC}
        \P(|S - T| > 3\delta) \leq C_0B\left(1 + \frac{|\log(1/B)|}{k}\right)\;\;\;\text{where }B:= \beta k\delta^{-3}\tag{YC}
    \end{equation}
    for some universal constant \(C_0\).
\end{tcolorbox}
In order to apply the coupling, we want to consider a first order approximation to the estimator
\[
    \frac{1}{\sqrt{n}}\sum_{i=1}^n \xi_i,\;\;\;\zeta_i = \Omega^{-1/2}p^k(x)\circ(\eps^k + r_k) 
.\] 
When \(\bar R_{2n} = o_p(a_n^{-1})\) a similar argument can be used with  \(\zeta_i = \Omega^{-1/2}p^k(x)\circ(\eps^k + r_k)\) replaced with \(\Omega^{-1/2}p^k(x)\circ \eps^k\). As before, the eigenvalues of  \(\Omega\) are bounded away from zero, therefore
 \begin{align*}
    \E\|\zeta_i\|^3 
    &\lesssim \E[\|p^k(x)\circ(\eps^k(x) + r_k)\|^3] \\
    &\lesssim \E[\|p^k(x)\|^3(|\bar\eps_k|^3 + |r_k|^3)]\\ 
    &\lesssim \E[\|p^k(x)\|^3](\bar\sigma_k^3 + \ell_k^3c_k^3) \\
    &\lesssim \E[\|p^k(x)\|^3]\xi_k(\bar\sigma_k^3 + \ell_k^3c_k^3) \\
    &\lesssim k\xi_k(\bar\sigma_k^3 + \ell_k^3c_k^3)
\end{align*}
Therefore, by Yurinskii's coupling lemma \eqref{eq:YC}, for each \(\delta > 0\),
 \begin{align*}
    \Pr\left\{\|\sum_{i=1}^n \zeta_i/\sqrt{n} - \calN_k\| > 3\delta a_n^{-1}\right\} 
    &\lesssim \frac{nk^2\xi_k(\bar\sigma_m^3 + \ell_k^3c_k^3)}{(\delta a_n^{-1}\sqrt{n})^3}\left(1 + \frac{\log(k^3\xi_k(\bar\sigma_k^3 + \ell_k^3c_k^3))}{k} \right) \\
    &\lesssim \frac{a_n^3k^2\xi_k(\bar\sigma_k^3 + \ell_k^3c_k^3)}{\delta^3n^{1/2}}\left(1 + \frac{\log n}{k}\right)   \to 0 
.\end{align*} 
because \(a_n^6k^2\xi_k(\bar\sigma_m^3 + \ell_k^3c_k^3)\log^2 n/n\to 0\). Using the first two results from \Cref{lemma:uniform-linearization}, 
\eqref{eq:uniform-linearization-1}-\eqref{eq:uniform-linearization-r1n-bound}, we obtain that
\[
    \|\sqrt{n}\alpha(x)'(\widehat\beta^k - \beta^k) - \alpha(x)'\Omega^{1/2}\calN_k\| \leq \|1/\sqrt{n}\sum_{i=1}^n \alpha(x)'\Omega^{1/2}\zeta_i - \alpha(x)'\Omega^{1/2}\calN_k\| + \bar R_{1n} = o_p(a_n^{-1})
.\] 
uniformly over \(x\in\calX\). Since  \(\|\alpha(x)'\Omega^{1/2}\|\) is bounded from below uniformly over \(x\in\calX\) we obtain the first statetment of \Cref{thm:uniform-convergence} from which the second statement directly follows.

Finally, under the assumption that  \(\sup_{x\in\calX}n^{1/2}|r(x)|/\|s(x)\| =o_p(a_n^{-1})\), 
\[
    \frac{\sqrt{np(x)'(\widehat\beta^k-\beta^k)}}{\|s(x)\|}  - \frac{\sqrt{n}(\widehat g(x) - g_0(x))}{\|s(x)\|} = o_p(a_n^{-1}) 
\] 
so that the third statement, \eqref{eq:uniform-normality-ghat} holds.

\subsubsection*{Proof of \Cref{thm:matrix-estimation}}
\begin{tcolorbox}[title = Preliminaries for Proof of \Cref{thm:matrix-estimation}]
\begin{lemma*}[Symmetrization]
    \label{lemma:symmetrization}
    Let \(Z_1,\dots,Z_n\) be independent stochastic processes with mean zero and let \(\eps_1,\dots,\eps_n\) be independent Rademacher random variables generated independetly of the data. Then 
    \begin{equation}
        \label{eq:symmetrization}\tag{SI}
        \E^*\Phi\bigg(\frac{1}{2}\big\|\sum_{i=1}^n \eps_iZ_i\big\|_\calF\bigg)
        \leq \E^*\Phi\bigg(\big\|\sum_{i=1}^n Z_i\big\|_\calF\bigg)
        \leq \E^*\Phi\bigg(2\big\|\eps_i(Z_i - \mu_i)\big\|_\calF\bigg)
    ,\end{equation}
    for every nondecreasing, convex \(\Phi:\SR\to\SR\) and arbitrary functions \(\mu_i:\calF\to\SR\).
\end{lemma*}
\tcblower 
For \(p\geq 1\) consider the Shatten norm \(S_p\) on symmetrix  \(k\times k\) matrices  \(Q\) defined by \(\|Q\|_{S_p} = (\sum_{j=1}^k |\lambda_j(Q)|^p)^{1/p}\) where \(\lambda_1(Q),\dots,\lambda_k(Q)\) are the eigenvalues of \(Q\). The case \(p=\infty\) recovers the operator norm and  \(p=2\) recovers the Frobenius norm.
\begin{lemma*}[Khinchin's Inequality for Matrices]
    \label{lemma:khincins-inequality}
    For symmetric \(k\times k\) matrices  \(Q_i\),  \(i = 1,\dots,n\), \(2 \leq p \leq \infty\), and an i.i.d sequence of Rademacher random variables \(\eps_1,\dots,\eps_n\) we have 
    \begin{equation}
        \label{eq:khinchin}\tag{KI-1}
        \bigg\|\left(\E_n[Q_i^2]\right)^{1/2}\bigg\|_{S_p} \leq \left(\E_\eps\|\mathbb{G}_n[\eps_iQ_i]\|_{S_p}^p\right)^{1/p} \leq C\sqrt{p}\left\|\left(\E_n[Q_i^2]\right)^{1/2}\right\|_{S_p}
    \end{equation}
    where \(C\) is an absolute constant. So, for  \(k \geq 2\), 
    \begin{equation}
        \label{eq:khinchin-2}\tag{KI-2}
        \E_\eps[\|\mathbb{G}_n[\eps_iQ_i]\|] \leq C\sqrt{\log k}\|(\E_n[Q_i^2])^{1/2}\|
    \end{equation}
    for some (possibly different) absolute constant \(C\).
\end{lemma*}
\end{tcolorbox}

We will establish consistent estimation of 
\begin{align*}
    \Sigma &= \E[\{p^k(x)\circ(\eps^k + r_k)\}\{p^k(x)\circ(\eps^k + r_k)\}']\intertext{using }
    \widehat\Sigma &=\E_n[\{p^k(x)\circ\widehat\eps^k\}\{p^k(x)\circ\widehat\eps^k\}']
\end{align*} 
Consistency of \(\widehat\Omega\) will then follow from the consistency of \(\widehat Q\) established by \Cref{lemma:matrix-lln}. To save notation, define the vectors
\begin{equation}
    \label{eq:hatYbarY}
    \widehat Y :=
    \begin{bmatrix} Y(\widehat \pi_1, \widehat m_1) \\ \vdots \\ Y(\widehat\pi_k, \widehat m_k) \end{bmatrix} \andbox\;
    \widehat Y :=
    \begin{bmatrix} Y(\widehat \pi_1, \widehat m_1) \\ \vdots \\ Y(\widehat\pi_k, \widehat m_k) \end{bmatrix}
\end{equation}
Also define \(\dot\eps^k := (\dot \eps^k_1, \dots, \dot \eps^k_k)\) so that \(\dot\eps^k_j := Y(\bar\pi_j, \bar m_j) - \widehat g(x)\). Ideally, we would like to use \(\dot\eps^k\) to estimate  \( \widehat\Sigma\), but we don't observe \(\dot\eps^k\). Define  \(\Delta := \widehat\eps^k - \dot\eps^k = \widehat Y^k - \bar Y^k\in \SR^k\).

Using this, we can decompose
\begin{equation}
    \label{eq:SigmaHat-decomp}
    \begin{split}
        \widehat\Sigma &= \E_n[\{p^k(x)\circ (\Delta + \dot\eps^k)\}\{p^k(x)\circ(\Delta + \dot\eps^k)\}] \\ 
                    &= \underbrace{\E_n[\{p^k(x)\circ\Delta\}\{p^k(x)\circ\Delta\}']}_{\Sigma_1} 
                    + \underbrace{\E_n[\{p^k(x)\circ\dot\eps^k\}\{p^k(x)\circ\Delta\}']}_{\Sigma_2}\\
                    &\;\;+ \underbrace{\E_n[\{p^k(x)\circ\Delta\}\{p^k(x)\circ\dot\eps^k\}']}_{\Sigma_3} +  \underbrace{\E_n[\{p^k(x)\circ \dot\eps^k\}\{p^k(x)\circ\dot\eps^k\}]}_{\Sigma_4}  
    \end{split}
\end{equation}
We first show that  \(\|\Sigma_4 - \Sigma\| \to_p 0\). This is nonstandard because of the Hadamard product.

\begin{lemma}[Psuedo-Variance Estimator Consistency]
    \label{lemma:psuedo-variance-estimator-consistency}
    Suppose \Cref{assm:second-stage-assumptions} and \Cref{assm:uniform-limit-theory} hold. Further, define \(v_n = \E[\max_{1\leq i \leq n}|\bar\eps_k|^2]^{1/2}\). In addition, assume that \(\bar R_{1n} + \bar R_{2n} \lesssim (\log k)^{1/2}\). Then, 
    \begin{align*}
        \|\widehat Q - Q\| &\lesssim_P \sqrt{\frac{\xi_k^2\log k}{n}} = o(1) \\
       \andbox \|\Sigma_4 - \Sigma\|&\lesssim_P (v_n \vee 1 + \ell_kc_k)\sqrt{\frac{\xi_k^2\log k}{n}}
    \end{align*}
\end{lemma}
\begin{proof}
    The first result is established by \Cref{lemma:matrix-lln} (Matrix LLN). Rest of proof will follow proof of Theorem 4.6 in \citet{BCCK-2015}. Like in \eqref{eq:SigmaHat-decomp} we can define \(\dot\Delta \equiv \dot\eps^k - \eps^k = g_0(x) - \widehat g(x)\)\footnote{It is useful to recall that \(\dot\eps^k = \bar Y^k - \widehat g(x)\) and \(\eps^k = \bar Y^k - g_0(x)\)} and decompose
    \begin{align*}
        \Sigma_4  &= \underbrace{\E_n[p^k(x)p^k(x)'\dot\Delta^2]}_{\Sigma_{41}}  + \underbrace{\E_n[\{p^k(x)\circ(\eps^k + r_k)\}\{p^k(x)\cdot\dot\Delta\}']}_{\Sigma_{42}}\\
        &\;\;+ \underbrace{\E_n[\{p^k(x)\cdot\dot\Delta\}\{p^k(x)\circ(\eps^k + r_k)\}']}_{\Sigma_{43}} +  \underbrace{\E_n[\{p^k(x)\circ (\eps^k + r_k)\}\{p^k(x)\circ(\eps^k + r_k)\}]}_{\Sigma_{44}}  
    \end{align*}
    The terms \(\Sigma_{41}, \Sigma_{42}\) and  \(\Sigma_{43}\) are simple to show are negligible.
    \begin{align*}
        \|\Sigma_{41} &+ \Sigma_{42} + \Sigma_{43}\|\\
        &\leq \|\E_n[\{p^k(x)'(\widehat\beta^k-\beta^k)\}p^k(x)p^k(x)']\|
        + \|\E_n[\{p^k(x)\circ(\eps^k + r_k)\}p^k(x)'\{p^k(x)'(\widehat\beta^k - \beta^k)\}]\| \\
        &\;\;+ \|\E_n[p^k(x)\{p^k(x)'(\widehat\beta^k - \beta^k)\}\{p^k(x)\circ(\eps^k + r_k)\}']\| \\
        &\leq \max_{1 \leq i \leq n}|p^k(x)(\widehat\beta^k - \beta^k)|^2\|\E_n[p^k(x)p^k(x)']\| \\ 
        &\;\;+ 2\max_{1\leq i\leq n} |\bar\eps_{k,i}| + |r_{k,i}|\max_{1\leq i\leq n}|p^k(x)'(\widehat\beta - \beta)| \|\E_n[p^k(x)p^k(x)']\|
    \end{align*}
    By \Cref{thm:uniform-convergence} \(|\max_{1\leq i\leq n} |p^k(x)'(\widehat\beta^k - \beta^k)| \lesssim_P \xi_k^2(\sqrt{\log k} + \bar R_{1n} + \bar R_{2n})^2/n\), by \Cref{assm:second-stage-assumptions} the approximation error is bounded \(\max_{1\leq i\leq n}|r_{k,i}| \leq \ell_kc_k\), by \Cref{assm:uniform-limit-theory} and Markov's inequality the errors are bounded \(\max_{1\leq i\leq n}|\bar\eps_{k,i}| \lesssim_p v_n^2\). Finally, by the first part of \Cref{lemma:psuedo-variance-estimator-consistency} \(\|\widehat Q\| \lesssim_P \|Q\| \lesssim 1\). Putting this all together with \(\bar R_{1n} + \bar R_{2n} \lesssim (\log k)^{1/2}\) and \(\xi_k^2 \log k/n \to 0\) gives
    \[
        \|\Sigma_{41} + \Sigma_{42} + \Sigma_{43}\| \lesssim_P (v_n \vee 1 + \ell_kc_k)\sqrt{\frac{\xi_k^2\log k}{n}}
    .\] 
    Next, we want to control \(\Sigma_{44} - \Sigma\). To do this, let \(\eta_1,\dots,\eta_n\) be independent Rademacher random variables generated independently from the data. Then for \(\eta = (\eta_1,\dots,\eta_n)\)
    \begin{align*}
        \E[\|\E_n[\{p^k(x)&\circ(\eps^k + r_k)\}\{p^k(x)\circ(\eps^k + r_k)\}'] - \Sigma\|] \\
        &\lesssim \E[\E_\eta[\E_n\|\eta\{p^k(x)\circ(\eps^k + r_k)\}\{p^k(x)\circ(\eps^k + r_k)\}'\|]] \\
        &\lesssim \sqrt{\frac{\log k}{n}} \E[(\|\E_n[\|p^k(x)\|^2(\bar\eps_k + r_k)^2\{p^k(x)\circ(\eps^k + r_k)\}\{p^k(x)\circ(\eps^k + r_k)\}' ]\|)^{1/2}]\\
        &\lesssim \sqrt{\frac{\xi_k^2\log k}{n}}\E[\max_{1\leq i \leq n} |\bar\eps_{k,i} + r_k|(\|\E_n[\{p^k(x)\circ(\eps^k + r_k)\}\{p^k(x)\circ(\eps^k + r_k)\}']\|)^{1/2}] \\
        &\leq \sqrt{\frac{\xi_k^2\log k}{n}}(\E[\max_{1\leq i\leq n}|\bar\eps_{k,i} + r_k|^2])^{1/2}\times(\E[\|\E_n[\{p^k(x)\circ(\eps^k +r_k)\}\{p^k(x)\circ(\eps^k + r_k)\}']\|])^{1/2}   
    \end{align*}
    where the first inequality holds from Symmetrization~\eqref{eq:symmetrization}, the second from Khinchin's inequality~\eqref{eq:khinchin}, the third by \(\max_{1\leq i\leq n}\|p^k(x)\| \leq \xi_k\) and the fourth by Cauchy-Schwarz inequality.
    
    Since for any positive numbers \(a,b\) and  \(R\),  \(a \leq R(a+b)^{1/2}\) implies \(a \leq R^2 + R\sqrt{b}\), the expression above and the triangle inequality yields
    \begin{align*}
        \E[\|\E_n[\{p^k(x)&\circ(\eps^k + r_k)\}\{p^k(x)\circ(\eps^k + r_k)\}'] - \Sigma\|] \\
        &\lesssim \frac{\xi_k^2\log k}{n}(v_n^2 + \ell_k^2c_k^2)  + \left(\frac{\xi_k^2\log k}{n} \{v_n^2 + \ell_k^2c_k^2\} \right)^{1/2}\|\Sigma\|^{1/2}
    \end{align*}
    and so, because \(\|\Sigma\|\lesssim 1\) and \((v_n^2 + \ell_k^2 c_k^2)\xi_k^2\log k/n\to 0\) we have
     \[
         \E[\|\E_n[\{p^k(x)\circ(\eps^k + r_k)\}\{p^k(x)\circ(\eps^k + r_k)\}'] - \Sigma\|] \lesssim (v_n \vee 1 + \ell_k c_k)\sqrt{\frac{\xi_k^2\log k}{n}}
     .\]
     The second result of \Cref{lemma:psuedo-variance-estimator-consistency} follows from Markov's inequality.
\end{proof}

Now, we need to take care of the terms 
\begin{align*}
    \Sigma_1 &= \E_n[\{p^k(x)\circ\Delta\}\{p^k(x)\circ\Delta\}'] \\
    \Sigma_2 &= \E_n[\{p^k(x)\circ\dot\eps^k\}\{p^k(x)\circ\Delta\}']\\
    \Sigma_3 &= \E_n[\{p^k(x)\circ\Delta\}\{p^k(x)\circ\dot\eps^k\}'] 
\end{align*}
where \(\Delta = \widehat Y^k - \bar Y^k\) and \(\dot\eps^k = \bar Y^k - \widehat g(x) = \widehat g(x) - g^k(x) + \eps^k\). To do so we will use \Cref{cond:variance-estimation}.

\begin{lemma}[Negligible Variance Bias]
    \label{lemma:variance-bias-negligible}
    Suppose that \Cref{cond:variance-estimation}, \Cref{assm:second-stage-assumptions} and \Cref{assm:uniform-limit-theory} hold. Then
    \[
        \|\Sigma_1 + \Sigma_2 + \Sigma_3\| = o_p(1)
    .\] 
\end{lemma}
\begin{proof}
    From \Cref{cond:variance-estimation}, the term \(\Sigma_1\) being negligible immediately follows from Cauchy-Schwarz. Notice that 
    \begin{align*}
        \|\Sigma_1\| &\leq k\sup_{\substack{1 \leq l \leq k \\ 1 \leq j \leq k}} |\E_n[p_l(X)(Y(\hat\pi_l,\hat m_l) - Y(\bar\pi_j, \bar m_j))p_l(X)(Y(\widehat\pi_l,\widehat m_l) - Y(\bar\pi_l, \bar m_l))]| \\
        &\leq k \sup_{1\leq l \leq k} (\E_n[p_j(X)^2(Y(\widehat\pi_j, \widehat m_j)- Y(\bar\pi_j,\bar m_j))^2])^{1/2}\sup_{1\leq j\leq k}(\E_n[p_j(X)^2(Y(\widehat\pi_j, \widehat m_j)- Y(\bar\pi_j,\bar m_j))^2])^{1/2} \\
        &= o_p(1)
    .\end{align*} 
    To see that \(\Sigma_2\) is negligible notice that
    \begin{align*}
        \|\Sigma_2\| &\leq k \sup_{\substack{1 \leq l \leq k \\ 1 \leq j \leq k}}\E_n[p_l(X)(\eps_l + p^k(x)'(\widehat\beta^k - \beta^k))p_j(X)(Y(\widehat\pi_j, \widehat m_j) - Y(\bar\pi_j,\bar m_j))] \\
                     &\leq k\sup_{1\leq l\leq k}\E_n[p_l(X)^2(\eps_l + p^k(x)'(\widehat\beta - \beta))^2]^{1/2}\E_n[p_j(X)^2(Y(\widehat\pi_j,\widehat m_j) - Y(\bar\pi_j, \bar m_j))^2]^{1/2} \\ 
                     &\leq \xi_{k,\infty}(\max_{1\leq i\leq n} |\bar\eps_k| + \max_{1\leq i \leq n}p^k(x)'(\widehat\beta - \beta))\E_n[p_j(X)^2(Y(\widehat\pi_j, \widehat m_j) - Y(\bar\pi_j,\bar m_j))^2]^{1/2} \\
        \intertext{Applying \Cref{assm:uniform-limit-theory} and \Cref{thm:uniform-convergence} gives }
                     &\lesssim_P k\xi_{k,\infty}n^{1/m}\E_[p_j(X)^2Y(\widehat\pi_j, \widehat m_j) - Y(\bar\pi_j, \bar m_j))^2]^{1/2} = o_p(1)
    \end{align*}
    where the final line is via \Cref{cond:variance-estimation}. Showing negligibility of \(\Sigma_3\) follows the same steps. 
\end{proof}

\subsubsection*{Proof of \Cref{thm:uniform-confidence-bands}}
Follows from the exact same steps as Theorem 3.5 in \citet{SC-2020} after establishing strong approximation by a gaussian process as in \Cref{thm:strong-approximation} and consistent variance estimation as in \Cref{thm:matrix-estimation}.

\newpage \setcounter{page}{0} \thispagestyle{plain}
\vspace*{\fill}
\Large Online Appendix \normalsize
\vspace*{\fill}

\newpage
\section{Supporting Lemmas for First Stage}%
\label{sec:proofs-lemmas}

Here we provide supporting lemmas and their proofs. We start off with non-asymptotic bounds for first stage parameters and means.

\subsection{Nonasymptotic Bounds for the First Stage}
\label{subsec:first-stage-nonasymptotic-bounds}
The nonasymptotic bounds for the first stage will depend on certain events. In \Cref{subsec:first-stage-probability-bounds} we will show that under \Cref{assm:logistic-model-convergence} these events happen with probability approaching one.  To control sparsity, define \(\calS_{\gamma,j} := \{j: \bar\alpha_j \neq 0\}\), \(\calS_{\alpha,j} := \{j: \bar\alpha_j \neq 0\}\). Recall \(s_k := \max_{1\leq j\leq k}\{|\calS_{\gamma, j}| \vee |\calS_{\alpha, j}|\}\). Define the scores 
\begin{equation}
    \label{eq:score-definition}
    \begin{split}
        S_{\gamma, j} &:= \E_n[U_{\gamma,j}Z] \\
        S_{\alpha, j} &:= \E_n[U_{\alpha,j}Z]
    \end{split}
\end{equation}
With these in mind, we will consider nonasymptotic bounds under the events:
\begin{equation}
    \label{eq:nonasymptotic-events-logit}
    \begin{split}
        \Omega_{k,1} &:= \{\lambda_{\gamma, j} \geq c_0\cdot\|S_{\gamma, j}\|_\infty, \forall j \leq k\}    \\
        \Omega_{k,2} &:= \{\lambda_{\gamma, j} \leq \bar\lambda_k, \forall j\leq k\} \\
    \end{split}
\end{equation}    
Following \citet{CS-2021}, the first event is referred to as ``score domination'' while the second event is referred to as ``penalty majorization''.

Bounds will be established on the \(\ell_1\) convergence rate of the estimated coefficient vector as well as on the symmetrized Bregman divergences,  \(D^\ddag_{\gamma, j}(\widehat\gamma_j, \bar\gamma_j)\) and \(D^\ddag_{\alpha,j}(\widehat\alpha_j, \bar\alpha_j; \gamma_j)\), defined by
\begin{equation}
    \label{eq:bregman-divergences}
    \begin{split}
        D_{\gamma,j}^\ddag(\widehat\gamma_j, \bar\gamma_j)
        &:= \E_n\left[p_j(X)D\{e^{-\widehat\gamma_j'Z} - e^{-\bar\gamma_j'Z}\}\{\bar\gamma_j'Z - \widehat\gamma_j'Z\}\right], \\
        D_{\alpha,j}^\ddag(\widehat\alpha_j, \bar\alpha_j; \widehat\gamma)
        &:= \E_n\left[p_j(X)De^{-\widehat\gamma_j'Z}(\bar\alpha_j'Z - \widehat\alpha_j'Z)^2\right].
    \end{split}
\end{equation}
We refer readers to discussion in \citet{Tan-2017} for details and motiviation. For now it suffices to note that the Bregman divergence is the error resulting from approximating the non-penalized loss function at the estimated value with a first order Taylor expansion of the non-penalizd loss function at the true values. Because our loss functions are convex, these errors will always be positive. Bounds on the Bregman divergence help directly control second order terms in the remainder of \eqref{eq:taylor-expansion-gamma-alpha}.

\begin{lemma}[Nonasymptotic Bounds for Logistic Model]
    \label{lemma:nonasymptotic-logit}
    Suppose that \Cref{assm:logistic-model-convergence} holds with \(\xi_0 > (c_0 + 1)/(c_0 - 1)\) and \(2C_0\nu_0^{-2}s_k\bar\lambda_k \leq \eta < 1\). Then, under the events \(\Omega_{k,1}\cap\Omega_{k,2}\) defined in \eqref{eq:nonasymptotic-events-logit}, there exists a finite constant \(M_0\) that does not depend on \(k\) such that  
    \begin{equation}
        \label{eq:nonasymptotic-logit}
        \max_{1\leq j\leq k}D^\ddag(\bar g, \widehat g) \leq M_0 s_k \bar\lambda_k^2 \andbox \max_{1\leq j\leq k}\|\widehat\gamma_j - \bar\gamma_j\|_1 \leq M_0 s_k \bar\lambda_k   
    \end{equation}
\end{lemma}
\begin{proof}
    We show that the bound of \eqref{eq:nonasymptotic-logit} holds for each \(j = 1,\dots,k\). For any \(\gamma\in\SR^d\) define \(\tilde\ell_j(\gamma) := \E_n[p_j(X)\{De^{-\gamma'Z} + (1-D)\gamma'Z\}]\). By optimality of \(\widehat\gamma_j\) we must have, for any \(u \in (0,1]\):
    \[
        \tilde\ell_j\left(\widehat\gamma_j\right) + \lambda_{\gamma, j}\|\widehat\gamma_j\|_1 \leq \tilde\ell\left((1-u)\widehat\gamma_j + u\bar\gamma_j\right) + \lambda_{\gamma, j}\|(1-u)\widehat\gamma_j + u\bar\gamma_j\|_1
    .\] 
    Using convexity of the \(\ell_1\) norm  \(\|\cdot\|_1\), this gives after rearrangment
    \[
        \tilde\ell_j\left(\widehat\gamma_j\right)-\tilde\ell\left((1-u)\widehat\gamma_j + u\bar\gamma_j\right) + \lambda_{\gamma, j} u \|\widehat\gamma_j\|_1 \leq \lambda_{\gamma, j}u\|\bar\gamma_j\|_1
    .\] 
    Divide both sides by \(u\) and let \(u \to^+ 0\)
    \[
        \E_n[p_j(X)D\{e^{-\widehat\gamma'Z} + (1-D)\}\{\widehat\gamma_j'Z - \bar\gamma_j'Z\}] + \lambda_{\gamma, j}\|\widehat\gamma_j\|_1 \leq \lambda_{\gamma, j}\|\bar\gamma_j\|_1
    .\] 
    By direct calculation, we have that \(D^\ddag_{\gamma,j}(\widehat\gamma_j, \bar\gamma_j)\) from \eqref{eq:bregman-divergences} can be expressed
    \[
        D_{\gamma,j}^\ddag(\widehat\gamma_j, \bar\gamma_j) = \E_n[p_j(X)D\{e^{-\widehat\gamma'Z} + (1-D)\}\{\widehat\gamma_j'Z - \bar\gamma_j'Z\}]  -  \E_n[p_j(X)D\{e^{-\bar\gamma'Z} + (1-D)\}\{\widehat\gamma_j'Z - \bar\gamma_j'Z\}] 
    .\] 
    Combining the last two displays yields
    \begin{equation}
        \label{eq:nonasymptotic-logit-step1}\tag{L.1}
        D_{\gamma,j}^\ddag(\widehat\gamma_j, \bar\gamma_j) + \E_n[p_j(X)D\{e^{-\bar\gamma'Z} + (1-D)\}\{\widehat\gamma_j'Z - \bar\gamma_j'Z\}] +\lambda_{\gamma,j}\|\widehat\gamma_j\|_1 \leq \lambda_{\gamma, j}\|\bar\gamma_j\|_1
    \end{equation}
    In the event \(\Omega_{k,1}\) we have that 
     \begin{equation}
        \label{eq:nonasymptotic-logit-score-domination}\tag{L.2}
        |\E_n[p_j(X)D\{e^{-\bar\gamma'Z} + (1-D)\{\widehat\gamma'Z - \bar\gamma'Z\}\}] \leq c_0^{-1}\lambda_{\gamma, j}\|\widehat\gamma_j - \bar\gamma_j\|_1
    \end{equation}
    Combining \eqref{eq:nonasymptotic-logit-step1} and \eqref{eq:nonasymptotic-logit-score-domination} yields 
    \[
        D_{\gamma,j}^\ddag(\widehat\gamma_j, \bar\gamma_j) + \lambda_{\gamma, j}\|\widehat\gamma_j\|_1 \leq \lambda_{\gamma,j}\|\bar\gamma_j\| + c_0^{-1}\lambda_{\gamma, j}\|\widehat\gamma_j - \bar\gamma_j\|_1
    .\] 
    Expanding \(\|\gamma_j\|_1 = \sum_{l\in \calS_{\gamma, j}}|\gamma_l| + \sum_{l \not\in \calS_{\gamma, j}}|\gamma_l|\) for \(\gamma =\widehat\gamma_j,\bar\gamma_j\) and applying the triangle inequalities \(|\widehat\gamma_{j,l}| \geq |\bar\gamma_{j,l}| - |\widehat\gamma_{j,l} - \bar\gamma_{j,l}|\) for \(l \in\calS_{\gamma, j}\) and the equality  \( \widehat\gamma_{j,l} = \widehat\gamma_{j,l} - \bar\gamma_{j,l}\) gives
    \begin{align*}
        D_{\gamma,j}^\ddag(\widehat\gamma_j, \bar\gamma_j) + \lambda_{\gamma, j}\bigg\{\sum_{l\in\calS_{\gamma, j}}
        &|\bar\gamma_{j,l}| - \sum_{l\in\calS_{\gamma, j}} |\widehat\gamma_{j,l} - \bar\gamma_{j,l}| + \sum_{j\not\in\calS_{\gamma, j}}|\widehat\gamma_{j,l} - \bar\gamma_{j,l}|\bigg\} \\
        &\leq \lambda_{\gamma, j}\bigg\{\sum_{l\in\calS_{\gamma, j}}|\bar\gamma_{j,l}| + c_0^{-1}\sum_{l\in\calS_{\gamma, j}}|\widehat\gamma_{j,l} - \bar\gamma_{j,l}| + c_0^{-1}\sum_{j\not\in\calS_{\gamma, j}}|\widehat\gamma_{j,l}-\bar\gamma_{j,l}|\bigg\}
    \end{align*}
    Rearrange to get
    \[
        D_{\gamma, j}^\ddag(\widehat\gamma_j,\bar\gamma_j) + (1-c_0^{-1})\lambda_{\gamma, j}\sum_{l\not\in\calS_{\beta}}|\widehat\gamma_{j,l} - \bar\gamma_{j,l}| \leq (1+c_0)^{-1}\lambda_{\gamma,j}\sum_{l\in\calS_{\gamma, j}}|\widehat\gamma_{j,l} - \bar\gamma_{j,l}|
    .\] 
    Adding \((1-c_0^{-1})\lambda_{\gamma, j}\sum_{l\in\calS_{\gamma, j}}|\widehat\gamma_{j,l} -\bar\gamma_{j,l}|\) gives
    \begin{equation}
        \label{eq:nonasymptotic-logit-step3}\tag{L.3}
        D_{\gamma, j}^\ddag(\widehat\gamma_j, \bar\gamma_j) + (1-c_0^{-1})\|\widehat\gamma_j - \bar\gamma_j\|_1 \leq 2\lambda_{\gamma, j}\sum_{l\in\calS_{\gamma, j}}|\widehat\gamma_{j,l} - \bar\gamma_{j,l}|
    \end{equation}
    By Lemma 4 in Appendix V.3 of \citet{Tan-2017} we have that for \(\delta_j := \widehat\gamma_j - \bar\gamma_j\)
    \begin{equation}
        \label{eq:non-asymptotic-logit-step4}\tag{L.4}
        D_{\gamma, j}^\ddag(\widehat\gamma_j, \bar\gamma_j) \geq \frac{1 - e^{-C_0}\|\delta_j\|_1}{C_0\|\widehat\delta_j\|}\left(\delta_j'\tilde\Sigma_{\gamma, j}\delta_j\right)
    \end{equation}
    By \eqref{eq:nonasymptotic-logit-step3} and \(\xi_0 > (c_0 + 1)/(c_0 - 1)\) we have that \(\sum_{l\not\in\calS_{\gamma, j}}|\delta_{j,l}| \leq \xi_0\sum_{l\in\calS_{\gamma, j}}|\delta_{j,l}|\). Applying the empirical compatability condition from \Cref{assm:logistic-model-convergence} to \eqref{eq:nonasymptotic-logit-step3} then yields
    \begin{equation}
        \label{eq:nonasymptotic-logit-step5}\tag{L.5}
        D_{\gamma, j}^\ddag(\widehat\gamma_j,\bar\gamma_j) + (1- c_0^{-1})\lambda_{\gamma,j}\|\delta_j\|_1 \leq 2\lambda_{\gamma, j}\nu_0^{-1}|\calS_{\gamma,j}|^{1/2}(\delta_j'\tilde\Sigma_{\gamma, j}\delta_j)^{1/2}
    \end{equation}
    Combining \eqref{eq:non-asymptotic-logit-step4} and \eqref{eq:nonasymptotic-logit-step5} to get an upper bound on \((\delta_j'\tilde\Sigma\delta_j)^{1/2}\) gives 
    \[
        \nu_0\|\delta_j\|_2 \leq (\delta_j'\tilde\Sigma_{\gamma, j}\delta_j)^{1/2} \leq 2\lambda_{\gamma, j}\nu_0^{-1}|\calS_{\gamma, j}|^{1/2}\frac{C_0\|\delta_j\|_1}{1 - e^{-C_0\|\delta_j\|_1}} 
    .\] 
    Plugging the second bound into \eqref{eq:nonasymptotic-logit-step5} gives
    \[
        D_{\gamma, j}^\ddag(\widehat\gamma_j,\bar\gamma_j) + (1- c_0^{-1})\lambda_{\gamma, j}\|\delta_j\|_1\leq 2\lambda\sum_{l\in\calS_{\gamma, j}}|\delta_{j,l}| \leq 4\lambda_{\gamma, j}^2\nu_0^{-2}|\calS_{\gamma, j}|\frac{C_0\|\delta_j\|_1}{1- e^{-C_0\|\delta_j\|_1}} 
    .\]
    The second inequality and \(\sum_{l\not\in\calS_{\gamma, j}}|\delta_{j,l}| \leq \xi_0\sum_{l\in\calS_{\gamma,j}}|\delta_{j,l}|\) imply \(1-e^{-C_0\|\delta_j\|_1} \leq 2C_0\lambda_{\gamma, j}\nu_0^{-2}|\calS_{\gamma, j}| \leq \eta\) so,
    \[
        \frac{1 - e^{-C_0\|\delta_j\|_1}}{C_0\|\delta_j\|_1} = \int_0^1e^{-C_0\|\delta_j\|_1 u}\,du \geq e^{-C_0\|\delta_j\|_1} \geq 1 - \eta
    .\] 
    Combining the last two displays gives
    \begin{equation}
        \label{eq:nonasymptotic-logit-step6}\tag{L.6}
        D_{\gamma,j}^\ddag(\widehat\gamma_j,\bar\gamma_j) + (1-c_0^{-1})\lambda_{\gamma, j}\|\widehat\gamma_j - \bar\gamma_j\|_1 \leq 4\lambda_{\gamma, j}^2 \nu_0^{-2}(1-\eta)|\calS_{\gamma, j}|
    \end{equation}
    Applying \(\Omega_{k,2}\) to bound  \(\lambda_{\gamma, j} \leq \bar\lambda_k\) and noting that \(|\calS_{\gamma, j}| \leq s_k\) by definition gives \eqref{eq:nonasymptotic-logit} with \(M_0 = \frac{4\nu_0^{-1}(1-\eta)}{1-c_0^{-1}}\).
\end{proof}

For each \(j\), consider the matrices,
\begin{equation}
    \label{eq:additional-matrices-outcome}
    \begin{split}
        \tilde\Sigma_{\alpha, j} &:= \E_n[p_j(X)De^{-\bar\gamma_j'Z}(Y-\bar\alpha_j'Z)^2ZZ']  \\
        \tilde\Sigma_{\gamma, j} &:= \E_n[p_j(X)De^{-\bar\gamma_j'Z}ZZ'] \\
    \end{split}
\end{equation}
In addition define \(\Sigma_{\alpha, j} := \E\tilde\Sigma_{\alpha, j}\) and \(\Sigma_{\gamma, j} := \E\tilde\Sigma_{\gamma, j}\). For the outcome regression model, we will consider nonasymptotic bounds under the following additional events:
\begin{equation}
    \label{eq:nonasymptotic-events-outcome}
    \begin{split}
        \Omega_{k,3} &:= \{\lambda_{\alpha, j} \geq c_0\|S_{\alpha, j}\|_\infty, \forall j \leq k\}  \\
        \Omega_{k,4} &:= \{\lambda_{\alpha, j} \leq \bar\lambda_k, \forall j\leq k\}  \\
        \Omega_{k,5} &:= \{\|\tilde\Sigma_{\alpha,j}- \Sigma_{\alpha,j}\|_\infty \leq \bar\lambda_k, \forall j\leq k\} \\
        \Omega_{k,6} &:= \{\|\tilde\Sigma_{\gamma, j} - \Sigma_{\gamma, j}\|_\infty \leq \bar\lambda_k, \forall j\leq k\}
    \end{split}
\end{equation}
\begin{lemma}[Nonasymptotic Bounds for Linear Model]
    \label{lemma:nonasymptotic-outcome}
    Suppose that \Cref{assm:logistic-model-convergence} holds, \(\xi_0  > (c_0 + 1)/(c_0 - 1)\), and \(2C_0\nu_0^{-2}s_k\bar\lambda_k \leq \eta < 1\). In addition, assume there is a constant \(c > 0\) such that  \(\lambda_{\alpha, j}/\lambda_{\gamma, j} \geq c\) for all \(j\leq k\). Then, under the event \(\bigcap_{m=1}^6 \Omega_{k,m}\) there is a constant \(M_1\) that does not depend on  \(k\) such that
    \begin{equation}
        \label{eq:nonasymptotic-outcome}
        \max_{1\leq j\leq k}D_{\alpha, j}^\ddag(\widehat\alpha_j, \bar\alpha_j;\bar\gamma_j) \leq M_1 s_k\bar\lambda_k^2 \andbox \max_{1\leq j\leq k}\|\widehat\alpha_j - \bar\alpha_j\|_1 \leq M_1s_k\bar\lambda_k
    \end{equation}
\end{lemma}
\begin{proof}
    We show that the bound of \eqref{eq:nonasymptotic-outcome} holds for each \(j=1,\dots,k\).  We proceed in a few steps.

    \emph{Step 1: Optimization Step.} Let \(\tilde\ell_j(\alpha;\widehat\gamma_j) := \E_n[p_j(X)De^{-\widehat\gamma_j'Z}\{Y - \alpha'Z\}^2]/2\). Optimality of \(\widehat\alpha_j\) implies that for any \(u \in (0,1]\): 
    \[
        \tilde\ell_j\left(\widehat\alpha_j; \widehat\gamma_j\right) - \tilde\ell_j\left((1-u)\widehat\alpha_j + u\bar\alpha_j;\widehat\gamma_j\right) + \lambda_{\alpha, j}\|\widehat\alpha_j\|_1 \leq \lambda_{\alpha, j}\|(1-u)\widehat\alpha_j + u\bar\alpha_j\|_1
    .\]
    Convexity of the \(\ell_1\) norm  \(\|\cdot\|_1\) gives
    \[
         \tilde\ell_j\left(\widehat\alpha_j; \widehat\gamma_j\right) - \tilde\ell_j\left((1-u)\widehat\alpha_j + u\bar\alpha_j;\widehat\gamma_j\right) + \lambda_{\alpha, j}u\|\widehat\alpha_j\|_1 \leq \lambda_{\alpha, j}u\|\bar\alpha_j\|_1
    .\] 
    Dividing both sides by \(u\) and letting  \(u \to 0^+\) gives:
    \[
        -\E_n[p_j(X)De^{-\widehat\gamma_j'Z}\{Y - \widehat\alpha_j'Z\}\{\widehat\alpha_j'Z - \bar\alpha_j'Z\}] + \lambda_{\alpha,j}\|\widehat\alpha_j\|_1 \leq \lambda_{\alpha, j}\|\bar\alpha_j\|_1
    .\] 
    Rearranging using the form of \(D_{\alpha, j}^\ddag\) in  \eqref{eq:bregman-divergences} yields:
    \begin{equation}
        \label{eq:nonasymptotic-outcome-step1}\tag{O.1}
        D_{\alpha, j}^\ddag(\widehat\alpha_j, \bar\alpha_j; \widehat\gamma_j) + \lambda_{\alpha, j}\|\widehat\alpha_j\|_1 \leq (\widehat\alpha_j - \bar\alpha_j')\E_n[p_j(X)De^{-\widehat\gamma'Z}\{Y - \bar\alpha_j'Z\}Z] + \lambda_{\alpha, j}\|\bar\alpha_j\|_1
    \end{equation}

    \emph{Step 2: Quasi-Score Domination and relating \(\bar\gamma_j\) to \(\widehat\gamma_j\).} For this step, we will use the fact that we are in the event \(\Omega_{k,1} \cap \Omega_{k,2} \cap \Omega_{k,3} \cap \Omega_{k,5} \cap \Omega_{k,6}\). Using the expression for \(D_{\gamma, j}^\ddag(\widehat\gamma_j, \bar\gamma_j)\) from \eqref{eq:bregman-divergences} we find that for some \(u\in (0,1)\):
     \begin{align*}
        D_{\gamma, j}^\ddag(\widehat\gamma_j, \bar\gamma_j) 
        &= -\E_n[p_j(X)D\{e^{-\widehat\gamma_j'Z} - e^{-\bar\gamma_j'Z}\}\{\widehat\gamma_j'Z - \bar\gamma_j'Z\}] \\
        &= \E_n[p_j(X)De^{-u(\widehat\gamma_j - \bar\gamma_j)'Z}e^{-\bar\gamma_j'Z}\{\widehat\gamma_j'Z - \bar\gamma_j'Z\}^2]
    \end{align*}
    where the second step uses the mean value theorem: 
    \begin{equation}
        \label{eq:nonasymptotic-outcome-mvt}\tag{O.2}
        e^{-\widehat\gamma_j'Z}-e^{-\bar\gamma_j'Z} = e^{-u\widehat\gamma_j'Z - (1-u)\bar\gamma_j'Z}(\widehat\gamma_j - \bar\gamma_j)'Z 
    \end{equation}
    In the event \(\Omega_{k,1}\cap\Omega_{k,2}\) using the bound in Online Appendix \Cref{lemma:nonasymptotic-logit} and the fact that \(C_0\nu_0^{-2}s_k\bar\lambda_k \leq \eta < 1\) gives us that
    \begin{equation}
        \label{eq:nonasymptotic-outcome-step2.1}\tag{O.3}
        C_0\|\widehat\gamma_j - \bar\gamma_j\|_1 \leq C_0M_0s_k\bar\lambda_k \leq M_0\eta 
    .\end{equation} 
    In the event \(\Omega_{k,1}\cap\Omega_{k,2}\) the bound in \eqref{eq:nonasymptotic-logit-step6} also gives us that \(D_{\gamma, j}^\ddag(\widehat\gamma_j, \bar \gamma_j) \leq M_0s_k\lambda_{\gamma,j}^2\). Combining the above displays then yields 
    \begin{equation}
       \label{eq:nonasymptotic-outcome-step2.2}\tag{O.4}
       M_0s_k\lambda_{\gamma, j}^2 \geq  D_{\gamma, j}^\ddag(\widehat\gamma_j, \bar\gamma_j) \geq e^{M_0\eta}\E_n[p_j(X)De^{-\bar\gamma_j'Z}\{\widehat\gamma_j'Z - \bar\gamma_j'Z\}^2]
    .\end{equation} 
    Again applying the bound on \(C_0\|\widehat\gamma_j - \bar\gamma_j\|_1\)~\eqref{eq:nonasymptotic-outcome-step2.1} gives 
    \begin{equation}
        \label{eq:nonasymptotic-outcome-step2.3}\tag{O.5}
        \begin{split}
            D_{\alpha,j}^\ddag(\widehat\alpha_j, \bar\alpha_j; \widehat\gamma_j)
            &= \E_n[p_j(X)De^{-\widehat\gamma_j'Z}(\widehat\alpha_j'Z - \bar\alpha_j'Z)^2] \\
            &= \E_n[p_j(X)De^{-(\widehat\gamma_j - \bar\gamma_j)'Z}e^{-\bar\gamma_j'Z}(\widehat\alpha_j'Z - \bar\alpha_j'Z)^2] \\
            &\geq e^{-M_0\eta}D_{\alpha,j}^\ddag(\widehat\alpha_j, \bar\alpha_j ; \bar\gamma_j)
        \end{split}
    \end{equation}
    Decomposing the empirical expectation on the RHS of \eqref{eq:nonasymptotic-outcome-step1} gives
    \begin{align*}
        (\widehat\alpha_j-\bar\alpha_j)'\E_n[p_j(X)De^{-\widehat\gamma_j'Z}\{Y-\bar\alpha_j'Z\}Z]
        &= \underbrace{(\widehat\alpha_j - \bar\alpha_j)'\E_n[p_j(X)De^{-\bar\gamma_j'Z}\{Y - \bar\alpha_j'Z\}Z]}_{\delta_{1,j}} \\
        &\;\;+\underbrace{\E_n[p_j(X)D\{e^{-\widehat\gamma_j'Z}-e^{-\bar\gamma_j'Z}\}\{Y- \bar\alpha_j'Z\}\{\widehat\alpha_j'Z - \bar\alpha_j'Z\} ]}_{\delta_{2,j}} 
    \end{align*}
    By Hölder's inequality, in the event \(\Omega_{k,3}\),  \(\delta_{1,j}\) is bounded 
    \begin{equation}
        \label{eq:nonasymptotic-outcome-step2-delta1-bound}\tag{O.6}
        \delta_{1,j} \leq c_0^{-1}\|\widehat\alpha_j - \bar\alpha_j\|_1\lambda_{\alpha, j} 
    \end{equation}
    By the mean value equation~\eqref{eq:nonasymptotic-outcome-mvt} and the Cauchy-Schwarz inequality, \(\delta_{2,j}\) can be bounded from above by 
    \begin{equation}
        \label{eq:nonasymptotic-outcome-step2-delta2-first-step}\tag{O.7}
        \begin{split}
            \delta_{2,j} \leq e^{C_0\|\widehat\gamma_j - \bar\gamma_j\|_1}
            &\times \E_n^{1/2}[p_j(X)De^{-\bar\gamma_j'Z}\{\widehat\alpha'Z - \bar\alpha'Z\}^2] \\
            &\times \E_n^{1/2}[p_j(X)De^{-\bar\gamma_j'Z}\{Y - \bar\alpha_j'Z\}^2\{\widehat\gamma_j'Z - \bar\gamma_j'Z\}^2]
        \end{split}
    \end{equation}
    Using \eqref{eq:nonasymptotic-outcome-step2.1} the first term in \eqref{eq:nonasymptotic-outcome-step2-delta2-first-step} can be bounded by \(e^{M_0\eta}\). The second term is exactly the square root of \(D_{\alpha,j}^\ddag(\widehat\alpha_j, \bar\alpha_j;\bar\gamma_j)\). The third term is bounded in a few steps. First, in the event \(\Omega_{k,5}\) we have that
    \[
        (\E_n - \E)[p_j(X)De^{-\bar\gamma_j'Z}\{Y - \bar\alpha_j'Z\}^2\{\widehat\gamma_j'Z - \bar\gamma_j'Z\}] 
        \leq \bar\lambda_k\|\widehat\gamma_j -\bar\gamma_j\|_1^2
    .\] 
    By \Cref{assm:logistic-model-convergence} and \Cref{lemma:subgaussian-expectation-bound} we have that \(\E[D\{Y - \bar\alpha_j'Z\}^2] \leq G_0^2 + G_1^2\) so that:
    \[
        \E[p_j(X)De^{-\bar\gamma_j'Z}\{Y - \bar\alpha_j'Z\}^2\{\widehat\gamma_j'Z - \bar\gamma_j'Z\}^2]
        \leq (G_0^2 + G_1^2)\E[p_j(X)De^{-\bar\gamma_j'Z}\{\widehat\gamma_j'Z - \bar\gamma_j'Z\}^2]
    .\] 
    In the event \(\Omega_{k,6}\) we have that 
    \[
        (\E_n - \E)[p_j(X)De^{-\bar\gamma_j'Z}\{\widehat\gamma_j'Z - \bar\gamma_j'Z\}^2] \leq \bar\lambda_k\|\widehat\gamma_j - \bar\gamma_j\|_1
    .\] 
    and we can bound \(\E_n[p_j(X)De^{-\bar\gamma_j'Z}\{\widehat\gamma_j'Z - \bar\gamma_j'Z\}^2]\) using~\eqref{eq:nonasymptotic-outcome-step2.2}. Putting this together gives 
    \begin{equation}
        \label{eq:nonasymptotic-outcome-step2.4}\tag{O.8}
        \begin{split}
            \E_n[p_j(X)De^{-\bar\gamma_j'Z}\{Y - \bar\alpha_j'Z\}^2\{\widehat\gamma_j'Z - \bar\gamma_j'Z\}^2]
            \leq \bar\lambda_k\|\widehat\gamma_j - \bar\gamma_j\|_1^2\;\;\;\;\;\;\;&  \\
            + (G_0^2 + G_1^2)&\bar\lambda_k\|\widehat\gamma_j - \bar\gamma_j\|_1^2 \\
            &+ (G_0^2 + G_1^2)e^{-M_0\eta}M_0s_k\lambda_{\gamma, j}^2
        \end{split}
    \end{equation}
    Applying convexity of \(\sqrt{\cdot}\) and the bounds on \(\|\widehat\gamma_j - \bar\gamma_j\|_1^2\) in the event \(\Omega_{k,1}\cap\Omega_{k,2}\) from \eqref{eq:nonasymptotic-logit-step6} gives
    \begin{equation}
        \label{eq:nonasymptotic-outcome-step2-delta2-bound}\tag{O.9}
        \begin{split}
            \delta_{2,j} 
            &\leq \{e^{M_0\eta}(1 + (G_0^2 + G_1^2)^{1/2})(M_0\bar\lambda_k\lambda_{\gamma, j}s_k)^{1/2} + (G_0^2 + G_1)^2(M_0s_k\lambda_{\gamma, j}^2)^{1/2}\}D_{\alpha, j}^\ddag(\widehat\alpha_j, \bar\alpha_j;\bar\gamma_j)^{1/2} \\
            &\leq \tilde C \{(\bar\lambda_k\lambda_{\gamma, j}s_k)^{1/2} + (s_k\lambda_{\gamma, j})^{1/2}\}D_{\alpha, j}^\ddag(\widehat\alpha_j, \bar\alpha_j ; \bar\gamma_j)^{1/2} 
        \end{split}
    \end{equation}
    where \(\tilde C = \max\{e^{M_0\eta} M_0^{1/2}(1 + G_0 + G_1), (G_0^2 + G_1^2)M_0^{1/2}\}\). Combining  \eqref{eq:nonasymptotic-outcome-step2-delta1-bound} and \eqref{eq:nonasymptotic-outcome-step2-delta2-bound} gives a bound on the empirical expectation on the RHS of \eqref{eq:nonasymptotic-outcome-step1}.
    \begin{equation}
        \label{eq:nonasymptotic-outcome-step2}\tag{O.10}
        \begin{split}
            (\widehat\alpha_j - \bar\alpha_j)'\E_n[p_j(X)De^{-\widehat\gamma_j'Z}\{Y - \bar\alpha_j'Z\}Z] 
            &\leq \underbrace{c_0^{-1}\|\widehat\alpha_j-\bar\alpha_j\|_1\lambda_{\alpha, j}}_{\text{Bound on \(\delta_{1,j}\) from \eqref{eq:nonasymptotic-outcome-step2-delta1-bound}}}\\
            &\;\;+ \underbrace{\tilde C\{(\bar\lambda_k\lambda_{\gamma, j}s_k)^{1/2} + (s_k\lambda_{\gamma, j}^2)^{1/2}\}D_{\alpha, j}^\ddag(\widehat\alpha_j, \bar\alpha_j; \bar\gamma_j)^{1/2}}_{\text{Bound on \(\delta_{2,j}\) from \eqref{eq:nonasymptotic-outcome-step2-delta2-bound}}}  
        \end{split}
    \end{equation}
    For convenience, we will sometimes continue to refer to the bound on \(\delta_{2,j}\) from \eqref{eq:nonasymptotic-outcome-step2-delta2-bound} as simply \(\delta_{2,j}\).

    \emph{Step 3: Express Minimization Constraint in Terms of \(\bar\gamma_j\) and Simplify.}
    We use the results from \emph{Step 2} to rewrite the minimization bound \eqref{eq:nonasymptotic-outcome-step1} from \emph{Step 1}. Using 
    \eqref{eq:nonasymptotic-outcome-step2.3} and \eqref{eq:nonasymptotic-outcome-step2} together with the minimization bound \eqref{eq:nonasymptotic-outcome-step1} yields
    \begin{equation}
        \label{eq:nonasymptotic-outcome-step3.1}\tag{O.11}
        e^{-M_0\eta}D_{\alpha, j}^\ddag(\widehat\alpha_j, \bar\alpha_j;\bar\gamma_j) + \lambda_{\alpha,j}\|\widehat\alpha_j\|_1 \leq c_0^{-1}\lambda_{\alpha, j}\|\widehat\alpha_j - \bar\alpha_j\|_1 + \lambda_{\alpha, j}\|\bar\alpha_j\|_1 + \delta_{2,j}
    \end{equation}
    Apply the triangle inequality \(|\widehat\alpha_{j,l}| \geq |\bar\alpha_{j,l}| - |\widehat\alpha_{j,l} - \bar\alpha_{j,l}|\) for \(l\in\calS_{\alpha, j}\) and  \(|\widehat\alpha_{j,l}| = |\widehat\alpha_{j,l} - \bar\alpha_{j,l}|\) for \(l\not\in\calS_{\alpha,j}\) to the above to obtain
    \[
        e^{-M_0\eta}D_{\alpha, j}^\ddag(\widehat\alpha_j, \bar\alpha_j;\bar\gamma_j) + (1-c_0^{-1})\|\widehat\alpha_j - \bar\alpha_j\|_1 \leq 2\lambda_{\alpha, j}\sum_{l\in\calS_{\alpha, j}} |\widehat\alpha_{j,l} - \bar\alpha_{j,l}| + \delta_{2,j}
    .\] 
    Let \(\delta_j = \widehat\alpha_j - \bar\alpha_j\). We use the form \(D_{\alpha, j}^\ddag(\widehat\alpha_j,\bar\alpha_j) = \E_n[p_j(X)De^{-\bar\gamma_j'Z}\{\widehat\alpha_j'Z - \bar\alpha_j'Z\}^2] = \delta_j'\tilde\Sigma_{\gamma,j}\delta_j\) to expand out
    \begin{equation}
        \label{eq:nonasymptotic-outcome-step3}\tag{O.12}
        \begin{split}
            e^{-M_0\eta}(\delta_j'\tilde\Sigma_{\gamma, j}\delta_j) + (1-c_0^{-1})\lambda_{\alpha,j}\|\delta\|_1 
            &\leq 2\lambda_{\alpha, j}\sum_{l\in\calS_{\alpha,j}}|\delta_{j,l}| \\
            &\;\;\;\;+\tilde C \{(s_k\bar\lambda_k\lambda_{\gamma, j})^{1/2} + (s_k\lambda_{\gamma, j})^{1/2}\}(\delta_j'\tilde\Sigma_{\gamma,j}\delta_j)^{1/2} 
        \end{split}
    \end{equation}

    \emph{Step 4: Apply Empirical Compatability Condition.}
    Let \(\delta_{3,j} := \tilde C \{(s_k\bar\lambda_k\lambda_{\gamma,j})^{1/2} + (s_k\lambda_{\gamma,j})^{1/2}\}\) and \(D_{\alpha,j}^\star := e^{-M_0\eta}(\delta_j'\tilde\Sigma_{\gamma,j}\delta_j) + (1-c_0^{-1})\lambda_{\alpha,j}\|\delta_j\|_1\). In the even \(\Omega_{k,1}\cap\Omega_{k,2}\cap \Omega_{k,3}\cap\Omega_{k,5}\cap\Omega_{k,6}\) that  \eqref{eq:nonasymptotic-outcome-step3} holds, there are two possibilities. For \(\xi_2 =1 - 2c_0/\{(\xi_1 + 1)(c_0 -1)\}\in(0,1]\) either 
    \begin{equation}
        \label{eq:nonasymptotic-outcome-step4-case1}\tag{O.13}
        \xi_2 D_{\alpha, j}^\star \leq \delta_{3,j}(\delta_j'\tilde\Sigma_{\gamma,j}\delta_j)^{1/2}
    \end{equation}
    or \((1-\xi_2)D_{\alpha,j}^\star \leq 2\lambda_{\alpha, j}\sum_{l\in\calS_{\alpha,j}}|\delta_{j,l}|\), that is 
    \begin{equation}
        \label{eq:nonasymptotic-outcome-step4-case2}\tag{O.14}
        D_{\alpha_j}^\star \leq (\xi_1+1)(c_0 -1)c_0^{-1}\lambda_{\alpha,j}\sum_{l\in\calS_{\alpha,j}}|\delta_{j,l}|
    \end{equation}
    We deal with these two cases separately. First, if \eqref{eq:nonasymptotic-outcome-step4-case2} holds, then \(\sum_{l\not\in\calS_{\alpha,j}}|\delta_{j,l}| \leq \xi_1\sum_{l\in\calS_{j,l}}|\delta_{j,l}|\). We can apply the empirical compatability of \Cref{assm:logistic-model-convergence} to \eqref{eq:nonasymptotic-outcome-step4-case2} to obtain.
    \[
        e^{-M_0\eta}(\delta_j'\tilde\Sigma_{\gamma,j}\delta_j) 
        + (1-c_0^{-1})\lambda_{\alpha, j}\|\delta_{j,l}\| \leq \nu_1(\xi_1+ 1)(\xi_1 - 1)\lambda_{\alpha, j}(s_j\delta_j\tilde\Sigma_{\gamma,j}\delta_j)^{1/2}
    .\] 
    Inverting for \((\delta_j\tilde\Sigma_{\gamma,j}\delta_j)^{1/2}\) and plugging in gives 
    \begin{equation}
        \label{eq:nonasymptotic-outcome-step4.1}\tag{O.15}
        e^{-M_0\eta}D_{\alpha,j}^\ddag(\widehat\alpha,\bar\alpha_j;\bar\gamma_j) + (1-c_0^{-1})\lambda_{\alpha,j}\|\widehat\alpha_j-\bar\alpha_j\|_1 \leq \tilde M s_k\lambda_{\alpha,j}^2
    \end{equation}
    where \(\tilde M =e^{M_0\eta}(\xi_1 + 1)(c_0 - 1)c_0^{-1}\). Next, assume that \eqref{eq:nonasymptotic-outcome-step4-case1} holds. In this case, we can directly invert for \((\delta_j\tilde\Sigma_{\gamma,j}\delta_j)^{1/2}\) to get that
    \begin{equation}
        \label{eq:nonasymptotic-outcome-step4.2}\tag{O.16}
        e^{-M_0\eta}D_{\alpha,j}^\ddag(\widehat\alpha_j,\bar\alpha_j;\bar\gamma_j) + (1-c_0^{-1})\lambda_{\alpha,j}\|\widehat\alpha_j-\bar\alpha_j\|_1 \leq \xi_2^{-1}\tilde C\{(s_k\bar\lambda_k\lambda_{\gamma,j})^{1/2} + (s_k\lambda_{\gamma,j}^2)^{1/2}\}^2 
    \end{equation}
    Combining \eqref{eq:nonasymptotic-outcome-step4.1} and \eqref{eq:nonasymptotic-outcome-step4.2} gives
    \begin{equation}
        \label{eq:nonasymptotic-outcome-step4}\tag{O.17}
        \begin{split}
            e^{-M_0\eta}D_{\alpha,j}^\ddag(\widehat\alpha_j,\bar\alpha_j;\bar\gamma_j) + (1-c_0^{-1})\lambda_{\alpha,j}\|\widehat\alpha_j-\bar\alpha_j\|_1 
            &\leq \tilde Ms_k \lambda_{\alpha,j}^2 \\
            &\;\;\;\;\;+ \xi_2^{-1}\tilde C\{(s_k\bar\lambda_k\lambda_{\gamma,j})^{1/2} + (s_k\lambda_{\gamma,j}^2)^{1/2}\}^2 
        \end{split}
    \end{equation}
    \emph{Step 5: Apply Penalty Majorization and Bounded Penalty Ratio.}
    Use the fact that \(\lambda_{\gamma,j}/\lambda_{\alpha,j} \leq c^{-1}\) to express \eqref{eq:nonasymptotic-outcome-step4} as 
    \begin{align*}
        D_{\alpha,j}^\ddag(\widehat\alpha_j,\bar\alpha_j;\bar\gamma_j) 
        &\leq  e^{M_0\eta}\tilde Ms_k\lambda_{\alpha,j}^2 + e^{M_0\eta}\xi_2^{-1}\tilde C\{(s_k\bar\lambda_k\lambda_{\gamma,j})^{1/2} + (s_k\lambda_{\gamma,j}^2)^{1/2}\}^2  \\
        \|\widehat\alpha_j - \bar\alpha_j\|_1 
        &\leq (1- c_0^{-1})^{-1}\tilde M s_k\lambda_{\alpha,j} + (1-c_0^{-1})^{-1}c^{-1}\tilde C\{(s_k\bar\lambda_k)^{1/2} + (s_k\lambda_{\gamma,j})^{1/2}\}^2
    \end{align*}
    In the event \(\Omega_{k,2} \cap \Omega_{k,3}\) we have that  \(\lambda_{\gamma,j}\vee\lambda_{\alpha,j}\leq \bar\lambda_k\), so that the above simplifies to 
    \begin{equation}
        \label{eq:nonasymptotic-asymptotic-step5}\tag{O.18}
        \begin{split}
            D_{\alpha,j}^\ddag(\widehat\alpha_j,\bar\alpha_j;\bar\gamma_j) 
            &\leq M_1s_k\bar\lambda_k^2 \\
            \|\widehat\alpha_j - \bar\alpha_j\|_1 
            &\leq M_1s_k\bar\lambda_k
        \end{split}
    \end{equation}
    for \(M_1 = \max\{e^{M_0\eta}, c^{-1}(1-c_0^{-1})^{-1}\}(\tilde M + 2e^{M_0\eta}\xi_2^{-1}\tilde C)\). This completes the result \eqref{eq:nonasymptotic-outcome}.
\end{proof}

\subsection{Nonasymptotic Bounds for Residual Estimation}
\label{subsec:residual-nonasymptotic-bounds}
We now provide nonasymptotic bounds on the empirical mean square error between the estimated residuals \(\widehat U_{\gamma,j}\) and \(\widehat U_{\alpha,j}\) and the true residuals
\begin{equation}
    \label{eq:true-residuals}
    \begin{split}
        U_{\gamma,j} &:= -p_j(X)\{De^{-\bar\gamma_j'Z} + (1-D)\}  \\
        U_{\alpha,j} &:= p_j(X)De^{-\bar\gamma_j'Z}(Y - \bar\alpha_j^{\text{\tiny pilot}'}Z),
    \end{split}
\end{equation}
These bounds will be shown under the events in  \eqref{eq:nonasymptotic-events-logit}, \eqref{eq:nonasymptotic-events-outcome}, and \eqref{eq:mean-event} using the results in \Cref{lemma:nonasymptotic-logit,lemma:nonasymptotic-outcome}.

\begin{lemma}[Nonasymptotic Logistic Residual Bound]
    \label{lemma:nonasymptotic-logit-residual-bound}
    Suppose that \Cref{assm:logistic-model-convergence} and the conditions of \Cref{lemma:nonasymptotic-logit} hold. Then, in the event  \(\Omega_{k,1} \cap \Omega_{k,2}\) described on \eqref{eq:nonasymptotic-events-logit} there is a constant \(M_{\gamma, r}\) that does not depend on  \(k\) such that:
    \begin{equation}
        \label{eq:nonasymptotic-logit-residual-bound}
        \max_{1\leq j \leq k} \E_n[(\widehat U_{\gamma, j} - U_{\gamma, j})^2] \leq M_{\gamma, r}\xi_{k,\infty}s_k \bar\lambda_k^2
    .\end{equation} 
\end{lemma}
\begin{proof}
    Consider each \(j\) separately. By applying the mean value theorem \eqref{eq:nonasymptotic-outcome-mvt} and \Cref{lemma:nonasymptotic-logit}, we can write
    \begin{align*}
        (\widehat U_{\gamma, j} - U_{\gamma, j})^2 
        &= p_j(X)^2D\{e^{-\widehat\gamma_j'Z} - e^{-\bar\gamma_j'Z}\}\{e^{-\widehat\gamma_j'Z} - e^{-\bar\gamma_j'Z}\} \\
        &\leq \xi_{k,\infty}p_j(X)D\{e^{-\widehat\gamma_j'Z} - e^{-\bar\gamma_j'Z}\}e^{-\bar\gamma_j'Z - u(\widehat\gamma_j - \bar\gamma_j)'Z}\{\bar\gamma_j'Z - \widehat\gamma_j'Z\}  \\
        &\leq \xi_{k,\infty}e^{-B_0 + M_0\eta}D\{e^{-\widehat\gamma_j'Z} -e^{- \bar\gamma_j'Z}\}\{\bar\gamma_j'Z - \widehat\gamma_j'Z\} 
    \end{align*}
    So that
    \begin{align*}
        \E_n[(\widehat U_{\gamma, j} - U_{\gamma,j})^2] 
        &\leq e^{-B_0 + M_0\eta}\xi_{k,\infty}\underbrace{\E_n[p_j(X)D\{e^{-\widehat\gamma_j'Z}\}\{\widehat\gamma_j'Z - \bar\gamma_j'Z\}]}_{= D_{\gamma, j}^\ddag(\widehat\gamma_j,\bar\gamma_j)}\\
        &\leq e^{-B_0 + M_0\eta}\xi_{k,\infty}s_k\bar\lambda_k^2
    \end{align*}
\end{proof}

\begin{lemma}[Nonasymptotic Linear Residual Bound]
    \label{lemma:nonasymptotic-linear-residual-bound}
    Suppose that \Cref{assm:logistic-model-convergence} and the conditions of \Cref{lemma:nonasymptotic-outcome} hold. Then, in the event \(\bigcap_{m=1}^6 \Omega_{k,m}\), there is a constant  \(M_{\alpha, r}\) that does not depend on  \(k\) such that
    \begin{equation}
        \label{eq:nonasymptotic-linear-residual-bound}
        \max_{1\leq j\leq k}\E_n[(\widehat U_{\alpha, j} - U_{\alpha, j})^2] \leq M_{\alpha, r}\xi_{k,\infty}^2s_k^2\bar\lambda_k^2
    \end{equation}
\end{lemma}
\begin{proof}
   Recall that  \( \widehat U_{\alpha, j} = p_j(X)De^{-\widehat\gamma_j'Z}(Y - \widehat\alpha_j'Z)\) and \(U_{\alpha, j} = p_j(X)De^{-\bar\gamma_j'Z}(Y - \bar\alpha_j'Z)\). As an intermediary, define \(\dot U_{\gamma, j} = p_j(X)De^{-\widehat\gamma_j'Z}(Y- \bar\alpha_j'Z)\). We will show a bound on the empirical mean square error between  \( \widehat U_{\alpha, j}\) and \(\dot U_{\alpha, j}\) as well as on the empirical mean square error between  \(\dot U_{\alpha, j}\) and  \(U_{\alpha, j}\). The bound in  \eqref{eq:nonasymptotic-linear-residual-bound} will then follow from \((a+b)^2 \leq 2a^2 + 2b^2\).

   First consider \((\widehat U_{\alpha, j} - \dot U_{\alpha, j})^2\):
   \begin{align*}
       \E_n[(\widehat U_{\alpha,j} - \bar U_{\alpha,j})^2]
       &= \E_np_j^2(X)De^{-2\widehat\gamma_j'Z}(\widehat\alpha_j'Z - \bar\alpha_j'Z)^2] \\
       &= \E_n[p_j^2(X)De^{-2(\bar\gamma_j'Z - (\widehat\gamma_j - \bar\gamma_j)'Z)}(\widehat\alpha_j'Z - \bar\alpha_j'Z)^2] \\
       &\leq \xi_{k\infty}e^{-B_0}e^{2M_0\eta}\underbrace{\E_n[p_j(X)De^{-\bar\gamma_j'Z}(\widehat\alpha_j'Z - \bar\alpha_j'Z)]}_{= D_{\alpha, j}^\ddag(\widehat\alpha_j, \bar\alpha_j;\bar\gamma_j)}\\
       &\leq e^{2M_0\eta - B_0}M_1\xi_{k,\infty}s_k\bar\lambda_k^2
   \end{align*}
   Where the last empirical expectation is bounded by \Cref{lemma:nonasymptotic-outcome}.
   Next, consider \((\dot U_{\alpha, j} - U_{\alpha, j})^2\):
   \begin{align*}
       \E_n[(\dot U_{\alpha, j} - U_{\alpha, j})^2]
       &= \E_n[p_j^2(X)D\{e^{-\widehat\gamma'Z}-e^{-\bar\gamma'Z}\}^2\{Y - \bar\alpha_j'Z\}^2] \\
       &= \E_n[p_j^2(X)D\{e^{-\bar\gamma'Z - u(\widehat\gamma - \bar\gamma)'Z}(\bar\gamma_j'Z - \widehat\gamma_j'Z)\}^2(Y - \bar\alpha_j'Z)^2] \\
       &\leq 2e^{M_0\eta - B_0}C_0^2\xi_{k,\infty}(M_1 s_k\bar\lambda_k)^2\E_n[p_j(X)De^{-\bar\gamma_j'Z}(Y - \bar\alpha_j'Z)^2]
   \end{align*}
    To proceed we assume that  \(Z\) contains a constant. That is  \(Z = (1,Z_2,\dots,Z_{d_z})\). However, this is not necessary it just simplifies the proof a bit. We bound the final empirical expectation in the event \(\Omega_{k,5}\). In this event we can bound
    \begin{align*}
        \E_n[p_j(X)De^{-\bar\gamma_j'Z}(Y - \bar\alpha_j'Z)^2] 
        &= (\E_n - \E)[p_j(X)De^{-\bar\gamma_j'Z}(Y - \bar\alpha_j'Z)^2] + \E[p_j(X)De^{-\bar\gamma_j'Z}(Y - \bar m_j(X))^2] \\
        &\leq \bar\lambda_k + \xi_{k,\infty}e^{-B_0}(D_0 + D_1)^2
    .\end{align*}
    Combining the above, and using the fact that \(s_k\bar\lambda_k \leq \eta < 1\) completes the reult.

\end{proof}

\subsection{Probability Bounds for the First Stage}
\label{subsec:first-stage-probability-bounds}
In this section we establish that each of the events in \eqref{eq:nonasymptotic-events-logit}, \eqref{eq:nonasymptotic-events-outcome}, and \eqref{eq:mean-event} occurs under \Cref{assm:logistic-model-convergence} with probability approaching one.

\begin{lemma}[Logistic Score Domination and Penalty Majorization]
    \label{lemma:logit-score-domination}
    Suppose \Cref{assm:logistic-model-convergence} holds and that the penalty parameter \(\lambda_{\gamma, j}\) is chosen as described in \Cref{sec:setup}. Then, for \(n\) sufficiently large, the event \(\Omega_{k,1}\) holds with probability \(1 - \eps - \rho_{\gamma,n}\) where 
    \begin{equation}
        \label{lemma:first-stage-logistic-bound}
        \rho_{\gamma,n} = C\max\left\{\frac{4kn + 4k}{n^2},\bigg(\frac{\tilde M\xi_{k,\infty}s_{k,\gamma}\bar c_n^2\ln^5(d_zn)}{n}\bigg)^{1/2}, \bigg(\frac{\tilde M\xi_{k,\infty}^4\ln^7(d_zkn)}{n}\bigg)^{1/6},\frac{1}{\ln^2(d_zkn)}\right\}
    .\end{equation} 
    where \(C, \tilde M\) are absolute constants that do not depend on \(k\). In particular so long as \(\eps \to 0\) as  \(n\to \infty\), this shows that  \(\Pr(\Omega_{k,1}) = 1 - o(1)\) under the rate conditions of \Cref{assm:logistic-model-convergence}.

    Moreover, with probability at least \(1 - \frac{5k}{n} - \frac{4k}{n^2} \) there is a constant \(M_2\) that does not depend on \(k\) such that \(\Omega_{k,2}\) holds with 
    \begin{equation}
        \label{eq:first-stage-logit-penalty-majorization}
        \bar\lambda_k = \max\{M_2, M_4, M_5, M_6, M_7\}\xi_{k,\infty}\sqrt\frac{\ln(d_zn)}{n} 
    \end{equation}
    where \(M_4, M_5, M_6\) and  \(M_7\) are all constants that also do not depend on  \(k\) described in Lemma~\ref{lemma:linear-score-domination} and Lemmas~\ref{lemma:omega-5}-\ref{lemma:omega-7}. In particular, so long as \(k/n\to 0\),  \(\Pr(\Omega_{k,2}) = 1 - o(1)\).
\end{lemma}
\begin{proof}
    Collecting the logistic nonasymptotic residual bound from 
    \Cref{lemma:nonasymptotic-logit-residual-bound} and the probability bounds from 
    \Crefrange{lemma:omega-5}{lemma:logit-score-domination-deterministic} we find that, (eventually) with probability at least \(1 - \frac{4k}{n} - \frac{4k}{n^2}\):
    \begin{equation}
        \label{eq:logit-mean-square-bound}\tag{P.1}
        \max_{\substack{1\leq j\leq k\\ 1\leq l \leq d_z}} \E_n[(\widehat U_{\gamma, j}Z_l - U_{\gamma, j}Z_l)^2] \leq M_{\gamma, r}C_0^2\frac{\xi_{k,\infty}s_{k,\gamma}\bar c_n^2\ln^3(d_zn)}{n} 
    .\end{equation} 
    where \(M_{\gamma, r}\) is a constant that does not depend on \(k\). Define the vectors
    \begin{align*}
        W_k &:= (U_{\gamma,1}Z',\dots,U_{\gamma,k}Z')' \in\SR^{kd_z}\\
            &:= (W_{k,1}',\dots,W_{k,k}')'\\
        \widehat W_k &:= (\widehat U_{\gamma,1}Z',\dots,\widehat U_{\gamma, k}Z')'\in\SR^{kd_z} \\
                     &:= (\widehat W_{k,1}',\dots,\widehat W_{k,k}')'
    .\end{align*}
    Notice by optimality of \(\bar\gamma_1,\dots,\bar\gamma_k\) that \(W_k\) is a mean zero vector. Under our assumptions the covariance matrix  \(\Sigma_k = \frac{1}{n} \sum_{i=1}^n \E[W_kW_k']\) exists and is finite. Define the sequences of constants
    \begin{align*}
        \delta_{\gamma,n}^2 &:= M_{\gamma,r}C_0^2\xi_{k,\infty}s_{k,\gamma}\bar c_n^2\ln^5(d_zn)/n \\
        \beta_{\gamma, n} &:= \frac{4k}{n} + \frac{4k}{n^2}  
    \end{align*}
    Then, by \eqref{eq:logit-mean-square-bound} we have that with probability at least \(1 -\beta_{\gamma,n}\)
    \begin{equation}
        \label{eq:logit-prob-bound}\tag{P.2}
        \Pr\left(\|\E_n[(\widehat W_k - W_k)^2] \|_\infty > \delta_n^2/\ln^2(d_zn)\right) \leq \beta_n
    .\end{equation}
    Let \(e_1,\dots,e_n\) be i.i.d normal random variables generated independently of the data. Define the scaled random variables and the multiplier bootstrap process 
    \begin{align*}
        \widehat S_{\gamma, n}^e 
        &:= n^{-1/2}\sum_{i=1}^n e_i \widehat W_{k,i} \\
        &:= (\widehat S_{\gamma, 1}^{e'}, \dots, \widehat S_{\gamma, k}^{e'})'
    \end{align*} 
    and let \(\Pr_e\) denote the probability measure with respect to the  \(e_i's\) conditional on the observed data. \Cref{assm:logistic-model-convergence} implies that the conditions of \eqref{eq:clt-conditions} hold for \(Z = W_k\) with  \(b\) replaced by  \(c_u\) and  \(B_n\) replaced by  \(B_k = (\xi_{k,\infty}C_0C_U)^3\vee 1\). Further, via \eqref{eq:logit-prob-bound} the residual estimation requirement of  with \(\delta_n\) and  \(\beta_n\) replaced by \(\delta_{\gamma, n}\) and  \(\beta_{\gamma, n}\).

    Let \(\widehat q_{\gamma,j}(\alpha)\) be the \(\alpha\) quantile of \(\|\widehat S_{\gamma,j}^{e'}\|\) \emph{conditional} on the data \(Z_i\) and the estimates  \( \widehat Z_i\). \Cref{thm:bootstrap-consistency} then shows that there is a finite constant depending only on \(c_u\) such that 
    \begin{align*}
        \max_{1\leq j\leq k}\sup_{\alpha\in(0,1)} 
        \left|\Pr(\|S_{\gamma, j}\|\geq \widehat q_{\gamma,j}(\alpha)) - \alpha\right| \leq C\max\left\{\beta_{\gamma,n}, \delta_{\gamma,n}, \bigg(\frac{B_k^4\ln^7(kd_zn)}{n}\bigg)^{1/6},\frac{1}{\ln^2(kd_zn)}\right\}
    .\end{align*} 
    This gives the first claim of \Cref{lemma:logit-score-domination} by construction of \(\lambda_{\gamma,j}\). 
    The second claim follows \Cref{lemma:quantile-bound}. For this second claim we will consider the marginal convergence of each \(U_{\gamma,j}Z\) as opposed to their joint convergence (the convergence of \(W_k\)). First, notice that condiitonal on the data, the random vector \(\E_n[e\widehat U_{\gamma, j}Z]\) is centered gaussian in \(\SR^{d_z}\). \Cref{lemma:quantile-bound} then shows that 
    \[
        \widehat q_{\gamma, j}(\eps) \leq (2 + \sqrt{2})\sqrt{\frac{\ln(d_z/\eps)}{n}\max_{1\leq l\leq d_z}\E_n[\widehat U_{\gamma,j}^2Z_l^2] }
    .\] 
    Furthermore, with probability at least \(1 - \beta_{\gamma, n} - \frac{1}{n}\) we have that, for all \(j=1,\dots,k\):
    \begin{align*}
        \max_{1\leq l\leq d_z}\E_n[\widehat U_{\gamma, j}^2Z_l^2] \leq C_0^2\E_n[\widehat U_{\gamma, j}^2] \leq 2C_0^2(\E_n[U_{\gamma, j}^2] + \E_n[(\widehat U_{\gamma, j}^2 - U_{\gamma, j})^2]) \leq 4C_0^2\xi_{k,\infty}^2C_U^2 + \delta_{\gamma, n}^2/\ln^2(d_zn))
    \end{align*}
    Under the rate conditions of \Cref{assm:logistic-model-convergence}, \(\delta_{\gamma, n}^2/\ln^2(d_zn)\) will eventually be smaller than \(1\) and so the claim in \eqref{eq:first-stage-logit-penalty-majorization} holds with \(M_2 = 8C_0^2C_U^2 \vee 1\) .
\end{proof}

\begin{lemma}[Linear Score Domination and Penalty Majorization]
    \label{lemma:linear-score-domination}
    Suppose \Cref{assm:logistic-model-convergence} holds and that the penalty parameters \(\lambda_{\gamma, j}\) and \(\lambda_{\alpha,j}\) are chosen as described in \Cref{sec:setup}. Then, for \(n\) sufficiently large, the event \(\Omega_{k,3}\) holds with probability \(1 - \eps - \rho_{\alpha,n}\) where:    \begin{equation}
        \label{lemma:first-stage-linear-bound}
        \rho_{\alpha,n} = C\max\left\{\frac{4kn + 4k}{n^2},\bigg(\frac{\tilde M\xi_{k,\infty}^2s_{k,\alpha}^2\bar c_n^2\ln^5(d_zn)}{n}\bigg)^{1/2}, \bigg(\frac{\tilde M\xi_{k,\infty}^4\ln^7(d_zkn)}{n}\bigg)^{1/6},\frac{1}{\ln^2(d_zkn)}\right\}
    .\end{equation} 
    where \(C, \tilde M\) are absolute constants that do not depend on \(k\). In particular so long as \(\eps \to 0\) as  \(n\to \infty\), this shows that  \(\Pr(\Omega_{k,3}) = 1 - o(1)\) under \Cref{assm:logistic-model-convergence}.

    Moreover, with probability at least \(1 - \frac{5k}{n} -\frac{4k}{n^2}\) there is a constant \(M_4\) that does not depend on \(k\) such that \(\Omega_{k,4}\) holds with 
    \begin{equation}
        \label{eq:first-stage-linear-penalty-majorization}
        \bar\lambda_k = \max\{M_2, M_4, M_5, M_6, M_7\}\xi_{k,\infty}\sqrt\frac{\ln(d_zn)}{n} 
    \end{equation}
    where  \(M_2, M_5, M_6\) and  \(M_7\) are all constants that also do not depend on  \(k\) described in Lemma~\ref{lemma:logit-score-domination} and Lemmas~\ref{lemma:omega-5}-\ref{lemma:omega-7}. In particular, so long as \(k/n\to 0\), \(\Pr(\Omega_{k,4}) = 1 - o(1)\).
\end{lemma}
\begin{proof}
    Apply the same steps as the proof of \Cref{lemma:logit-score-domination} with  
    \begin{align*}
        \delta_{\alpha,n}^2 &= M_{\alpha, r}C_0^2\xi_{k,\infty}^2s_k^2\bar c_n^2\ln^5(d_zn)/n \\
        \beta_{\alpha,n} &= \frac{4}{n} + \frac{4}{n^2}  
    \end{align*}
\end{proof}

\begin{lemma}[Probabilistic Bound on \(\Omega_{k,5}\)]
    \label{lemma:omega-5}
    Let \(\tilde\Sigma_{\alpha, j}\) and \(\Sigma_{\alpha, j} = \E\tilde\Sigma_{\alpha, j}\) be as in \eqref{eq:additional-matrices-outcome}. Under \Cref{assm:logistic-model-convergence} if
    \[
        \bar\lambda_k \geq 4\xi_{k,\infty}(G_0^2 + G_0G_1)C_0^2\left[G_0^2\log(d_z/\eps)/n + G_0G_1\sqrt{\log(d_z/\eps)/n}\right]
    \] 
    Then \(\Pr(\Omega_{k,5}) \geq 1 - 2k\eps^2\). In particular, there is a constant \(M_5\) that does not depend on \(k\), such that if \(\bar\lambda_k \geq \xi_{k,\infty}M_5\sqrt{\log(d_z/\eps)/n}\) and \(k\eps^2 \to 0\) as  \(n\to\infty\) then under the conditions of  \Cref{assm:logistic-model-convergence}, \(\Pr(\Omega_{k,5}) = 1 - o(1)\).
\end{lemma}
\begin{proof}
    We show that this happens with probability \(1 - 2\eps^2\) for each  \(j = 1,\dots,k\).
    For any \(l,h = 1,\dots,d_z\), the variable 
    \[
        p_j(X)e^{-\bar\gamma'Z}D\{Y - \bar m_j(Z)\}^2Z_lZ_h 
    \] 
    is the product of \(p_k(X)e^{-\bar\gamma_j'Z}Z_lZ_h\), which is bounded in absolute value by  \(\xi_{k,\infty}C_0^2e^{-B_0}\), and \(D\{Y - \bar m_j(Z)\}\), which is uniformly sub-gaussian conditional on \(Z\). By 
    \Cref{lemma:tan-18} we have:
    \[
        \E\left[|(\tilde\Sigma_{\alpha, j})_{lh} - (\tilde\Sigma_{\alpha, j})_{lh}|^k\right] \leq \frac{k!}{2}(2\xi_{k,\infty}C_0^{-2}e^{-B_0}G_0^2)^{k-2}(2\xi_{k,\infty}C_0^2e^{-B_0}G_0G_1)^2,\;\;k=2,3,\dots 
    .\] 
    Apply the above and \Cref{lemma:tan-16} with \(t = \log(d_z^2/\eps^2)/n\) to obtain
    \[
        \Pr\left(|(\tilde\Sigma_{\alpha, j})_{lh} - (\tilde\Sigma_{\alpha, j})_{lh}| > 2e^{-B_0}\xi_{k,\infty}C_0^2G_0^2t + 2e^{-B_0}\xi_{k,\infty}C_0^2G_0G_1\sqrt{2t}\right) \leq 2\eps^2/d_z^2
    .\]
    A union bound completes the argument.
\end{proof}

\begin{lemma}[Probabilistic Bound on \(\Omega_{k,6}\)]
    \label{lemma:omega-6}
    Let \(\tilde\Sigma_{\gamma, j}\) and \(\Sigma_{\gamma,j} = \E\tilde\Sigma_{\gamma, j}\) be as in \eqref{eq:additional-matrices-outcome}. Under \Cref{assm:logistic-model-convergence} if
    \[
        \bar\lambda_k \geq \xi_{k,\infty}\sqrt{2}(e^{-B_0}+1)C_0\sqrt{\log(d_z/\eps)/n}
    ,\] 
    then \(\Pr(\Omega_{k,6}) \leq 1 - 2k\eps^2\). In particular, there is a constant \(M_6\) that does not depend on \(k\), such that if \(\bar\lambda_k \geq \xi_{k,\infty}M_6\sqrt{\log(d_z/\eps)/n}\) and \(k\eps^2 \to 0\) as  \(n\to\infty\) then under the conditions of  \Cref{assm:logistic-model-convergence}, \(\Pr(\Omega_{k,6}) = 1 - o(1)\).
\end{lemma}

\begin{proof}
    Consider each \(j\) separately. For any  \(l,h = 1,\dots,d_z\), note \(|(\tilde\Sigma_{\gamma, j})_{lh}| = |p_j(X)De^{-\bar\gamma_j'Z}Z_lZ_h| \leq \xi_{k,\infty}C_0^2e^{-B_0}\) so that \((\tilde\Sigma_{\gamma, j})_{lh} - (\Sigma_{\gamma,j})_{lh}\) is mean zero and bounded in abosulte values by  \(2\xi_{k,\infty}C_0^2e^{-B_0}\). Applying \Cref{lemma:tan-14} with  \(\bar\lambda_k \geq 4\xi_{k,\infty}C_0^2e^{-B_0}\sqrt{\log(d_z/\eps)/n}\) yields:
    \[
        \Pr\left(|(\tilde\Sigma_{\gamma,j})_{lh} - (\Sigma_{\gamma, j})_{lh}| \geq \bar\lambda_k \right) \leq 2\eps^2/d_z^2
    .\]
    A union bound completes the argument.
\end{proof}

\begin{lemma}[Probabilitstic Bound on \(\Omega_{k,7}\)]
    \label{lemma:omega-7}
    Let \(\tilde\Sigma_{\alpha, j}^1\) and \(\Sigma_{\alpha, j}^1 = \E\tilde\Sigma_{\alpha, j}^1\) be as in \eqref{eq:mean-event}.  Under \Cref{assm:logistic-model-convergence} if 
     \[
        \bar\lambda_k \geq \xi_{k\infty}4(G_0^2 + G_1^2)^{1/2}e^{-B_0}C_0^2\sqrt{\log(d_z/\eps)/n}
    ,\]
    then \(\Pr(\Omega_{k,7}) \geq 1 - 2k\eps^2\). In particular, there is a constant \(M_7\) that does not depend on  \(k\) such that if  \(\bar\lambda_k \geq \xi_{k,\infty}M_7\sqrt{\log(d_z/\eps)/n}\) and \(k\eps^2 \to 0\) as  \(n\to\infty\) then, under the conditions of \Cref{assm:logistic-model-convergence}, \(\Pr(\Omega_{k,7}) \geq 1 - o(1)\). 
\end{lemma}
\begin{proof}
    We deal with each \(j\) term separately. The variables \(p_j(X)e^{-\bar\gamma_j'Z}|Y - \bar m_j(Z)|Z_lZ_h\) are uniformly sub-gaussian conditional on \(Z\) because \(|p_j(X)e^{-\bar\gamma_j'Z}Z_lZ_h| \leq \xi_{k,\infty}e^{-B_0}C_0^2\) and \(D|Y - \bar m_j(Z)|\) is uniformly sub-gaussian. Applying \Cref{lemma:tan-15} for  \(\bar\lambda_k \geq e^{-B_0}\xi_{k,\infty}C_0^2\sqrt{8(G_0^2 + G_1)^2}\sqrt{\log(d_z/\eps)/n}\) yields 
    \[
        \Pr\left(|(\tilde\Sigma_{\gamma, j})_{lh} - (\Sigma_{\gamma, j})_{lh}| \geq \bar\lambda_k \right) \leq 2\eps^2/d_z^2
    .\] 
    A union bound completes the argument.
\end{proof}

\subsection{Probability Bounds for Residual Estimation}%
\label{subsec:residual-estimation-probability-bounds}
For showing consistent residual estimation, we employ the following two lemmas.

\begin{lemma}[Deterministic Logistic Score Domination]
    \label{lemma:logit-score-domination-deterministic}
    Under \Cref{assm:logistic-model-convergence} let
    \[
        \bar\lambda_k \geq \xi_{k,\infty}\sqrt{2}(e^{-B_0} + 1)C_0\sqrt{\ln(d_z/\eps)/n}
    .\] 
    Then if for all \(j = 1,\dots,k\) we let  \(\lambda_{\gamma, j} \equiv \bar\lambda_k\), \(\Pr(\Omega_{k,1}\cap\Omega_{k,2}) \geq 1 - 2k\eps\).  In particular, there is a constant \(M_8^p\) that does not depend on \(k\) such that if \(\bar\lambda_k \geq M_8^p\xi_{k,\infty}\sqrt{\ln(d_zn)/n}\) \(\Pr(\Omega_{k,1}\cap\Omega_{k,2}) \geq 1 - 2k/n^p\).  

\end{lemma}
\begin{proof}
    Let us recall that  
    \[
        \|S_j\|_\infty = \max_{1\leq l\leq d_z}|\E_n[p_j(X)\{-De^{-\bar\gamma_j'Z} + (1-D)\}Z_j]|
    .\] 
    Notice for each \(1\leq l \leq d_z\), \(S_{j,l} = p_j(X)\{-De^{-\bar\gamma_j'Z} + (1-D)\}Z_l\) is bounded in absolute value by  \(C_0\xi_{k,\infty}(e^{-B_0} + 1)\) and is mean zero by optimality of  \(\bar\gamma_j\). For \(\bar\lambda_k \geq 2(e^{-B_0} +1)C_0\sqrt{\ln(d_z/\eps)/n}\) apply \Cref{lemma:tan-14} to see the result.
\end{proof}

\begin{lemma}[Deterministic Linear Score Domination]
    \label{lemma:linear-score-domination-deterministic}
    Under \Cref{assm:logistic-model-convergence} let 
    \[
        \bar\lambda_k \geq \xi_{k,\infty}(e^{-B_0}C_0)\sqrt{8(G_0^2 + G_1^2)}\sqrt{\ln(d_z/\eps)/n}
    .\] 
    Then if for all \(j=1,\dots,k\) we let \(\lambda_{\gamma, j} \equiv \bar\lambda_k\),  \(\Pr(\Omega_{k,3}\cap\Omega_{k,4}) \geq 1 - 2k\eps\). In particular, there is a constant \(M_9^p\) that does not depend on \(k\) such that if \(\bar\lambda_k \geq M_9^p\xi_{k,\infty}\sqrt{\ln(d_zn)/n}\), \(\Pr(\Omega_{k,3}\cap\Omega_{k,4}) \geq 1 - 2k/n^p\).  
\end{lemma}
\begin{proof}
    Notice \(S_{j,l} = p_j(X)De^{-\bar\gamma_j'Z}\{Y - \bar m_j(Z)\}Z_l\) for \(l = 1,\dots,p\). By optimality of \(\bar\alpha_j\),  \(S_{j,l}\) is mean zero. Under \Cref{assm:logistic-model-convergence},  \(|S_{j,l}| \leq e^{-B_0}C_0|D\{Y - \bar m_j(Z)\}|\) so by \Cref{assm:logistic-model-convergence} the variables \(S_{j,l}\) are uniformly sub-gaussian conditional on  \(Z\) in the following sense:
    \[
         \max_{l=1,\dots,p} \tilde G_0^2\E[\exp(S_{j,l}^2 / \tilde G_0^2) - 1] \leq \tilde G_1^2
    \] 
    for \(\tilde G_0 = \xi_{k,\infty}C_0G_0e^{-B_0}\) and  \(\tilde G_1 = \xi_{k,\infty} C_0 G_1 e^{-B_0}\). Apply \Cref{lemma:tan-15} for  \(\bar \lambda_k\) defined above in the statement of \Cref{lemma:linear-score-domination-deterministic} and union bound to obtain the result.
\end{proof}

\section{Additional Second Stage Results}
\label{sec:additional-second-stage}

\begin{theorem}[Integrated Rate of Convergence]
    \label{thm:rate-of-convergence}
    Assume that \Cref{cond:no-effect} and \Cref{assm:second-stage-assumptions} hold. In addition suppose that \(\xi_k^2\log k/n \to 0\) and  \(c_k\to 0\). Then if either the propensity score our outcome regression model are correctly specified: 
    \begin{equation}
        \label{eq:l2-rate-ghat}
        \|\widehat g_k - g_0\|_{L,2} = (\E[(\widehat g(x) - g_0(x))^2])^{1/2} \lesssim_p \sqrt{k/n} + c_k 
    \end{equation}
\end{theorem}

\begin{proof}
    We begin with a matrix law of large numbers from \citet{Rudelson99randomvectors}, which is used to show \(\widehat Q \to_p Q\).
\begin{lemma}[Rudelson's LLN for Matrices]
    \label{lemma:matrix-lln}
    Let \(Q_1,\dots,Q_n\) be a sequence of independent, symmetric, non-negative \(k\times k\) matrix valued random variables with \(k\geq 2\) such that \(Q = \E[\E_n Q_i]\) and \(\|Q_i\| \leq M\) a.s. Then for \( \widehat Q = \E_n[Q_i]\),
    \[
        \Delta := \E\|\widehat Q - Q\| \lesssim \frac{M\log k}{n} + \sqrt{\frac{M\|Q\|\log k}{n}} 
    .\] 
    In particular if \(Q_i = p_ip_i'\) with  \(\|p_i\|\leq \xi_k\) almost surely, then 
    \[
        \Delta := \E\|\widehat Q - Q\| \lesssim \frac{\xi_k^2\log k}{n} + \sqrt{\frac{\xi_k^2\|Q\|\log k}{n}} 
    .\] 
\end{lemma}
Now, to prove \Cref{thm:rate-of-convergence} we have that:
\begin{align*}
    \|\widehat g_k - g_0\|_{L,2}
    &\leq \|p^k(x)'\widehat\beta^k - p^k(x)'\beta^k\|_{L,2} + \|p^k(x)'\beta^k - g\|_{L,2}\\
    &\leq \|p^k(x)'\widehat\beta^k - p^k(x)'\beta^k\|_{L,2} + c_k
\end{align*}
where under the normalization \(Q = I_k\) we have that 
\begin{align*}
    \|p'\widehat\beta - p'\beta\|_{L,2} = \|\widehat\beta - \beta\|
\end{align*}
Further, 
\begin{align*}
    \|\widehat\beta^k - \beta^k\| 
    &= \|\widehat{Q}^{-1}\E[p^k(x)\circ (\widehat Y - \bar Y)]\| + \|\widehat{Q}^{-1}\E_n[p^k(x)\circ (\eps^k + r_k)]\| \\
    &\leq  \|\widehat{Q}^{-1}\E[p^k(x)\circ (\widehat Y - \bar Y)]\| +  \|\widehat{Q}^{-1}\E_n[p^k(x)\circ\eps^k]\| + \|\widehat{Q}^{-1}\E_n[p^k(x)r_k]\|
    \intertext{By the matrix LLN  (\Cref{lemma:matrix-lln}) we have that since \(\xi_k^2\log k/n\to 0\),  \(\|\widehat Q - Q\|\to_p 0\).  This means that with probability approaching one all eigenvalues of \(\widehat Q\) are boundedaway from zero, in particular they are larger than \(1/2\). So w.p.a 1}
    &\lesssim  \|\E[p^k(x)\circ (\widehat Y - \bar Y)]\| +  \|\E_n[p^k(x)\circ\eps^k]\| + \|\E_n[p^k(x)r_k]\|
\end{align*}
Under \Cref{cond:no-effect} the first term is \(o_p(\sqrt{k/n})\). By equation (A.48) in \citet{BCCK-2015} the third term is bounded in probability by \(c_k\). For the second term apply the third condition in \Cref{assm:second-stage-assumptions} to see
\[
    \E\|\E_n[p^k(x)\circ\eps^k]\|^2 = \E\sum_{j=1}^k \eps_j^2p_j(x)^2/n  \leq \bar\sigma^2\E_n[p^k(x)p^k(x)'/n] \lesssim \E[p^k(x)p^k(x)'/n] = k/n
.\] 
This gives \(\|\E_n[p^k(x)\circ \eps^k]\| \lesssim_p \sqrt{k/n}\) and thus shows \eqref{eq:l2-rate-ghat}. 

\end{proof}

The following lemma is a building block for asymptotic pointwise normality. It establishes conditions under which the coefficient estimator \(\widehat\beta\) is asymptotically linear in the sense of \citet{bickel1993efficient}.

\begin{lemma}[Pointwise Linearization]
    \label{lemma:pointwise-linearization}
    Suppose that \Cref{cond:no-effect} and \Cref{assm:second-stage-assumptions}, hold. In addition assume that \(\xi_k^2\log k/n\to 0\). Then for any  \(\alpha \in S^{k-1}\),
    \begin{equation}
        \label{eq:pointwise-linearization-one}
        \sqrt{n}\alpha'(\widehat\beta^k - \beta^k) = \alpha'\mathbb{G}_n[p^k(x)\circ(\eps^k + r_k)] + R_{1n}(\alpha)
    \end{equation}
    where the term \(R_{1n}(\alpha)\), summarizing the impact of unknown design, obeys
    \begin{equation}
        \label{eq:pointwise-linearization-first-error}
        R_{1n}(\alpha) \lesssim_p \sqrt{\frac{\xi_k^2\log k}{n}}(1 + \sqrt{k}\ell_kc_k)
    \end{equation}
    Moreover,
    \begin{equation}
        \label{eq:pointwise-linearization-two}
        \sqrt{n}\alpha'(\widehat\beta^k - \beta^k) = \alpha'\mathbb{G}_n[p^k(x)\circ\eps^k] + R_{1n}(\alpha) + R_{2n}(\alpha)
    \end{equation}
    where the term \(R_{2n}\), summarizing the impact of approximation error on the sampling error of the estimator, obeys
     \begin{equation}
        \label{eq:pointwise-linearization-second-error}
        R_{2n}(\alpha) \lesssim_p \ell_k c_k
    \end{equation}
\end{lemma}

\begin{proof}
    Decompose as before,
\begin{align*}
    \sqrt{n}\alpha'(\widehat\beta^k - \beta^k) &= \sqrt{n}\alpha'\widehat Q^{-1}\E_n[p^k(x)\circ(\widehat Y - \bar Y)]  \\
    &\;\;\;+ \alpha'\mathbb{G}_n[p^k(x)\circ(\eps^k + r_k)] \\
    &\;\;\;\;\;\;+ \alpha'[\widehat{Q}^{-1} - I]\mathbb{G}_n[p^k(x)\circ(\eps^k + r_k)]
.\end{align*} 
The first term is \(o_p(1)\) under \Cref{cond:no-effect}, we can just include this term in \(R_{1n}(\alpha)\). Now bound \(R_{1n}(\alpha)\) and \(R_{2n}(\alpha)\). 

\textbf{Step 1.} Conditional \(X = [x_1,\dots,x_n]\), the term
\[
    \alpha'[\widehat{Q}^{-1} - I]\mathbb{G}_n[p^k(x)\circ\eps^k]
.\] 
has mean zero and variance bounded by \(\bar\sigma^2\alpha'[\widehat{Q}^{-1}-I]\widehat{Q}^{-1}[\widehat{Q}^{-1}-I]\alpha\). Next, by \Cref{lemma:matrix-lln}, with probability approaching one, all eigenvalues of \(\widehat{Q}^{-1}\) are bounded from above and away zero. So,
\[
    \bar\sigma^2\alpha'[\widehat{Q}^{-1}-I_k]\widehat{Q}^{-1}[\widehat{Q}^{-1}-I_k]\alpha \lesssim \bar\sigma^2 \|\widehat{Q}\|\|\widehat{Q}^{-1}\|^2\|\widehat{Q}^{-1}-I_k\|^2 \lesssim_p \frac{\xi_k^2\log k}{n} 
.\] 
so by Chebyshev's inequality,
\[
    \alpha'[\widehat{Q}^{-1}-I]\mathbb{G}_n[p^k(x)\circ \eps^k] \lesssim_p \sqrt{\frac{\xi_k^2\log k}{n}}
.\] 
\textbf{Step 2.} From the proof of Lemma 4.1 in \citet{BCCK-2015}, we get that
\begin{align*}
    \alpha'(\widehat{Q}^{-1} - I_k)\mathbb{G}_n[p^k(x)r_k] \lesssim_p \sqrt{\frac{\xi_k^2\log k}{n}}\ell_k c_k \sqrt{k}
\end{align*}
This completes the bound on \(R_{1n}(\alpha)\) and gives \eqref{eq:pointwise-linearization-one}-\eqref{eq:pointwise-linearization-first-error}. Next, also from the proof of Lemma 4.1 from \citet{BCCK-2015}, 
\[
    R_{2n}(\alpha) = \alpha'\mathbb{G}_n[p^k(x)r_k] \lesssim_p \ell_kc_k
,\] 
which gives \eqref{eq:pointwise-linearization-two}-\eqref{eq:pointwise-linearization-second-error}.

\end{proof}

The following lemma shows that, after adding \Cref{assm:uniform-limit-theory} the linearization of our coefficient estimator \(\widehat\beta^k\) established in \Cref{lemma:pointwise-linearization} holds uniformly over all points \(x\in\calX\). That is to say the error from linearization is bounded in probability uniformly over all \(x\in\calX\). It will form an important building block in uniform consistency and strong approximation results presented in Theorems~\ref{thm:uniform-convergence} and \ref{thm:strong-approximation}.

\begin{lemma}[Uniform Linearization]
    \label{lemma:uniform-linearization}
    Suppose that \Cref{cond:no-effect} and \Cref{assm:second-stage-assumptions}-\ref{assm:uniform-limit-theory} hold. Then if either the propensity score model our outcome regression model is correctly specified:  
    \begin{equation}
        \label{eq:uniform-linearization-1}
        \sqrt{n}\alpha(x)'(\widehat\beta^k-\beta^k) = \alpha(x)'\mathbb{G}_n[p^k(x)\circ(\eps^k + r_k)] + R_{1n}(\alpha(x))
    \end{equation}
    where \(R_{1n}(\alpha(x))\) describes the design error and satisfies
     \begin{equation}
        \label{eq:uniform-linearization-r1n-bound}
        R_{1n}(\alpha(x)) \lesssim_p \sqrt{\frac{\xi_k^2\log k}{n}}(n^{1/m}\sqrt{\log k} + \sqrt{k}\ell_kc_k) := \bar R_{1n}
    \end{equation}
    uniformly over \(x\in\calX\). Moreover, 
     \begin{equation}
        \label{eq:uniform-linearization-2}
        \sqrt{n}\alpha(x)'(\widehat\beta^k -\beta^k) = \alpha(x)'\mathbb{G}_n[p^k(x)\circ\eps^k] + R_{1n}(\alpha(x)) + R_{2n}(\alpha(x))
    \end{equation}
    where \(R_{2n}(\alpha(x))\) describes the sampling error and satisfies, uniformly over \(x\in\calX\):
    \begin{equation}
        \label{eq:uniform-linearization-r2n-bound}
        R_{2n}(\alpha(x)) \lesssim_P \sqrt{\log k}\cdot\ell_kc_k := \bar R_{2n}
    \end{equation}
\end{lemma}

\begin{proof}
    As in the proof of \Cref{lemma:pointwise-linearization}, we decompose
    \begin{equation}
        \label{eq:uniform-linearization-decomposition}
        \begin{split}
            \sqrt{n}\alpha(x)'(\widehat\beta^k - \beta^k) &= \sqrt{n}\alpha(x)'\widehat Q^{-1}\E_n[p^k(x)\circ(\widehat Y - \bar Y)]  \\
            &\;\;\;+ \alpha(x)'\mathbb{G}_n[p^k(x)\circ(\eps^k + r_k)] \\
            &\;\;\;\;\;\;+ \alpha(x)'[\widehat{Q}^{-1} - I]\mathbb{G}_n[p^k(x)\circ(\eps^k + r_k)]
    .\end{split} 
    \end{equation}
    Using \Cref{cond:no-effect}, the matrix LLN (\Cref{lemma:matrix-lln}), and bounded eigenvalues of the design matix, we have that: 
    \begin{align*}
        \sup_{x\in\calX} \sqrt{n}\alpha(x)'\widehat{Q}^{-1}\E_n[p^k(x)\circ(\widehat Y - \bar Y)] = o_p(1)
    .\end{align*} 
    Since this is \(o_p(1)\), we can simply include this term in \(R_{1n}(\alpha(x))\). Now derive bounds on \(R_{1n}(\alpha(x))\) and \(R_{2n}(\alpha(x))\).

     \textbf{Step 1:} Conditional on the data let 
     \[
        T := \left\{t = (t_1,\dots,t_n) \in\SR^n : t_i = \alpha(x)'(\widehat{Q}^{-1} - I)p^k(x)\circ\eps^k, x\in\calX \right\} 
     .\] 
     Define the norm \(\|\cdot\|_{n,2}\) on \( \SR^n\) by \(\|t\|_{n,2}^2 = n^{-1}\sum_{i=1}^n t_i^2\). For an \(\varepsilon > 0\) an  \(\varepsilon\)-net of the normed space  \((T,\|\cdot\|_{n,2})\) is a subset \(T_\varepsilon\) of  \(T\) such that for every  \(t\in T\) there is a point \(t_\varepsilon\in T_\varepsilon\) such that  \(\|t-t_\varepsilon\|_{n,2} < \varepsilon\).  The covering number \(N(T, \|\cdot\|_{n,2},\varepsilon)\) of \(T\) is the infimum of the cardinality of  \(\varepsilon\)-nets of  \(T\).

     Let  \(\eta_1,\dots,\eta_n\) be independent Rademacher random variables that are independent of the data. Let \(\eta = (\eta_1,\dots, \eta_n)\). Let \(\E_\eta[\cdot]\) denote the expectation with respect to the distribution of \(\eta\). By Dudley's inequality \citep{Dudley-1967},
      \[
         \E_\eta\left[\sup_{x\in\calX} \left|\alpha(x)'[\widehat{Q}^{-1}- I]\mathbb{G}_n[\eta_ip^k(x)\circ\eps^k]\right|\right] \lesssim \int_{0}^{\theta} \sqrt{\log N(T, \|\cdot\|_{n,2}, \varepsilon)}\,d\varepsilon
     .\] 
     where 
     \begin{align*}
         \theta &:= 2\sup_{t\in T} \|t\|_{n,2} \\
                &= 2\sup_{x\in\calX} \left(\E_n[(\alpha(x)'(\widehat{Q}^{-1}-I)p^k(x)\circ \eps^k)^2]\right)^{1/2} \\
                &\leq 2 \max_{1 \leq i \leq n}|\bar\eps_{k,i}|\|\widehat{Q}^{-1} - I\|\|\widehat Q\|^{1/2},
     \end{align*}
     by \eqref{eq:hadamard-rough-bound}. Now, for any \(x\in\calX\), 
     \begin{align*}
         \bigg(\E_n[&(\alpha(x)'(\widehat{Q}^{-1} - I)p^k(x)\circ\eps^k - \alpha(\tilde x)'(\widehat{Q}^{-1} - I)p^k(x)\circ\eps^k)^2]\bigg)^{1/2} \\  
                    &\leq \max_{1\leq i \leq n}|\bar\eps_{k,i}|\|\alpha(x) - \alpha(\tilde x)\|\|\widehat{Q}^{-1} - I\|\|\widehat{Q}\|^{1/2} \\
                    &\leq \xi_k^L \max_{1\leq i \leq n}|\bar\eps_{k,i}|\|\widehat{Q}^{-1}- I\|\|\widehat Q\|^{1/2}\|x - \tilde x\|
     \end{align*}
     So, for some \(C > 0\), 
     \[
         N(T, \|\cdot\|_{n,2}, \varepsilon) \leq \left(\frac{C\xi_k^L\max_{1\leq i\leq n}|\bar\eps_{k,i}|\|\widehat{Q}^{-1}-I\|\|\widehat Q\|^{1/2}}{\varepsilon}\right)^{d_x}
     .\] 
     This gives us that 
     \[
         \int_0^\theta \sqrt{\log(N(T,\|\cdot\|_{2,n},\varepsilon))}\,d\varepsilon \leq \max_{1\leq i\leq n}|\bar\eps_{k,i}|\|\widehat{Q}^{-1}-I\|\|\widehat Q\|^{1/2}\int_0^2\sqrt{d_x\log(C\xi_k^L/\varepsilon)}\,d\varepsilon
     .\] 
     By \Cref{assm:uniform-limit-theory} we have that \(\E[\max_{1\leq i\leq n}|\bar\eps_{k,i}|\mid X ]\lesssim_P n^{1/m}\) where \(X = (x_1,\dots,x_n)\). In addition \(\xi_k^{2m/(m-2)}\log k/n\lesssim 1\) for  \(m> 2\) gives that \(\xi_k^2/\log k/n\to 0\). So, \(\|\widehat{Q}^{-1} - I\|\lesssim_P (\xi_k^2\log k/n)^{1/2}\) and \(\|\widehat{Q}^{-1}\| \lesssim_P 1\). Combining this all with \(\log\xi_k^L\lesssim \log k\) implies 
     \begin{align*}
         \E\left[\sup_{x\in\calX}\big|\alpha(x)'[\widehat{Q}^{-1}-I]\mathbb{G}_n[p^k(x)\circ\eps^k]\big|\mid X\right]  
         &\leq 2\E\left[\E_\eta\sup_{x\in\calX}\big|\alpha(x)'[\widehat{Q}^{-1}-I]\mathbb{G}_n[\eta_ip^k(x)\circ\eps^k]\big|\mid X\right] \\
         &\lesssim_P n^{1/m}\sqrt{\frac{\xi_k^2\log^2k}{n}}
     \end{align*}
     where the first line is due to symmetrization inequality. This gives us
     \begin{equation}
         \label{eq:uniform-r1n-error-first-part}
         \sup_{x\in\calX} \big|\alpha(x)'[\widehat{Q}^{-1}-I]\mathbb{G}_n[p^k(x)\circ \eps^k]\big| \lesssim_p n^{1/m}\sqrt{\frac{\xi_k^2\log^2k}{n}}
     \end{equation}

     \textbf{Step 2:} Now simply report the results on approximation error from \citet{BCCK-2015} . Since the approximation error is the same for all signals \(Y(\bar\pi_k, \bar m_k)\), there is no Hadamard product to deal with. 
     \begin{align}
         \label{eq:uniform-r1n-error-second-part}
         \sup_{x\in\calX} \big|\alpha(x)'[\widehat{Q}^{-1}-I]\mathbb{G}_n[p^k(x)r_k]\big| 
         &\lesssim_P \sqrt{\frac{\xi_k^2\log k}{n}}\ell_kc_k\sqrt{k}\\
         \label{eq:uniform-r2n-error}
         \sup_{x\in\calX}\big|\alpha(x)'\mathbb{G}_n[p^k(x)r_k]\big|
         &\lesssim_P \ell_kc_k\sqrt{\log k}
     \end{align}
     Looking at \eqref{eq:uniform-linearization-decomposition} and combining \eqref{eq:uniform-r1n-error-first-part}-\eqref{eq:uniform-r1n-error-second-part} gives the bound on \(R_{1n}(\alpha(x))\) while \eqref{eq:uniform-r2n-error} gives the bound on \(R_{2n}(\alpha(x))\). 

\end{proof}

\Cref{thm:uniform-convergence} gives conditions under which our estimator converges in probability to the true conditional counterfactual outcome \(g_0(x)\). In particular, this convergence happens uniformly at the rates defined in \eqref{eq:uniform-convergence-beta-bound}-\eqref{eq:uniform-convergence-ghat-bound}. If these two terms go to zero, the entire estimator will converge uniformly to the true conditional expectation of interest.

\begin{theorem}[Uniform Rate of Convergence]
    \label{thm:uniform-convergence}
    Suppose that \Cref{cond:no-effect} and Assumptions~\ref{assm:second-stage-assumptions}-\ref{assm:uniform-limit-theory} hold. Then so long as either the propensity score model or outcome regression model is correctly specified:
    \begin{equation}
        \label{eq:uniform-convergence-first-bound}
        \sup_{x\in\calX}\big|\alpha(x)'\mathbb{G}_n[p^k(x)\circ\eps^k]\big| \lesssim_P \sqrt{\log k}
    \end{equation}
    Moreover, for 
    \begin{align*}
        \bar R_{1n} &:= \sqrt{\frac{\xi_k^2\log k}{n}}(n^{1/m}\sqrt{\log k} + \sqrt{k}\ell_k c_k) \\
        \bar R_{2n} &:= \sqrt{\log k}\cdot\ell_kc_k
    \end{align*}
    we have that
    \begin{equation}
        \label{eq:uniform-convergence-beta-bound}
        \sup_{x\in\calX} \big| p^k(x)'(\widehat\beta^k - \beta^k)\big| \lesssim_P \frac{\xi_k}{\sqrt{n}}\left(\sqrt{\log k} + \bar R_{1n} + \bar R _{2n}\right) 
    \end{equation}
    and
    \begin{equation}
        \label{eq:uniform-convergence-ghat-bound}
        \sup_{x\in\calX} \big|\widehat g(x) - g_0(x)\big| \lesssim_P \frac{\xi_k}{\sqrt{n}}\left(\sqrt{\log k} + \bar R_{1n} + \bar R_{2n}\right) + \ell_k c_k 
    \end{equation}
\end{theorem}
\begin{proof}
    The goal will be to apply the following two theorems from \citet{GK-2006} and \citet{vdvw-1996}. 

\begin{tcolorbox}[title = Preliminaries for Proof of \Cref{thm:uniform-convergence}]
\begin{theorem*}[Gine and Koltchinskii, 2006]
    Let \(\xi_1,\dots,\xi_n\) be \(i.i.d\) random variables taking values in a measurable space  \((S, \mathscr{S})\) with a common distribution \(P\) defined on the underlying  \(n\)-fold product space. Let  \(\calF\) be a measurable class of functions mapping  \(S \to \SR\) with a measurable envelope \(F\). Let  \(\sigma^2\) be a constant such that  \(\sup_{f\in\calF}\Var(f) \leq \sigma^2 \leq \|F\|_{L^2(P)}^2\). Suppose there exist constats \(A > e^2\) and  \(V \geq 2\) such that \(\sup_Q N(\calF, L^2(Q), \varepsilon\|F\|_{L^2(Q)}) \leq (A/\varepsilon)^V\) for all \(0 <\varepsilon\leq 1\). Then
    \begin{equation}
        \label{eq:maximal-inequality-GK}
        \tag{GK}
        \begin{split}
        \E\bigg[\bigg\|\sum_{i=1}^n \{f(\xi_i) - \E[f(\xi_1)]\} \bigg\|_\calF\bigg] \leq C\left[\sqrt{n\sigma^2V\log\frac{A\|F\|_{L^2(P)}}{\sigma}} + V\|F\|_\infty\log\frac{A\|F\|_{L^2(P)}}{\sigma} \right]
        \end{split}
    .\end{equation} 
    where \(C\) is a universal constant.
\end{theorem*}
\tcblower
\begin{theorem*}[VdV\&W 2.14.1]
    Let \(\calF\) be a  \(P\)-measurable class of measurable functions with a measurable envelope function  \(F\). Then for any \(p\geq 1\), 
    \begin{equation}
        \label{eq:vdvw-2.14.1}
        \tag{VW}
        \left\|\|\mathbb{G}_n\|_\calF^*\right\|_{P,p} \lesssim \|J(\theta_n, \calF)\|F\|_n\|_{P,p} \lesssim J(1,\calF) \|F\|_{P, 2\vee p}
    \end{equation}
    where \(\theta_n = \left\|\|f\|_n\right\|_\calF^*/\|F\|_n\), where \(\|\cdot\|_n\) is the \(L_2(\P_n)\) seminorm and the inequalities are valid up to constants depending only on the  \(p\) in the statement. The term \(J(\cdot, \cdot)\) is given
    \[
        J(\delta, \calF) = \sup_Q\int_0^\delta \sqrt{1 + \log N(\calF, \|\cdot\|_{L^2(Q)},\varepsilon\|F\|_{L^2(Q)})}\,d\varepsilon
    .\] 
\end{theorem*}
\end{tcolorbox}
We would like to apply these theorems to bound \(\sup_{x\in\calX} |\alpha(x)'\mathbb{G}_n[p^k(x)\circ\eps^k]|\) and thus show \eqref{eq:uniform-convergence-first-bound}. The other two statements of \Cref{thm:uniform-convergence} follow from this. To this end, let's consider the class of functions 
\[
    \calG := \{(\eps^k, x) \mapsto \alpha(v)'(p^k(x)\circ\eps^k), v\in\calX\} 
.\] 
Let's note that \(|\alpha(v)'p^k(x)| \leq \xi_k\), \(\Var(\alpha(v)'p^k(x)) = 1\), and for any \(v, \tilde v\in \calX\)
 \[
    |\alpha(v)'(p^k(x)\circ\eps^k) - \alpha(\tilde v)'(p^k(x)\circ\eps^k)| \leq |\bar\eps_k|\xi_k^L\xi_k\|v-\tilde v\|
,\] 
where \(\bar\eps_k = \|\eps^k\|_\infty\). Then, taking \(G(\eps^k, x) \leq \bar\eps_k\xi_k\)  we have that 
\begin{equation}
    \label{eq:bracketing-bound-uniform-convergence}
    \sup_Q N(\calG, L^2(Q), \varepsilon\|G\|_{L^2(Q)}) \leq \left(\frac{C\xi_k^L}{\varepsilon}\right)^d
.\end{equation} 

Now, for a \(\tau \geq 0\) specified later define \(\eps_k^- = \eps^k\bm{1}\{|\bar\eps_k| \leq \tau\} - \E[\eps^k\bm{1}\{|\bar\eps_k| \leq \tau\}\mid X] \) and \(\eps_k^+ = \eps^k\bm{1}\{|\bar\eps_k| > \tau\} - \E[\eps^k\bm{1}\{|\bar\eps_k| >\tau\}\mid X ]\). Since \(\E[\eps^k \mid X ]=0\) we have that \(\eps^k = \eps_k^- + \eps_k^+\). Using this decompose:
 \[
    \frac{1}{\sqrt{n}}\sum_{i=1}^n \alpha(v)'(p^k(x)\circ\eps^k) = \sum_{i=1}^n \alpha(v)'(p^k(x)\circ\eps_k^-)/\sqrt{n} + \sum_{i=1}^n \alpha(v)'(p^k(x)\circ\eps_k^+)/\sqrt{n}
.\] 
We deal with each of these terms individually, in two steps.

\textbf{Step 1:} 
For the first term, we set up for an application of \eqref{eq:maximal-inequality-GK}. \Cref{eq:bracketing-bound-uniform-convergence} gives us the constants \(A = C\xi_k^L\) and  \(V = d_x \vee 2\). To get \(\sigma^2\) note that  for any \(v\in\calX\), 
\begin{align*}
    \Var(\alpha(v)'(p^k(x)\circ\eps_k^-)/\sqrt{n}) 
    &\leq \E[(\alpha(v)'(p^k(x)\circ \eps_k^-)/\sqrt{n})^2] \\ 
    &\leq \frac{1}{n}\E[(\alpha(v)'p^k(x))^2]\sup_{x\in\calX}\E[\|\eps_k^-\|_\infty^2\mid X = x]\\
    &\leq \frac{\bar\sigma_k^2\wedge \tau^2}{n} 
\end{align*}
Finally note that we can take the envelope \(G = \|\eps_k^-\|_\infty\xi_k/\sqrt{n}\) where \(\|G\|_{L^2(P)}\leq \frac{\bar\sigma_k\wedge \tau}{\sqrt{n}}\) and \(\|G\|_\infty \leq \tau\xi_k/\sqrt{n}\).

We can now apply \eqref{eq:maximal-inequality-GK} to get that
\begin{align*}
   \E[\sup_{x\in\calX}|\alpha(x)'\mathbb{G}_n[p^k(x)\circ\eps_k^-]|] 
   &\lesssim \sqrt{\bar\sigma_k^2\wedge\tau^2\log(\xi_k^L)} + \frac{\tau\xi_k \log(\xi_k^L)}{\sqrt{n}} 
.\end{align*}

\textbf{Step 2:}
For the second term, we set up for an application of \eqref{eq:vdvw-2.14.1} with the envelope function \(G = \|\eps_k^+\|_\infty\xi_k/\sqrt{n}\) and note that 
\[
    \E[\|\eps_k^+\|_\infty^2] \leq \E[\bar\eps_k^2\bm{1}\{|\bar\eps_k| > \tau\} ] \leq \tau^{-m+2}\E[|\bar\eps_k|^m] 
\]
We can now use \eqref{eq:vdvw-2.14.1} to bound
\begin{align*}
    \E\left\|\sup_{x\in \calX}|\alpha(x)'\mathbb{G}_n[p^k(x)\circ\eps_k^+]|\right\|
    &\lesssim \sqrt{\E[|\bar\eps_k|^m]}\tau^{-m/2+1}\xi_k\int_0^1\sqrt{\log(\xi_k^L/\varepsilon)}\,d\varepsilon \\
    &\lesssim \sqrt{\sigma_k^m}\tau^{-m/2+1}\xi_k\sqrt{\log(\xi_k^L)}
.\end{align*} 

\textbf{Step 3:}
Let \(\tau = \xi_k^{2/(m-2)}\) and apply Markov's inequality. The bounds from step one and two become
\begin{align*}
    \sup_{x\in\calX}|\alpha(x)'\mathbb{G}_n[p^k(x)\circ\eps_k^-]| 
   &\lesssim_P \sqrt{\bar\sigma_k^2\log(\xi_k^L)} + \frac{\xi_k^{2m/(m-2)} \log(\xi_k^L)}{\sqrt{n}} \\
    \sup_{x\in \calX}|\alpha(x)'\mathbb{G}_n[p^k(x)\circ\eps_k^+]|
    &\lesssim_P \sqrt{\bar\sigma_k^m\log(\xi_k^L)}
\end{align*}
Applying \Cref{assm:uniform-limit-theory} along with the inequality 
\[
    \frac{\xi_k^{m/(m-2)}\log k}{\sqrt{n}} = \sqrt{\log k}
    \sqrt{\frac{\xi_k^{2m/(m-2)}\log k}{n}}\lesssim \log k 
\]
completes the proof.

\end{proof}

\begin{theorem}[Validity of Gaussian Bootstrap]
    \label{thm:gaussian-bootstrap}
    Suppose that the assumptions of \Cref{thm:strong-approximation} hold with \(a_n = \log n\) and the assumptions of \Cref{thm:matrix-estimation} hold with  \(a_n = O(n^{-b})\) for some  \(b > 0\). In addition, suppose that there exists a sequence \(\xi_n'\) obeying  \(1 \lesssim \xi_n' \lesssim \|p^k(x)\|\) uniformly for all \(x\in\calX\) such that  \(\|p^k(x) - p^k(x')\|/\xi_n' \leq L_n \|x - x'\|\), where \(\log L_n\lesssim \log n\). Let  \(N_k^b\) be a bootstrap draw from \(N(0,I_k)\) and  \(P^\star\) be the distribution conditional on the observed  data \(\{Y_i, D_i, Z_i\}_{i=1}^n \). Then the following approximation holds uniformly in  \(\ell^\infty(\calX)\):
    \begin{equation}
        \label{eq:gaussian-bootstrap}
        \frac{p^k(x)'\widehat\Omega^{1/2}}{\widehat\Omega^{1/2}p^k(x)}N_k^b =^d \frac{p^k(x)'\Omega^{1/2}}{\|\Omega^{1/2}p^k(x)\|}  + o_{P^\star}(\log^{-1}N)
    \end{equation}
\end{theorem}

\begin{proof}
    See Theorem 3.4 in \citet{SC-2020}. 

\end{proof}

\section{Additional Details on Empirical Application}
\label{sec:empirical-details}

As mentioned in the setup, to avoid outlier contamination we drop the top 3\% and bottom 3\% of birthweights by maternal age. We also drop ages for which there are fewer than 10 smoker or non smoker observations. The result is a dataset with 4107 (of an initial 4602) observations on the outcome variable, birthweight. In addition to the 21 control variables \((Z)\) available in the dataset, we further generate an additional 11 interaction/higher order variables that we believe may be useful in controlling for confounding. \Cref{tab:data-descr} provides a final summary of the data after our cleaning process. The generated variables represent the bottom 11 variables in \Cref{tab:data-descr}.\footnote{This table is generated using the wonderful stargazer package in \texttt{R} \citep{stargazer}.}

\begin{table}[!htb] \centering \caption{Summary of Data used in Emprical Exercise} 
  \label{tab:data-descr} 
\begin{tabular}{@{\extracolsep{5pt}}lccccc} 
\\[-1.8ex]\hline 
\hline \\[-1.8ex] 
Statistic & \multicolumn{1}{c}{N} & \multicolumn{1}{c}{Mean} & \multicolumn{1}{c}{St. Dev.} & \multicolumn{1}{c}{Min} & \multicolumn{1}{c}{Max} \\ 
\hline \\[-1.8ex] 
bweight & 4,107 & 3,384.354 & 447.616 & 1,544 & 4,668 \\ 
mmarried & 4,107 & 0.708 & 0.455 & 0 & 1 \\ 
mhisp & 4,107 & 0.034 & 0.181 & 0 & 1 \\ 
fhisp & 4,107 & 0.038 & 0.192 & 0 & 1 \\ 
foreign & 4,107 & 0.054 & 0.226 & 0 & 1 \\ 
alcohol & 4,107 & 0.031 & 0.174 & 0 & 1 \\ 
deadkids & 4,107 & 0.252 & 0.434 & 0 & 1 \\ 
mage & 4,107 & 26.125 & 5.025 & 16 & 36 \\ 
medu & 4,107 & 12.703 & 2.470 & 0 & 17 \\ 
fage & 4,107 & 27.000 & 9.022 & 0 & 60 \\ 
fedu & 4,107 & 12.324 & 3.624 & 0 & 17 \\ 
nprenatal & 4,107 & 10.822 & 3.613 & 0 & 40 \\ 
monthslb & 4,107 & 21.938 & 30.255 & 0 & 207 \\ 
order & 4,107 & 1.858 & 1.056 & 0 & 12 \\ 
msmoke & 4,107 & 0.390 & 0.890 & 0 & 3 \\ 
mbsmoke & 4,107 & 0.183 & 0.386 & 0 & 1 \\ 
mrace & 4,107 & 0.847 & 0.360 & 0 & 1 \\ 
frace & 4,107 & 0.822 & 0.382 & 0 & 1 \\ 
prenatal & 4,107 & 1.204 & 0.507 & 0 & 3 \\ 
birthmonth & 4,107 & 6.556 & 3.352 & 1 & 12 \\ 
lbweight & 4,107 & 0.025 & 0.155 & 0 & 1 \\ 
fbaby & 4,107 & 0.443 & 0.497 & 0 & 1 \\ 
prenatal1 & 4,107 & 0.803 & 0.398 & 0 & 1 \\ 
mbsmoke \textasteriskcentered  alcohol & 4,107 & 0.017 & 0.128 & 0 & 1 \\ 
medu \textasteriskcentered  fedu & 4,107 & 161.518 & 64.291 & 0 & 289 \\ 
mage \textasteriskcentered  fage & 4,107 & 730.422 & 328.522 & 0 & 2,088 \\ 
msmoke$\hat{\mkern6mu}$2 & 4,107 & 0.944 & 2.422 & 0 & 9 \\ 
msmoke \textasteriskcentered  alcohol & 4,107 & 0.037 & 0.302 & 0 & 3 \\ 
mage$\hat{\mkern6mu}$2 & 4,107 & 707.741 & 262.383 & 256 & 1,296 \\ 
mage \textasteriskcentered  mmarried & 4,107 & 19.570 & 13.090 & 0 & 36 \\ 
mage \textasteriskcentered  medu & 4,107 & 336.911 & 108.588 & 0 & 612 \\ 
mage \textasteriskcentered  fedu & 4,107 & 328.405 & 128.438 & 0 & 612 \\ 
monthslb$\hat{\mkern6mu}$2 & 4,107 & 1,396.431 & 3,509.883 & 0 & 42,849 \\ 
msmoke \textasteriskcentered  monthslb$\hat{\mkern6mu}$2 & 4,107 & 750.703 & 4,407.398 & 0 & 112,908 \\ 
\hline \\[-1.8ex] 
\end{tabular} 
\end{table}

In conducting analysis, we found it quite helpful to the stability of the final model assisted estimator to do some light trimming of the estimated propensity score and outcome regression models. In particular we trim the estimated propensity score(s) to be between \(0.01\) and  \(0.99\) and trim the estimated mean regression models so that they take a value no more than roughly 12.5\% higher or lower than the maximum or minimum value of \(Y\) observed in the data. 

Because the control variables are all of  different magnitudes, it is common to do some normalization before estimating the \(\ell_1\)-regularized propensity score and outcome regression models so that all variables are ``punished'' equally by the penalty. We normalize our data by scaling each variable to take on values between zero and one. 

In addition to the results presented in \Cref{sec:empirical} we present some additional specifications below. \Cref{fig:K4_D0} presents results from using a local constant regression with 3 knots in the first stage while \Cref{fig:K3_D2}. Both show that the results of analysis generally hold up in a variety of specifications, namely that the effect of maternal smoking on birthweight is negative and increasing in magnitude with age. The model-assisted estimated typically produce values that are more in line with previous work, though the shape of CATE estimates that use standard loss functions in the first stage are more stable to second stage basis.

\begin{figure}[!htp]
    \centering
    \includegraphics[width=\linewidth]{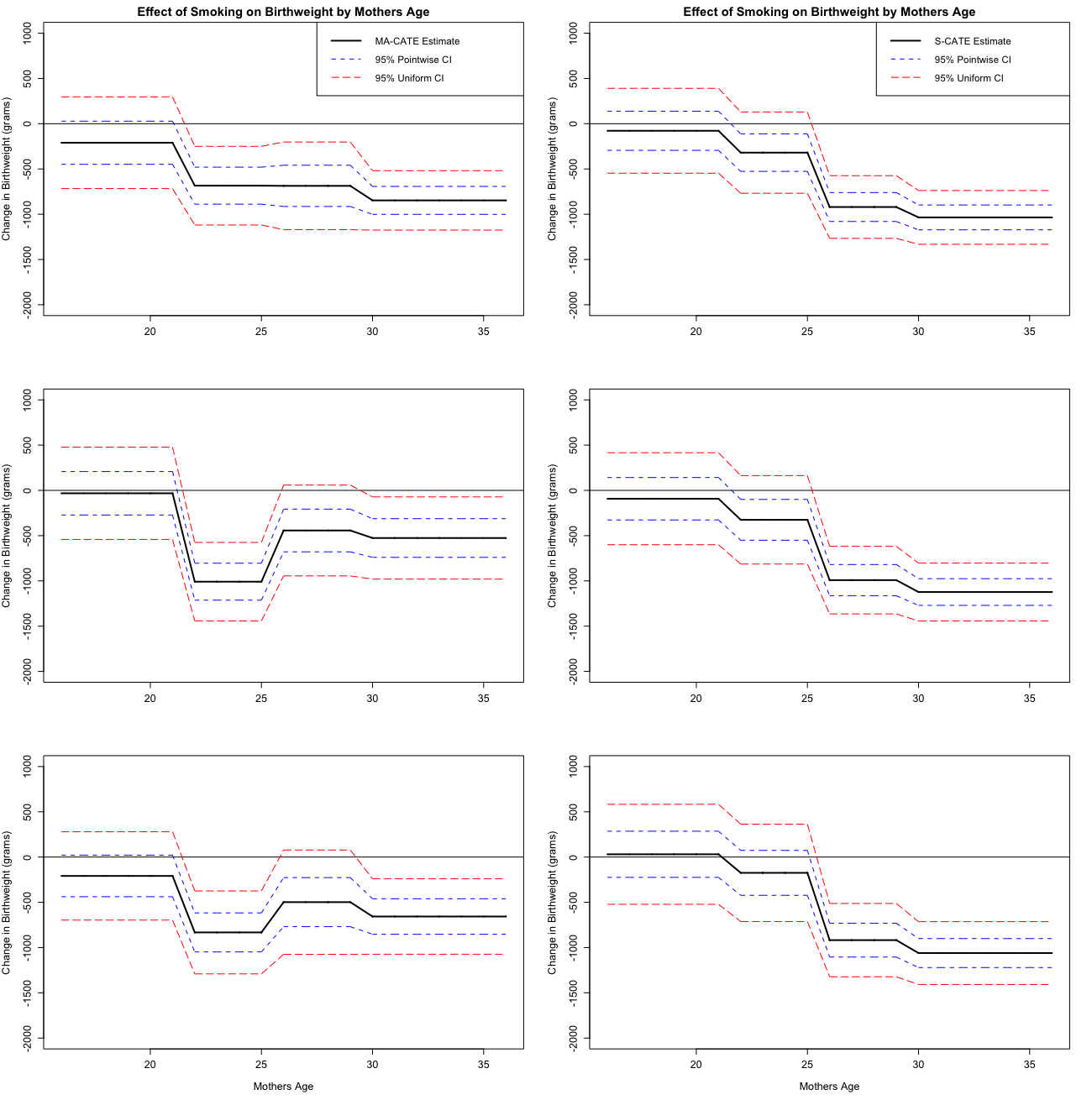}
    \caption{CATE of maternal smoking estimated using model assisted estimating equations (left) and standard MLE/OLS estimating equations (right). Top row uses the 99\textsuperscript{th} quantile of the bootstrap distribution to select the penalty parameters, second row uses 95\textsuperscript{th} quantile, and final row uses the 90\textsuperscript{th} quantile. Second stage is computed using a local constant with 3 knots. 95\% pointwise confidence intervals are displayed in blue short dashes and 95\% uniform confidence bands are displayed in long red dashes.}%
    \label{fig:K4_D0}
\end{figure}

\begin{figure}[!htp]
    \centering
    \includegraphics[width=\linewidth]{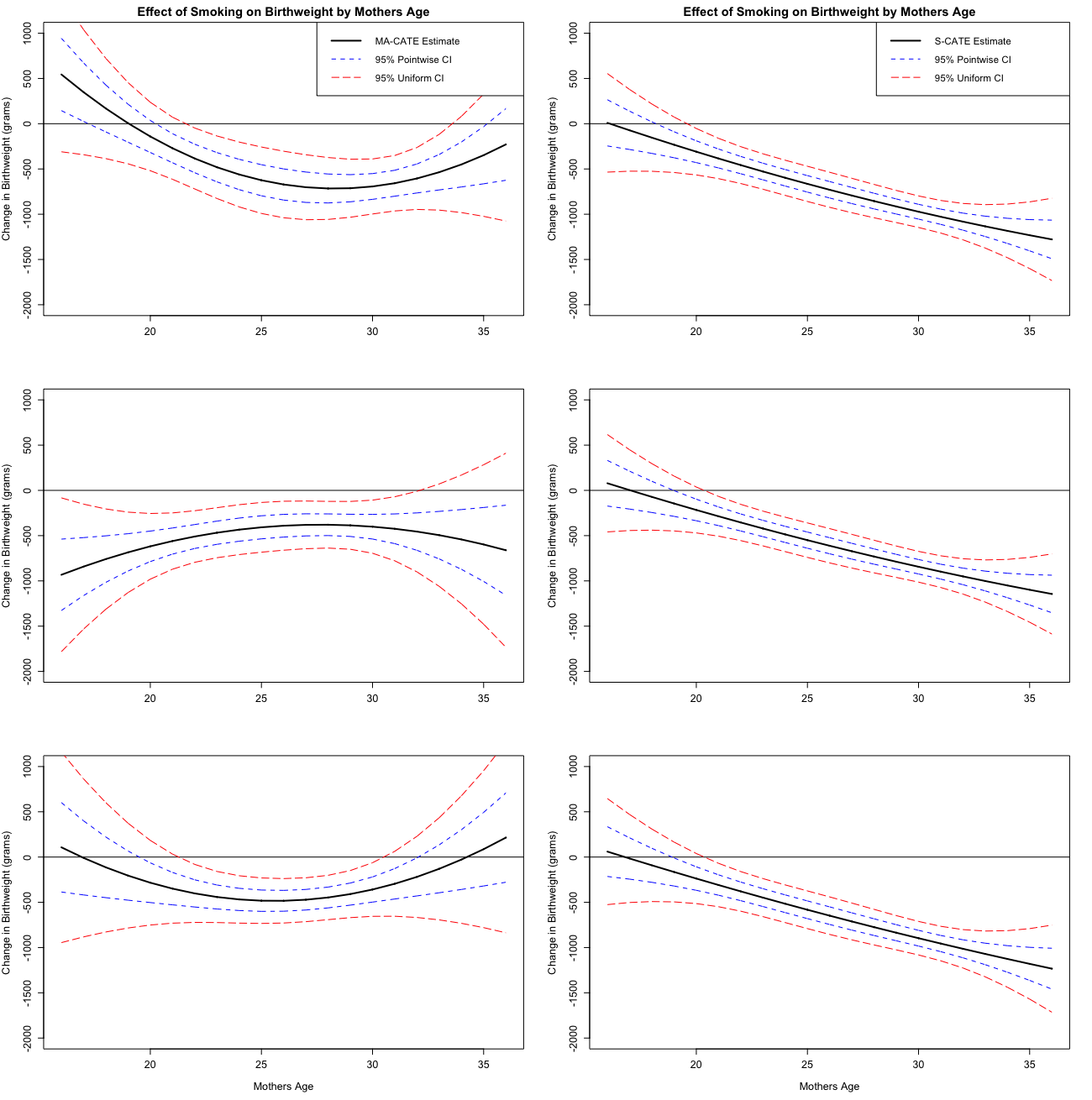}
    \caption{CATE of maternal smoking estimated using model assisted estimating equations (left) and standard MLE/OLS estimating equations (right). Top row uses the 99\textsuperscript{th} quantile of the bootstrap distribution to select the penalty parameters, second row uses 95\textsuperscript{th} quantile, and final row uses the 90\textsuperscript{th} quantile. Second stage is computed using second degree b-splines with 2 knots. 95\% pointwise confidence intervals are displayed in blue short dashes and 95\% uniform confidence bands are displayed in long red dashes.}%
    \label{fig:K3_D2}
\end{figure}

\newpage
\section{Consistency between First Stage and Second Stage Assumptions}
\label{sec:first-stage-second-stage}

In this section, we examine the consistency between the first stage and second stage assumptions on the basis terms \(p^k(x)\). In particular, we are interested in finding a positive basis that also satisfies the bounded eigenvalue condition on the design matrix in \Cref{assm:second-stage-assumptions}. We also discuss how to construct the model assisted estimator with weights in \eqref{eq:gamma-j-estimating-equation}-\eqref{eq:alpha-j-estimating-equation} that are not directly the second stage basis terms in case the researcher is worried about their choice of basis terms satisfying the first stage and second stage stage assumptions simultaneously.

Suppose that \(\calX = [0,1]\). First, note that the first stage non-negativity and second stage design assumptions can be trivially satisfied by using a locally constant basis; that is by taking 
\begin{equation}
    \label{eq:local-constant}
    p_j(x) = \bm{1}_{[\ell_{j-1}, \ell_j)}(x)
\end{equation} 
for some \(0 = \ell_0 < \ell_1 < \dots < \ell_t = 1\). While this basis may have poor approximation qualities, the general principle can be extended to any basis whose elements have disjoint (or limitedly overlapping) supports. Higher order piecewise polynomial approximations can often be implemented using \emph{B-splines} which are orthonormalized regression splines. See \citet{DeBoor-2001} for an in-depth discussion or \citet{Newey-1997} for an application of B-splines to nonparametric series regression.

These higher order splines can be defined recursively. For a given (weakly increasing) knot sequence \(\ell := (\ell_j)_{j=1}^{t}\) we define the ``first-order'' B-splines denoted \(B_{1,1}(x),\dots,B_{t,1}(x)\) using \eqref{eq:local-constant}, that is \(B_{j,1}(x) = p_j(x)\).  On top of these functions, we can define higher order B-splines via the recursive relation (\citet{DeBoor-2001}, p.90)
\begin{equation}
    \label{eq:recursive-relationship}
    B_{j,d+1} := \omega_{j,d}(x)B_{j,d}(x) + [1 - \omega_{j+1,d}(x)]B_{j+1,d}(x) 
.\end{equation}
where 
\[
    \omega_{j,d}(x) := \begin{cases}
    \frac{x - \ell_j}{\ell_{j+d} - \ell_j} 
         & \text{if }\ell_{j+d}\neq \ell_j \\ 0 & \text{otherwise}
    \end{cases}
.\] 
If \(X\) is continuously distributed on an open set containing the knots  \((\ell_j)\), \citet{DeBoor-2001} shows that the  B-spline basis is almost surely positive. Moreover, B-splines is locally supported in the sense each \(B_{j,d}\) is positive on \((\ell_j, \ell_{j + d})\), zero off this support and for each \(d\):
\[
    \sum_{j=1}^{t} B_{j,d} = 1\;\;\;\;\text{ on }[0,1]
.\] 
where the summation is taken pointwise (see \citet{DeBoor-2001}, p.36). 
From the final property we can see the B-spline basis using \(k = td\) basis terms, \(p^k(x) = (B_{j,l}(x))_{\substack{j=1,..,t \\l = 1,\dots,d}}\) are totally bounded so that.

B-splines used directly in this manner, however, do not lead to a design matrix \(Q = \E[p^k(x)p^k(x)']\) with eigenvalues which are bounded away from zero. To achieve this, the basis fucntions must be divided by their \(\ell_2\) norm. In practice, this leads to b-spline terms who are grown at rate \(\xi_{k,\infty}\lesssim \sqrt{k}\). The pilot penalty constants can be chosen from a set whose bounds are on the order of \(\sqrt{k}\) and the sparsity bounds of \Cref{assm:logistic-model-convergence} reduce to
\[
    \frac{s_k\,k^{3/2}\ln^5(d_zn)}{n}\to 0 \andbox \frac{k^2\ln^7(d_zkn)}{n}\to 0
\]
while the bounds in \eqref{eq:first-stage-sparsity-bound} and \eqref{eq:variance-sparsity-bound} reduce respectively to 
\[
    \frac{s_k\,k^{3/2}\ln(d_z)}{\sqrt{n}}\to 0\andbox \frac{s_k^2\,k^{7/2}\ln(d_z)}{n^{(m-1)/m}} \to 0
.\]

\subsection{Alternate Weighting}

So long as the second stage basis \(p^k(x)\) contains a constant term, it is possible to weight the estimating equations \eqref{eq:gamma-j-estimating-equation}-\eqref{eq:alpha-j-estimating-equation} by some \(p^k(x) = p^k(x) + c_k\) with minimal modification to the model assisted estimator. The constants  \(c_k \in \SR\) can be allowed to grow with \(k\) so long as we replace \(\xi_{k,\infty}\) with the maximum of \(\tilde\xi_{k,\infty} := \sup_{x\in\calX}\|\tilde p^k(x)\|_\infty\) and \(\xi_{k,\infty}\) in the sparsity bounds of \Cref{sec:first-stage}. Without loss of generality we will assume that the first basis term is a constant so that \(p_1(x) \equiv 1\)

After estimating the models \((\hat\pi_1,\hat m_1),\dots,(\hat \pi_k, \hat m_k)\) using \((\tilde p_1(x),\dots,\tilde p_k(x))\) in place of \((p_1(x),\dots,p_k(x))\) in \eqref{eq:gamma-j-estimating-equation}-\eqref{eq:alpha-j-estimating-equation} we would construct the second stage estimate \(\hat\beta^k\) 
\[
    \tilde\beta^k = \widehat{Q}^{-1}
    \E_n\begin{bmatrix} \tilde p_1(x)Y(\hat\pi_1,\hat m_1) - c_kY(\hat\pi_1,\hat m_1)\\
    \tilde p_2(x)Y(\hat\pi_2,\hat m_2) - c_kY(\hat\pi_1,\hat m_1) \\
    \vdots \\ \tilde p_k(x)Y(\hat\pi_k, \hat m_k) - c_kY(\hat \pi_1, \hat m_1)\end{bmatrix} 
.\] 
Via the same analysis of \Cref{sec:theory-overview,sec:first-stage} we will still be able to show that the bias passed on from first stage estimation to the second stage parameter \(\tilde\beta^k\) remains negligible even under misspecification of either first stage model. This is because \Cref{thm:first-stage-convergence} will establish that
\begin{equation*}
    \begin{split}
        \max_{1\leq j \leq k}|\E_n[\tilde p_j(x)Y(\hat\pi_j,\hat m_j)] - \E_n[\tilde p_j(x)Y(\bar\pi_j,\bar m_j)]| 
        &= o_p(n^{-1/2}k^{-1/2})\andbox\\
        \max_{1\leq j \leq k}\tilde\xi_{k,\infty}\max_{1\leq j \leq k} \E_n[\tilde p_j(x)^2(Y(\widehat\pi_j,\widehat m_j) - Y(\bar\pi_j,\bar m_j))^2] 
        &= o_p(k^{-2}n^{-1/m})
    \end{split}.
\end{equation*}
Using the first statement, we can immediately establish via the triangle inequality that
\begin{align*}
    \max_{1 \leq j\leq k} |\E_n[\tilde p_j(x)Y(\hat\pi_j,\hat m_j) - c_kY(\hat\pi_1,\hat m_1)]-\E_n[\tilde p_j(x)Y(\bar\pi_j,\bar m_j) - c_kY(\bar \pi_1,\bar m_1)]| = o_p(n^{-1/2}k^{-1/2})
\end{align*}
which is the exact analog of \Cref{cond:no-effect} needed to establish consistency at the nonparameteric rate of the modified model assisted estimator. Similarly, using the second statement and \((a+b)^2 \leq 2a^2 + 2b^2\) we can immediately establish that 
\begin{align*}
    \max_{1\leq j \leq k}\E_n[(\tilde p_j(x)Y(\hat\pi_j,\hat m_j) - cY(\hat\pi_1,\hat m_1) - \tilde p_j(x)Y(\bar\pi_j, \bar m_j) + cY(\bar\pi_j,\bar m_j))^2] = o_p(k^{-2}n^{-1/m})
\end{align*}
which is the exact analog of \Cref{cond:variance-estimation} needed to establish a consistent variance estimator when \(\tilde\beta^k\) is used instead of the \(\hat\beta^k\) from  \eqref{eq:beta-k}.

This logic can be extended slightly if the researcher would like to weight the estimating equations \eqref{eq:gamma-j-estimating-equation}-\eqref{eq:alpha-j-estimating-equation} by some \(\tilde p^k(x) = G^kp^k(x)\) for an invertible and bounded sequence of linear operators \(G^k:\SR^k \to \SR^k\). In this case, one would again use \(\tilde p^k(x)\) in place of \(p^k(x)\) in \eqref{eq:gamma-j-estimating-equation}-\eqref{eq:alpha-j-estimating-equation} and construct the second stage coeffecients via
\begin{align*}
    \tilde\beta^k := \widehat{Q}^{-1} G^{k,-1}
    \E_n\begin{bmatrix} \tilde p_1(x)Y(\hat\pi_1,\hat m_1) \\ \vdots \\ \tilde p_k(x)Y(\hat\pi_k,\hat m_k) \end{bmatrix} 
\end{align*}
After constructing the second stage estimator using \(\tilde\beta^k\), inference procedures would proceed normally as described in \Cref{sec:setup}.

\section{Alternative CV-Type Method for Penalty Parameter Selection}
\label{sec:alt-penalty-CV}

In this section we consider a procedure for penalty parameter selection where we use the pilot penalty parameters described in \eqref{eq:pilot-penalty} directly, after choosing constants \(c_{\gamma, j}\) and  \(c_{\alpha, j}\) from a (finite) set via cross validation. For each \(j\) we will assume that 
\begin{equation}
    \label{eq:alt-penalty-CV-set}
    c_{\gamma, j}, c_{\alpha, j} \in \Lambda_n \subseteq [\underline c_n, \bar c_n]
\end{equation} 
where \(|\Lambda_n|\) can be fairly large (on the order of \(n^2/k\)). 

\subsection{Theory Overview}%
Let \(M_5,M_6,M_7, M_8^2, M_9^2\) be constants that do not depend on \(k\) as in \Cref{lemma:omega-5,lemma:omega-6,lemma:omega-7,lemma:logit-score-domination-deterministic,lemma:linear-score-domination-deterministic}. Whenever
\begin{equation}
    \label{eq:alt-penalty-parameter-determinstic-score-domination}
    \underline c_n\sqrt{\frac{\ln^3(d_zn)}{n}} \geq \xi_{k,\infty}\max\left\{M_5, M_6, M_7, M_8^2, M_9^2\right\}\sqrt{\frac{\ln(d_zn)}{n}}
.\end{equation} 
we will have that, under \Cref{assm:logistic-model-convergence}(i)-(iv) the event \(\bigcap_{k=1}^7 \Omega_{k,7}\) occurs with probability at least \(1 - 10k/n^2\) for the \(2k\) pilot penalty parameters chosen with any values \(c_{\gamma, j}, c_{\alpha,j} \in\Lambda_n\) and 
\[
    \bar\lambda_k := \bar c_n \sqrt{\frac{\ln^3(d_zn)}{n}}
.\]
In this event, apply \Cref{lemma:nonasymptotic-logit,lemma:nonasymptotic-outcome} to obtain the following finte sample bounds for the parameter estimates
\begin{align*}
    \max_{1\leq j \leq k}D_{\gamma,j}^\ddag(\widehat\gamma_j,\bar\gamma_j) &\leq M_0\frac{s_k\bar c_n^2 \ln^3(d_zn)}{n} \andbox \max_{1\leq j\leq k}\|\widehat\gamma_j-\bar\gamma_j\|_1 \leq M_0s_k \bar c_n\sqrt{\frac{\ln^3(d_zn)}{n} } \\
    \max_{1\leq j \leq k}D_{\alpha, j}^\ddag(\widehat\alpha_j, \bar\alpha_j;\bar\gamma_j) &\leq M_1\frac{s_k\bar c_n^2\ln^3(d_zn)}{n} \andbox \max_{1\leq j\leq l}\|\widehat\alpha_j - \bar\alpha_j\|_1 \leq M_1s_k\bar c_n \sqrt{\frac{\ln^3(d_zn)}{n}} \\
\end{align*}
and \Cref{lemma:nonasymptotic-means} to obtain the following finite sample bound for the weighted means:
\begin{align}
    \label{eq:alt-penalty-means-bound}
    \max_{1\leq j\leq k} |\E_n[p_j(X)(Y(\widehat\pi_j,\widehat m_j) - Y(\bar\pi_j, \bar m_j))]| 
    &\leq  M_2\frac{\bar c_n^2s_k\ln^3(d_zn)}{n} \\
    \label{eq:alt-penalty-second-moments-bound}
    \max_{1\leq j\leq k}|\E_n[p_j^2(X)(Y(\widehat\pi_j, \widehat m_j) - Y(\bar\pi_j, \bar m_j))^2] 
    &\leq M_3\frac{\xi_{k,\infty}^2\bar c_n^2 s_k^2\ln^3(d_zn)}{n} 
\end{align}
Combining \eqref{eq:alt-penalty-parameter-determinstic-score-domination} and \eqref{eq:alt-penalty-means-bound} we can see that \Cref{cond:no-effect} can be obtained under  \Cref{assm:logistic-model-convergence}(i)-(iv) and the following modified sparsity bounds
\begin{equation}
   \label{eq:alt-penalty-cond1-sparsity-bounds}
   \frac{k|\Lambda_n|}{n^2}\to 0,\; \frac{\underline c_n^{-1}\xi_{k,\infty}}{\ln(d_zn)} \to 0\andbox \frac{\bar c_n^2s_k k^{1/2}\ln^3(d_zn)}{\sqrt{n}}\to 0 
.\end{equation} 
Simlarly combining \eqref{eq:alt-penalty-parameter-determinstic-score-domination} and \eqref{eq:alt-penalty-second-moments-bound}, \Cref{cond:variance-estimation} can additionally be obtained by strengthening the rates in \eqref{eq:alt-penalty-cond1-sparsity-bounds} to include
\begin{equation}
    \label{eq:alt-penalty-cond2-sparsity-bounds}
    \frac{\xi_{k,\infty}^2\bar c_n^2s_kk^2\ln^3(d_zn)}{n^{(m-1)/m}} \to 0
\end{equation}
for \(m > 2\) as in \Cref{assm:uniform-limit-theory}. These rates are comparable and in certain cases may be more palatable than those presented in the main text, \Cref{assm:logistic-model-convergence}(vi). They come at the cost of slower rates of convergence for the weighted means as seen by comparing \crefrange{eq:alt-penalty-means-bound}{eq:alt-penalty-second-moments-bound} to \cref{eq:uniform-mean-convergence,eq:second-moment-convergence}.

\subsection{Practical Implementation}%

In practice, the constants \(M_5, M_6, M_7, M_8^2, M_9^2\) from \Cref{lemma:omega-5,lemma:omega-6,lemma:omega-7,lemma:logit-score-domination-deterministic,lemma:linear-score-domination-deterministic} are roughly on the order of \(\|Z\|_\infty\). We therefore reccomend setting 
\begin{align*}
    \underline c_n &= \frac{1}{2\log^{1/2}(d_zn)} \max_{1\leq i\leq n}\|p^k(X_i)\|_\infty \max_{1\leq i \leq n}\|Z_i\|_\infty \\
    \bar c_n &= \frac{3\log^{1/2}(d_zn)}{2} \max_{1\leq i \leq n}\|p^k(X_i)\|_\infty\max_{1\leq i\leq n}\|Z_i\|_\infty
\end{align*} 
and letting \(\Lambda_n\) be a set of points evenly spaced between \(\underline c_n\) and  \(\bar c_n\). The cross validation procedure then can be implemented in the following steps.
\begin{enumerate}
    \item[\texttt{1.}] \texttt{Split the sample into \(K\) folds.}
    \item[\texttt{2.}] \texttt{Consider a single pair of values for \(c_\alpha, c_\gamma\) and designate a fold to hold out.}
    \item[\texttt{3.}] \texttt{Estimate nuisance model parameters using \(\lambda_{\gamma, j}^{\text{\tiny pilot}}\) and \(\lambda_{\alpha, j}^{\text{\tiny pilot}}\) on the remaining  folds.}
    \item[\texttt{4.}] \texttt{Evaluate the resulting models on held out fold using non-penalized loss functions.}
    \item[\texttt{5.}] \texttt{Repeat \(K\) times and record average loss over all folds.}
    \item[\texttt{6.}] \texttt{Choose values of \(c_{\gamma, j}\) and  \(c_{\alpha, j}\) with the lowest average loss.} 
\end{enumerate}
In practice we find this procedure works well with small \(K\), around \(K = 5\) and with  \(|\Lambda_n|\) on the order of about 10-20.

\section{High Dimensional Probability Results}%
\label{sec:high-dimensional-clt}

\subsection{High Dimensional Central Limit and Bootstrap Theorems}%
\label{sec:clt-bootstrap}

\begin{lemma}[Gaussian Quantile Bound]
	\label{lemma:quantile-bound}
	Let \(Y=(Y_1,\dots,Y_p)\) be centered Gaussian in \(\SR^p\) with \(\sigma^2 \leq \max_{1\leq j\leq p}\E[Y_j^2]\) and \(\rho \geq 2\). Let \(q^Y(1-\eps)\) denote the  \((1-\eps)\)-quantile of  \(\|Y\|_\infty\) for \(\eps \in (0,1)\). Then  \(q^Y(1-\eps) \leq (2+\sqrt{2})\sigma\sqrt{\ln(p/\eps)}\).
\end{lemma}
\begin{proof}
    See \citet{CS-2021}, Lemma D.2.
\end{proof}

Now let \(Z_1,\dots,Z_n\) be independent, mean zero random variables in \(\SR^p\), and denote their scaled average and variance by 
\[
	S_n := \frac{1}{\sqrt{n}}\sum_{i=1}^n Z_i \andbox \Sigma :=  \frac{1}{n} \sum_{i=1}^n \E[Z_iZ_i'] 
.\]
For \(\SR^p\) values random variables \(U\) and  \(V\), define the distributional measure of distance
 \[
	 \rho(U,V) := \sup_{A\in\calA_p} \left|\Pr(U\in A) - \Pr(V\in A)\right|
\] 
where \(\calA_p\) denotes the collection of all hyperrectangles in  \(\SR^p\). For any symmetric positive matrix \(M \in \SR^{p\times p}\), write \(N_M := N(\bm{0}, M)\).

\begin{theorem}[High-Dimensional CLT]
	\label{thm:clt}
	If, for some finite constants \(b > 0\) and  \(B_n \geq 1\),
	\begin{equation}
		\label{eq:clt-conditions}
		\frac{1}{n} \sum_{i=1}^n \E[Z_{ij}^2] \geq b,\hbox{ }\frac{1}{n} \sum_{i=1}^n \E[|Z_{ij}|^{2+k}] \leq B_n^k\andbox \E\left[\max_{1\leq j\leq p} Z_{ij}^4\right] \leq B_n^4
	.\end{equation} 
	for all \(i \in \{1,\dots,n\}, j\in \{1,\dots,p\}\) and \(k\in \{1,2\} \), then there exists a finite constant \(C_b\), depending only on \(b\), such that:
	\[
		\rho(S_n, N_{\Sigma}) \leq C_b\left(\frac{B_n^4\ln^7(pn)}{n}\right)^{1/6}
	.\] 
\end{theorem}
\begin{proof}
    See \citet{CCK-2018-HDCLT}, Proposition 2.1.
\end{proof}

Let \(\widehat Z_i\) be an estimator of \(Z_i\) and let  \(e_1,\dots,e_n\) be i.i.d \(N(0,1)\) and independent of both the  \(Z_i\)'s and  \( \widehat Z_i\)'s. Define 
\(
    \widehat S_n^e := \frac{1}{\sqrt n}\sum_{i=1}^n e_i \widehat Z_i 
\) 
and let \(\Pr_e\) denote the conditional probability measure computed with respect to the  \(e_i's\) for fixed  \(Z_i\)'s and  \(\widehat Z_i\)'s. Also abbreviate 
\[
	\tilde \rho(\widehat S_n^e, N_{\Sigma}) := \sup_{A\in\calA_p}\left|\text{Pr}_e\left(\widehat S_n^e \in A\right) - \Pr\left(N_\Sigma\in A\right)\right|
.\] 

\begin{theorem}[Multiplier Bootstrap for Many Approximate Means]
	\label{thm:bootstrap}
	Let \eqref{eq:clt-conditions} hold for some finite constants \(b > 0\) and \(B_n \geq 1\), and let \(\{\beta_n\}_{\SN}\) and  \(\{\delta_n\}_{\SN}\) be sequences in \(\SR_{++}\) converging to zero such that 
	\begin{equation}
		\label{eq:bootstrap-condition}
		\Pr\left(\max_{1\leq j\leq p}\frac{1}{n} \sum_{i=1}^n (\widehat Z_{ij} - Z_{ij})^2 > \frac{\delta_n^2}{\ln^2(pn)} \right) \leq \beta_n
	\end{equation}
	Then, there exists a finite constant \(C_b\) depending only on  \(b\) such that with probability at least  \(1-\beta_n - 1/\ln^2(pn)\), 
	 \[
		 \tilde \rho(\widehat S_n^e, N_\Sigma) \leq C_b\max\left\{\delta_n, \left(\frac{B_n\ln^6(pn)}{n} \right)^{1/6}\right\}
	.\] 
\end{theorem}
\begin{proof}
    See \citet{belloni2018highdimensional}, Theorem 2.2 or \citet{CS-2021} Theorem D.2.
\end{proof}

We now consider a partition of \(Z\) and  \( \widehat Z\) into \(k\) subvectors.
\begin{align*}
    Z := (Z_1',\dots,Z_k')'\in\SR^{d_1,\dots,d_k}\andbox \widehat Z := (\widehat Z_1',\dots,Z_k')' \in \SR^{d_1,\dots,d_k}
\end{align*}
where \(\sum_{j=1}^k d_j = p\).
Given such a partition, for any symmetric, positive definite \(M \in \SR^{p\times p}\) let \(N_{M,j}\) denote the marginal distribution of the subvector of \(N_M\) corresponding the the indices of partition  \(j\). In other words, \(N_{M_1}\) would denote the marginal distribution of the first  \(d_1\) elements of an  \(\SR^p\) vector with distribution \(N_M\), \(N_2\) would denote the marginal distribution of the next \(d_2\) elements and so on. For each \(j=1,\dots,k\) define \(q_{M,j}^N: \SR \to \bar \SR\) as the (extended) quantile function of \(\|N_{M,j}\|_{\infty}\),
\[
	q_{M,j}^N(\eps) := \inf\left\{t \in \SR: \Pr(\|N_{M,j}\|_{\infty} \leq t) \geq \eps\right\}
.\] 
Define \(q_{M,j}^N(\eps) = +\infty\) if \(\eps \geq 1\) and \(-\infty\) if  \(\eps\leq 0\) so that \(q_{M,j}^N\) is always montone (strictly) increasing.

\begin{lemma}[]
	\label{lemma:quantile-setup}
    Let \(M\in\SR^{p\times p}\) be symmetric positive definite, let \(U\) be a random variable in  \( \SR^p\). Partition \(U\) into  \(k\) subvectors,  \(U=(U_1',\dots,U_k')' \in \SR^{d_1,\dots,d_k}\) where \(d_1+\dots+d_k =p\). For each \(j =1,..,k\) let \(q_j\) denote the quantile function of  \(\|U_j\|_{\infty}\). Then for any \(j=1,\dots,k\),
	\[
		q_{M,j}^N(\eps - 2\rho(U,N_M)) \leq q_j(\eps) \leq q_{M,j}^N(\eps + \rho(U,N_M))\hbox{ }\text{ for all }\eps\in(0,1) 
	.\] 
\end{lemma}
\begin{proof}
	Proof is a slight modification of that of Lemma D.3 in \citet{CS-2021}. Main idea is to add and substract a \(\|N_M\|_\infty\) term and use the fact that the approximation is achieved over all hyperrectangles. We show the bound holds for each \(j=1,\dots,k\). Without loss of generality, consider \(U_1\). Let \(N_{M,1}\) denote the maginal distribution of the first  \(d_1\) elements of a \(\SR^p\) vector with distribution \(N_M\).
	\begin{align*}
		\Pr(\|U_1\|_\infty \leq t) &= \Pr(\|N_{M,1}\|_\infty \leq t) + \Pr(\|U_1\|_{\infty} \leq t) - \Pr(\|N_{M,1}\|_\infty \leq t) \\
								 &= \Pr(\|N_{M,1}\|_\infty \leq t) + \left(\Pr(U \in [-t,t]^p \times \SR^{p-d_1}) - \Pr(N_M \in [-t,t]^p \times \SR^{p-d_1})\right) \\
								 &\leq \Pr(\|N_{M,1}\|_{\infty} \leq t) + \rho(U,N_M)
	\end{align*}
	for any \(t \in\SR\). A similar construction will give that 
	\[
		\Pr(\|U_1\|_{\infty} \leq t) \geq \Pr(\|N_{M,1}\|_\infty \leq t) - \rho(U, N_M)
	.\] 

	Substituting \(t = q_{M,1}^N(\eps - 2\rho(U,N_M))\) into the upper bound on \(\Pr(\|U_1\|_\infty \leq t)\) gives the lower bound statement, while \(t = q_{M,1}^N(\eps + \rho(U, N_M))\) and using the lower bound on \(\Pr(\|U_1\|_\infty \leq t)\) gives the upper bound statement.
\end{proof}

As with \(Z\) partition \(S_n\) and  \(\widehat S_n^e\) into  
\[
    S_n = (S_{n,1}',\dots,S_{n,k}')' \in \SR^{d_1,\dots,d_k}\andbox \widehat S_n^e = (\widehat S_{n,1}^{e'},\dots,\widehat S_{n,k}^{e'})'\in\SR^{d_1,\dots,d_k}
.\] 
For each \(j=1,\dots,k\) define \(q_{n,j}(\eps)\) as the  \(\eps\)-quantile of  \(\|S_{n,j}\|_\infty\)
\[
	q_{n,j}(\eps) := \inf\{t \in \SR: \Pr(\|S_{n,j}\|_\infty \leq t) \geq \eps\}\hbox{ }\text{ for }\eps \in (0,1)
.\] 
Let \( \widehat q_{n,j}(\eps)\) be the \(\eps\)-quantile of  \( \|\widehat S_{n,j}^e\|_\infty\), computed conditionally on \(X_i\) and  \( \widehat X_i\)'s,
\[
	\widehat q_{n,j}(\eps) := \inf\{t \in \SR: \text{Pr}_e(\|\widehat S_{n,j}^e\|_\infty \leq t) \geq \eps\} \text{ for }\eps \in (0,1)
.\] 

\begin{theorem}[Quantile Comparasion]
	\label{thm:quantile-comp}
	If \eqref{eq:clt-conditions} holds for some finite constants \(b> 0\) and  \(B_n \geq 1\), and 
	\[
		\rho_n := 2C_b\left(\frac{B_n^4\ln^7(pn)}{n} \right)^{1/6}
	\] 
	denotes the upper bound in \Cref{thm:clt} multiplied by two, then for all \(j =1,\dots,k\)
	\[
		q_{\Sigma,j}^N(1- \eps - \rho_n) \leq q_{n,j}(1-\eps) \leq  q_{\Sigma,j}^N(1-\eps + \rho_n)\hbox{ }\text{ for all }\eps \in (0,1)
	.\] 
	If, in addition, \eqref{eq:bootstrap-condition} holds for some sequences \(\{\delta_n\}_\SN\) and \(\{\beta_n\}_\SN\) converging to zero, and 
	\[
		\rho_n' \leq  2C_b'\max\left\{\delta, \left(\frac{B_n^4\ln^6(pn)}{n} \right)^{1/6}\right\} 
	\] 
	denotes the upper bound in \Cref{thm:bootstrap} multiplied by two, then with probability at least  \(1-\beta_n - 1/\ln^2(pn)\) we have for all \(j=1,\dots,k\), 
	\[
		q_{\Sigma,j}^N(1-\eps-\rho_n') \leq \widehat q_{n,j}(1-\eps) \leq q_{\Sigma,j}^N(1-\eps + \rho_n')\hbox{ }\text{ for all }\eps \in (0,1)
	.\] 
\end{theorem}
\begin{proof}
	From \Cref{lemma:quantile-setup} with \(U = S_n\) we obtain
	 \[
		 q_{\Sigma,j}^N(1-\eps-2\rho(S_n,N_\Sigma)) \leq q_{n,j}(1-\eps) \leq q_{\Sigma,j}^N(1-\eps + \rho(S_n, N_\Sigma))
	.\] 
	The first chain of inequalities then follows from \(2\rho(S_n, N_\Sigma) \leq \rho_n\) by \Cref{thm:clt}. 

	For the second claim, apply \Cref{lemma:quantile-setup} with \(U = \widehat S_n^e\) and condition on the \(Z_i\)'s and  \( \widehat Z_i\)'s obtain
	\[
		q_{\Sigma,j}^N(1-\eps - 2\tilde\rho(\widehat S_n^e, N_\Sigma)) \leq \widehat q_n(1-\eps) \leq q_{\Sigma,j}^N(1-\eps + \tilde\rho(\widehat S_n^e, N_\Sigma))
	.\] 
	The second chain of inequalities then follows on the event \(2\tilde\rho(\widehat S_n^e, N_\Sigma) \leq \rho_n'\), which by \Cref{thm:bootstrap} happens with probability at least \(1-\beta_n - 1/\ln^2(pn)\).
\end{proof}

\begin{theorem}[Multiplier Bootstrap Consistency]
	\label{thm:bootstrap-consistency}
	Let \eqref{eq:clt-conditions} and \eqref{eq:bootstrap-condition} hold for some constants \(b>0\) and  \(B_n \geq 1\) and some sequences \(\{\delta_n\}_\SN\) and \(\{\beta_n\}_\SN\) in \(\SR_{++}\) converging to zero. Then, there exists a finite constant  \(C_b\), depending only on  \(b\), such that 
	 \[
		 \max_{1\leq j\leq k}\sup_{\eps\in (0,1)} \left|\Pr(\|S_{n,j}\|_\infty \geq \widehat q_{n,j}(1-\alpha)) - \alpha\right| \leq C_b\max\left\{\beta_n, \delta_n, \left(\frac{B_n^4\ln^7(pn)}{n}\right)^{1/6}, \frac{1}{\ln^2(pn)} \right\}
	.\] 
\end{theorem}
\begin{proof}
	By \Cref{thm:clt} and \Cref{thm:quantile-comp}, 
	\begin{align*}
		\Pr(\|S_{n,j}\|_\infty \leq \widehat q_{n,j}(1-\eps)) 
		&\leq \Pr(\|S_{n,j}\|_{\infty} \leq q_{\Sigma,j}^N(1-\eps+\rho_n')) + \beta_n + \frac{1}{\ln^2(pn)} \\
		&\leq \Pr(\|N_{\Sigma,j}\|_\infty \leq q_{\Sigma,j}^N(1-\eps + \rho_n')) + \rho_n + \beta_n + \frac{1}{\ln^2(pn)}  \\
		&\leq 1-\eps + \rho_n' + \rho_n + \beta_n + \frac{1}{\ln^2(pn)} 
	\end{align*}
	Where the second inequality is making use of the same rectangle argument as before. A parallel argument shows that 
	\[
		\Pr(\|S_{n,j}\|_{\infty} \leq \widehat q_{n,j}(1-\eps)) \geq 1-\eps - \left(\rho_n'+\rho_n + \beta_n + \frac{1}{\ln^2(pn)} \right)
	.\] 
	Combining these two inequalities gives the result.

\end{proof}

\subsection{Concentration and Tail Bounds}
\label{sec:tail-concentration-bounds}

We make use of the following concentration and tail bounds. \Crefrange{lemma:tan-14}{lemma:tan-18} can be found in \citet{Bulmann-VanDeGeer-2011}. The proof of \Cref{lemma:sub-gaussian-absolute-bound} is trivial but provided here.

\begin{lemma}[]
	\label{lemma:tan-14}
	Let \((Y_1,\dots,Y_n)\) be independent random variables such that \(\E[Y_i]=0\) for \(i=1,\dots,n\) and \(\max_{i=1,\dots,m}|Y_i|\leq c_0\) for some constant \(c_0\). Then, for any \(t>0\),
	\[
		\Pr\bigg(\bigg|\frac{1}{n}\sum_{i=1}^n Y_i\bigg|>t\bigg) \leq 2\exp\left(-\frac{nt^2}{2c_0^2}\right)
	.\] 
\end{lemma}

\begin{lemma}[]
	\label{lemma:tan-15}
	Let \((Y_1,\dots,Y_n)\) be independent random variables such that \(\E[Y_i]=0\) for \(i=1,\dots,\) and \((Y_1,\dots,Y_n)\) are uniformly sub-gaussian: \(\max_{1\leq i\leq n}c_1^2\E[\exp(Y_i^2/c_1^2)-1]\leq c_2^2\) for some constants \((c_1,c_2)\). Then for any \(t>0\),
	 \[
		 \Pr\bigg(\bigg|\frac{1}{n}\sum_{i=1}^n Y_i\bigg|>t\bigg) \leq 2\exp\left(-\frac{nt^2}{8(c_1^2+c_2^2)}\right)
	.\] 
\end{lemma}
\begin{lemma}[]
	\label{lemma:tan-16}
	Let \((Y_1,\dots,Y_n)\) be independent variables such that \(\E[Y_i]=0\) for \(i=1,\dots,n\) and 
	\[
		\frac{1}{n}\sum_{i=1}^n \E[|Y_i|^k] \leq \frac{k!}{2}c_3^{k-2}c_4^2,\hbox{ }\hbox{ }k=2,3,\dots,
	\]
	for some constants \((c_3,c_4)\). Then, for any \(t>0\),
	\[
		\Pr\bigg(\bigg|\frac{1}{n}\sum_{i=1}^n Y_i\bigg|>c_3t+c_4\sqrt{2t}\bigg)\leq 2\exp(-nt)
	.\] 
\end{lemma}
\begin{lemma}[]
	\label{lemma:subgaussian-expectation-bound}
	Suppose that \(Y\) is sub-gaussian: \(c_1^2\E[\exp(Y^2/c_1^2)-1]\leq c_2^2\) for some constants \((c_1,c_2)\). Then
	\[
		\E[|Y|^k] \leq \Gamma\left(\frac{k}{2}+1\right)(c_1^2+c_2^2)c_1^{k-2},\hbox{ }\hbox{ }k=2,3,\dots
	.\] 
\end{lemma}
\begin{lemma}[]
	\label{lemma:tan-18}
	Suppose that \(X\) is bounded, \(|X|\leq c_0\),  and \(Y\) is sub-gaussian, \(c_2^2\E[\exp(Y^2/c_1^2)-1]\leq c_2^2\) for some constants \((c_1,c_2)\). Then \(Z=XY^2\) satisfies 
	\[
		\E\left[|Z-\E[Z]|^k\right] \leq \frac{k!}{2}c_3^{k-2}c_4^2,\hbox{ }\hbox{ }k=2,3,\dots 
	,\] 
	for \(c_3 = 2c_0c_1^2\) and  \(c_4 = 2c_0c_1c_2\).
\end{lemma}

\begin{lemma}[]
	\label{lemma:sub-gaussian-absolute-bound}
	Suppose that \(Y\) is sub-gaussian in the following sense, there exist positive constants \(c_0,c_1 > 0\) such that \(c_0^2\E[\exp(Y^2/c_0^2)-1] \leq c_1^2\). Then 
	\[
	    \E[|Y|] \leq c_1^2/c_0 + c_0
	.\] 
\end{lemma}
\begin{proof}
	Using the fact that \(e^{x^2} > |x|\) gives 
	 \begin{align*}
		 c_0^2\E[\exp(Y^2/c_0^2)-1] \leq c_1^2 
		 &\implies \E[\exp(Y^2/c_0^2)] \leq c_1^2/c_0^2 + 1 \\
		 &\implies \E[|Y/c_0|] \leq c_1^2/c_0^2 + 1 \\
		 &\implies \E[|Y|] \leq c_1^2/c_0 + c_0 
	\end{align*}
\end{proof}

\end{document}